\documentclass{article}

\usepackage[preprint]{neurips_2026}                 

\usepackage[utf8]{inputenc}
\usepackage[T1]{fontenc}
\usepackage{hyperref}
\usepackage{url}
\usepackage{physics}
\usepackage{booktabs}
\usepackage{amsfonts}
\usepackage{amsmath,amssymb,amsthm,mathtools}
\usepackage{nicefrac}
\usepackage{microtype}
\usepackage{xcolor}
\usepackage{bm}
\usepackage{bbm}
\usepackage{enumitem}
\usepackage[ruled,vlined,linesnumbered]{algorithm2e}
\usepackage{amssymb}
\usepackage{float}
\usepackage{makecell}

\newtheorem{theorem}{Theorem}[section]
\newtheorem{proposition}[theorem]{Proposition}
\newtheorem{lemma}[theorem]{Lemma}
\newtheorem{corollary}[theorem]{Corollary}
\newtheorem{fact}[theorem]{Fact}
\newtheorem{problem}[theorem]{Problem}
\newtheorem{definition}[theorem]{Definition}

\DeclareMathOperator{\sign}{sign}
\DeclareMathOperator{\ad}{ad}
\DeclareMathOperator{\HS}{HS}
\DeclareMathOperator{\TV}{TV}
\DeclareMathOperator{\Var}{Var}
\newcommand{\F}{\mathbb{F}_2}
\newcommand{\E}{\mathbb{E}}

\newcommand{\wt}{\widetilde}
\newcommand{\wh}{\widehat}
\newcommand{\cU}{\mathcal{U}}

\newcommand{\cK}{\mathcal{K}}
\newcommand{\cS}{\mathcal{S}}
\newcommand{\cA}{\mathcal{A}}

\renewcommand{\epsilon}{\varepsilon}

\SetKwInput{KwIn}{Input}
\SetKwInput{KwOut}{Output}
\DontPrintSemicolon

\title{Optimal Ansatz-free Hamiltonian Learning In Situ}

\author{
  Taiqi Zhou\thanks{Department of Information Engineering, The Chinese University of Hong Kong. Email: \href{mailto:1155242363@link.cuhk.edu.hk}{1155242363@link.cuhk.edu.hk}}\\
  \And
  Weiyuan Gong\thanks{Corresponding author. John A. Paulson School of Engineering and Applied Sciences, Harvard University. Email:\href{mailto:wgong@g.harvard.edu}{wgong@g.harvard.edu}. A part of this work was conducted while visiting California Institute of Technology.}\\
}

\begin{document}
\maketitle

\begin{abstract}
Characterizing the features of a Hamiltonian that governs a quantum system serves as a fundamental subroutine of quantum device calibration, signal sensing, and error correction. 
Recent works have proposed protocols achieving the optimal Heisenberg-limited scaling learning ansatz-free Hamiltonians from their real-time evolutions without fully specifying interaction structures.
However, these protocols rely on both deep circuits with interleaving probes and control, and extremely short time resolution, making them difficult to implement on near- and intermediate-term in situ quantum experiments.
In this work, we propose a computationally efficient, \emph{control-free}, and \emph{ancilla-free} algorithm that uses only \emph{Pauli product state preparation and measurement}, and learns an ansatz-free Hamiltonian $H$ with $\norm{H}\leq\Lambda$ in total evolution time of $\Theta\left(\tfrac{\Lambda}{\epsilon^2}\log\left(\tfrac{\Lambda}{\epsilon}\right)\right)$.
The evolution time cost of our algorithm is \emph{optimal} for any control-free protocols as we further prove a lower bound of $\Omega\left(\tfrac{\Lambda}{\epsilon^2}\log\left(\tfrac{\Lambda}{\epsilon}\right)\right)$.
Technically, our method introduces a randomized-sampling framework that combines band-limited kernel-based time sampling with a displacement sieve for Hamiltonian structure learning.
The characteristic probe time resolution depends only on $\Lambda$ instead of $\varepsilon$, which makes our protocol especially appealing in the high-precision regime for sensing and calibration applications. 
We also analyze the robustness of the algorithm against state-preparation-and-measurement (SPAM) errors, where we show that, by calibrating preparation and measurement errors, the algorithm maintains the same asymptotic total evolution time in the presence of SPAM noise when the Hamiltonian is local.
Our results demonstrate the fundamental cost of experimentally friendly Hamiltonian learning and provide a practical route to rigorous in situ characterization of near-term quantum platforms.
\end{abstract}

\section{Introduction}\label{sec:intro}
Understanding the interactions governing a quantum system is a central task in quantum information science. 
As programmable analog simulators, NISQ devices, and early fault-tolerant quantum processors continue to mature, the need for rigorous and scalable characterization protocols has become increasingly urgent for validating engineered quantum dynamics and for practical tasks such as calibration, sensing, verification, error mitigation, and error correction~\cite{eisert2020quantum,carrasco2021theoretical}. 
In closed systems, interactions of quantum systems are encoded by the Hamiltonians.
A standard key subroutine for characterization is Hamiltonian learning, which infers the underlying $H$ from accessible data. 
At a high level, one may access $H$ either through equilibrium information~\cite{anshu2021sample,haah2022optimal,bakshi2024learning,gu2024practical,qi2019determining,li2020hamiltonian,evans2019scalable,chen2025quantum,chen2025learning}, or through real-time evolutions $e^{-iHt}$. 
The latter viewpoint, which is also the access model we focus on in this work, is especially appealing for in situ characterization.
Hamiltonian learning from time-evolution is deeply connected to quantum metrology and sensing \cite{giovannetti2006quantum,giovannetti2011advances,degen2017quantum}, which aims to detect tiny local signal perturbations to the Hamiltonian.
The early works then extend the sensing problem to a global parameter estimation task, focusing on highly structured models with heuristic guarantees, sub-optimal costs, and experimentally motivated protocols~\cite{granade2012robust,wiebe2014hamiltonian,wang2015hamiltonian,holzapfel2015scalable,che2021learning,kokail2021entanglement,wilde2022scalably,hangleiter2024robustly,guo2025hamiltonian}.

For Hamiltonians with \emph{known} and \emph{structured} interaction families, Huang et al.~\cite{huang2023learning} proposed a protocol with total evolution time achieving the optimal Heisenberg-limited scaling.
Several similar Heisenberg-limited scaling algorithms are also obtained for local Hamiltonians with unknown supports~\cite{ma2024learning,bakshi2024structure}, interacting bosons and fermionic systems~\cite{li2024heisenberg,ni2024quantum,mirani2024learning}, other structured many-body systems~\cite{yu2023robust,stilck2024efficient,francca2025learning}, and testing structured Hamiltonians ~\cite{kallaugher2025hamiltonian,bluhm2026certifying,gao2026quantum,arunachalam2025testing,bluhm2026hamiltonian,sinha2025improved}. 

However, in the \emph{ansatz-free} setting, the interaction structure is not specified. 
The Hamiltonian may be a sparse but arbitrary linear combination of potentially nonlocal Pauli strings. 
In this setting, Zhao~\cite{zhao2025learning} obtained nearly Heisenberg-limited scaling assuming queries to reversal time evolutions, ancillary qubits, and deep circuits with interleaving queries to time evolutions and unitary control gates.
Later, Hu et al~\cite{hu2025ansatz} proposed a protocol achieving Heisenberg-limited scaling of total evolution time without ancillary qubits and reversal time evolutions.
More recently, for ansatz-free Hamiltonian supported by $M$ potentially nonlocal Pauli strings, Refs.~\cite{sinha2025improved,abbas2025nearly} obtained the state-of-the-art total evolution time of $\widetilde{O}(M/\epsilon)$.
However, reliance on ingredients such as ancillary entanglements, interleaved quantum controls during time evolutions, engineered inverse dynamics, and extremely short resolved evolution times makes these protocols difficult to realize in realistic \emph{in situ} settings.
On near-term quantum hardware, especially analog hardware, one may not be able to pause the native dynamics and insert coherent control.
Ancillary entanglement can also introduce substantial experimental and calibration overhead. 
In sensing and calibration with high precision requirements, the short resolved evolution times proportional to the accuracy demand are also experimentally infeasible.
This raises the question of whether one can retain strong evolution time (sample) guarantees in a \emph{control-free and ancilla-free} setting.
From the lower bound perspective, quantum controls can fundamentally change the complexity of Hamiltonian learning: with control, one can reach Heisenberg-limited scaling, whereas without control, broad classes of many-body Hamiltonians are only learnable at the standard quantum limit (SQL) of $\Omega(1/\epsilon^2)$~\cite{dutkiewicz2024advantage,chen2025lower}.
Yet before this work, the best-known ansatz-free Hamiltonian learning protocols in the control-free setting require an $O(1/\epsilon^4)$ evolution time scaling without ancilla qubits or $O(1/\epsilon^3)$ with ancilla qubits~\cite{zhao2025learning,ivashkov2026ansatz,caro2024learning,yu2023robust,arunachalam2025testing}, leaving a substantial gap between lower and upper bounds.
This motivates the main question of the present work:
\vspace{-0.5em}
\begin{quote}
\centering
\emph{\mbox{Can we achieve SQL in learning ansatz-free Hamiltonian without control or ancilla?}}
\end{quote}
In this work, we answer this question essentially affirmatively.
We show that arbitrary Hamiltonians $H$ with bounded operator norm can be learned using \emph{control-free and ancilla-free} protocols achieving an SQL scaling total evolution time cost.
\begin{theorem}[Informal, see Theorem~\ref{thm:main}]\label{thm:upper_informal}
Let $H=\sum_{u}\lambda_u P_u$ be an unknown Hamiltonian on $n$ qubits with $\norm{H}\leq \Lambda$. 
There is a computationally efficient algorithm that estimates every coefficient $\{\lambda_u\}$ up to $\epsilon$ with high probability using total evolution time
\begin{align}
O\left(\frac{\Lambda}{\epsilon^2}\log\frac{\Lambda}{\epsilon}\right).
\end{align}
\end{theorem}
\noindent A key feature of our result is that the protocol achieving this bound is not only control-free and ancilla-free, but also uses only \emph{Pauli product state preparation measurement}. 
Moreover, the average probe time remains $O(1/\Lambda)$ and is independent of $\epsilon$. 
This feature is particularly appealing for sensing and calibration applications, where extremely short-time resolution is often a major practical bottleneck.
As a result, our protocol is substantially closer to the capabilities of near- and intermediate-term in situ experiments.

The main idea behind the $1/\epsilon^2$ scaling is a two-stage procedure. 
First, we use product Pauli-basis experiments to find where the large Hamiltonian terms could be. 
In each experiment, we prepare a random product eigenstate in a Pauli basis, let the system evolve for a random time, and measure in the same basis. 
We only keep the bitwise difference between the prepared string and the measured string. 
We call this bit-flip pattern a displacement. 
Several Pauli terms may create the same displacement, and the mass of a displacement means the sum of the squared coefficients of all terms that can create that same pattern. 
For example, a Hamiltonian coefficient of size $\epsilon$ contributes mass $\epsilon^2$. 
Naively estimating every possible $\epsilon^2$-size mass directly would require about $1/\epsilon^4$ samples. 
Our protocol avoids this by choosing the evolution time from a band-limited second-order derivative sampling rule, based on the norm bound $\|H\|\le\Lambda$ and a heavy sampling scheme within the band~\cite{whittaker1915functions,kotelnikov1933transmission,shannon1949communication,butzer2012shannon}. 
With the choice of evolution time, heavy displacements are sampled directly: any displacement with mass at least $\epsilon^2$ appears with probability $\Omega(\epsilon^2/\Lambda^2)$. 
Hence, $O(\Lambda^2/\epsilon^2)$ shots find all heavy displacements. 
Since the average shot duration is $O(1/\Lambda)$, this step requires $O(\Lambda/\epsilon^2)$ total evolution time.

After these bit-flip patterns are found in two complementary Pauli bases, we combine them into a candidate list of Pauli strings by a direct product. 
The second stage estimates the signed coefficients only on this candidate list. 
Each shot in this stage produces a pair of complete input and output sign strings given the choice of input and output Pauli bases. 
The same recorded strings can be reused by classical postprocessing to compute a parity estimate for every candidate that the sampled basis can see. 
A fixed nonzero Pauli label is seen with probability at least $1/3$, so a simple reweighting gives an unbiased estimate. 
Thus, one dataset estimates all candidates together, and the simultaneous guarantee costs only logarithmically in the candidate-list size. 
This parallelization does not rely on locality or bounded-degree commuting measurement groups as the previous work~\cite{francca2025learning,evans2019scalable,bakshi2024structure}.

We compare our algorithm with the previous results as in the following Table~\ref{tab:results}.
\begin{table}[htbp]\label{tab:results}
\centering
\resizebox{1.0\columnwidth}{!}{
\begin{tabular}{lcccc}
\toprule
Reference & Hamiltonian structure & Control & Ancilla & Total evolution time \\
\midrule
Huang et al.~\cite{huang2023learning} & Local \& Known structure  & Discrete & No & $O(1/\epsilon)$ \\
Bakshi et al.~\cite{bakshi2024structure} & Local \& Bounded-interaction & Discrete & No & $O(\log n/\epsilon)$ \\
Zhao~\cite{zhao2025learning} & Ansatz-free & Continuous / Discrete & Yes & $\widetilde O(M/\epsilon)$ / $\widetilde O(\Lambda^3/\epsilon^4)$ \\
Hu et al.~\cite{hu2025ansatz} & Ansatz-free & Discrete & Yes / No & $\widetilde O(M^2/\epsilon)$ / $\widetilde O(M^3/\epsilon)$ \\
\makecell[l]{Sinha and Tong~\cite{sinha2025improved}, \\Abbas et al~\cite{abbas2025nearly}} & Ansatz-free & Discrete & Yes & $\widetilde O(M/\epsilon)$ \\
Arunachalam et al~\cite{hu2025ansatz} & Ansatz-free & No & Yes & $\widetilde O(\Lambda^2/\epsilon^3)$ \\
\textbf{This work} & Ansatz-free & No & No & $\widetilde{\Theta}(\Lambda/\epsilon^2)$ \\
\bottomrule
\end{tabular}}
\vspace{0.2em}
\caption{Comparison with representative prior works on learning $n$-qubit Hamiltonians from time evolutions. For ansatz-free Hamiltonians, there are two settings which use $M$ to denote the sparsity in the Pauli basis and $\Lambda$ to denote the operator norm.}
\vspace{-0.8em}
\end{table}

\vspace{-1em}
We also prove that this performance is optimal for the control-free setting.
\begin{theorem}[Informal, see Theorems~\ref{thm:unified-log}]
Given accuracy demand $\epsilon$ and norm bound $\Lambda$ satisfying $\Lambda^2/\epsilon^2\leq O(n)$, there exists a family of Hamiltonians of bounded norm $\Lambda$ for which any control-free protocol learning requires a total evolution time of at least
\begin{align}
\Omega\left(\frac{\Lambda}{\epsilon^2}\log\frac{\Lambda}{\epsilon}\right).
\end{align}
\end{theorem}
Our result clarifies the relation between Hamiltonian learning and quantum metrology. 
In quantum metrology, one typically detects perturbations on parameters of a Hamiltonian with a known structure. 
In our setting, by contrast, the support of the Hamiltonian is unknown and may lie among exponentially many Pauli strings. 
Our results show that this seemingly harder \emph{structure-learning} task can be carried out at essentially the same total-time scale with only a logarithmic overhead that our lower bound shows is unavoidable.

Our design also makes the protocol robust to state-preparation-and-measurement (SPAM) errors. 
Previous research already shows that preparation and measurement errors can be treated as perturbations of the sampling probabilities and of the coefficient-estimation signals~\cite{hu2025ansatz}. 
In our setting, the argument is cleaner because there are no coherent controls, no ancillary Bell measurements, and no inverse dynamics. 
Assuming the state preparation and measurement noise are local depolarization channels of constant rate, which is a standard assumption in the previous work~\cite{ivashkov2026ansatz}, our algorithm, after calibration, maintains the same total evolution time as Theorem~\ref{thm:upper_informal} when the underlying Hamiltonian is local, i.e., supported on Pauli observables of weight at most some constant.
\begin{theorem}[Informal, see Theorem~\ref{thm:spam_robust_main}]
Assume that the preparation and measurement errors are calibrated local depolarizing channels with constant reliability, and that the calibrated error in the recorded same-basis displacement labels is below the signal scale.
If the Hamiltonian is local, i.e., supported on Pauli observables of at most constant weight, then a calibrated version of the algorithm in Theorem~\ref{thm:upper_informal} recovers the Hamiltonian up to $\epsilon$ with total evolution time remaining $\tilde{O}(\Lambda/\epsilon^2)$.
\end{theorem}

\section{Preliminaries}\label{sec:prelim}

In this section, we provide the background concepts and results required throughout this paper.
We use $\norm{A}$ and $\norm{A}_{\HS}$ to represent the operator norm and Hilbert-Schmidt (Frobenius) norm of matrix $A$.
For a function $f(x)$ defined on $x\in\mathbb{R}$, we use $\norm{f(x)}_1=\int_{\mathbb{R}}\abs{f}dx$ to represent the $L^1$ norm of the function.
We use standard big-O notations $(O,\Theta,\Omega)$, and use $\tilde{O}$ and $\tilde{\Theta}$ to hide the poly-logarithmic dependence in big-O notations.


\subsection{Pauli observables}
We start with some standard concepts in quantum information. 
A general $n$-qubit quantum state can be represented as a positive semi-definite matrix $\rho\in\mathbb{C}^{d\times d}$ of $d=2^n$ dimensions with $\Tr(\rho)=1$.
When the state has rank-$1$ and thus $\Tr(\rho^2)=1$, it is known as a pure state and is denoted as $\ket{\psi}$, $\ket{\phi}$, or $\ket{\varphi}$ throughout this paper.
We denote the $n$-qubit Hilbert space as $\mathcal{H}^d$ with $d=2^n$, and denote the $n$-qubit identity by $I_n$, omitting the subscript when it is clear from the context.

We define $\mathcal{P}_n=\{I,X,Y,Z\}^{\otimes n}$ to be the set of Pauli observables (strings), where
\begin{align}
I=\begin{pmatrix}
1 & 0 \\ 0 & 1
\end{pmatrix},\quad
X=\begin{pmatrix}
0 & 1 \\ 1 & 0
\end{pmatrix},\quad
Y=\begin{pmatrix}
0 & -i \\ i & 0
\end{pmatrix},\quad
Z=\begin{pmatrix}
1 & 0 \\ 0 & -1
\end{pmatrix}.
\end{align}
are single-qubit Pauli operators which represent transformations on a single qubit.
For any two Pauli strings $P,Q\in\mathcal{P}_n$, we have either $P$ and $Q$ commute, i.e. $[P,Q]=PQ-QP=0$, or $P$ and $Q$ anti-commute, i.e. $\{P,Q\}=PQ+QP=0$

Due to the relationship $Y = iXZ$, each Pauli operator on a single qubit can be also represented by a two-dimensional vector. 
This vector contains the parts representing the action of the $X$ and the part representing the action of the $Z$. 
Generally, for an $n$-qubit Pauli string, we have a binary vector in $V=\F^{2n}=\F^n\times\F^n$ as $u=(x\mid z)$ for $x,z\in\F^n$.
We call $V$ the symplectic space and $u$ the symplectic vector, and it can be used as the label of any Pauli string. 
To each label $u=(x\mid z)\in V$, we associate the canonical Hermitian Pauli string $P_u$ with $P_u^\dagger=P_u$ and $P_u^2=I$. 
The interactions between Pauli operators $P_u$ and $P_v$ with $u=(x\mid z)$ and $v=(x'\mid z')$ can also be represented in symplectic form. 
We define $\omega(u,v)\coloneqq x\cdot z'+z\cdot x'\in\F$ and the Walsh character
$\chi_u(v)\coloneqq(-1)^{\omega(u,v)}$, and one can verify that
\begin{align}
P_uP_v=\chi_u(v)P_vP_u,\qquad
2^{-2n}\sum_{v\in V}\chi_u(v)\chi_{u'}(v)=\mathbbm{1}[u=u'].
\end{align}


\subsection{Hamiltonians of interacting systems}

We now introduce the concept of Hamiltonian, which encodes the interaction forces between quantum particles in a physical system.
A Hamiltonian is a Hermitian operator $H\in\mathcal{H}^d$ with $H^\dagger=H$, with its time evolution described as $e^{-iHt}$.
Physically, it determines how the wave function of a quantum system evolves over time via the Schr\"{o}dinger equation $i\frac{\partial}{\partial t}|\psi(t)\rangle = \hat{H}|\psi(t)\rangle$, where $H$ is Hamiltonian and $|\psi(t)\rangle$ is the state of the quantum system at time $t$.
When the system energy is conserved, the above equation has a solution $|\psi(t)\rangle  = e^{-{iHt}}|\psi(0)\rangle$.
Given a Hamiltonian, we assume it has eigenvalues $\{E_j\}_j$ and corresponding eigenstates $\{\ket{e_j}\}_j$ with $H|e_j\rangle=E_j|e_j\rangle$.
When the Hamiltonian $H$ has bounded operator norm $\norm{H}\leq\Lambda$, we also have $\abs{E_j}\leq\Lambda$.
Besides the above spectral decomposition of Hamiltonians according to their eigenstates and eigenvalues, Hamiltonians can also be decomposed in the operator space. 
A common set of complete orthogonal bases is the Pauli operator $\mathcal{P}^{\otimes n}$, and any Hermitian operator (thus Hamiltonian) can be decomposed into the Pauli basis. 
A Hamiltonian is said to be $M$-sparse if the Hamiltonian can be written as $H=\sum_{u\in S}\lambda_u P_u$ with $\abs{S}\leq M$.
In the extreme case when $M=4^n$, the set of $4^n$-sparse Hamiltonian covers all Hamiltonians.
Based on the terminology above, we can clearly define our learning objectives. 
\begin{problem}[Hamiltonian Learning]\label{prob:sparse_ham_recover}
Let $H$ be a $M$-sparse Hamiltonian with $\norm{H}\leq\Lambda$ such that $H=\sum_{u\in S}\lambda_u P_u$, where $\lambda_u\in\mathbb{R}$ denotes the interaction strength for corresponding Pauli string $P_u$, and $S \subset V$ denotes the unknown support with $\abs{S}\leq M$.
Suppose we are given oracle access to the time-
evolution unitary $U(t)=e^{-iHt}$, our goal is to output an estimate $\wh{H}=\sum_{u\in \wh{S}}\wh{\lambda}_u P_u$ such that $|\wh{\lambda}_u-\lambda_u|\leq \epsilon$ for any $u\in V\setminus\{0\}$.
\end{problem}
Note that if $H=cI+H_0$ for some constant $c$, then $e^{-iHt}=e^{-ict}e^{-iH_0t}$ differs from $e^{-iH_0t}$ only by a global phase and is therefore invisible. 
So we do not consider learning terms that are entirely composed of identities.
Given an accuracy demand $\epsilon > 0$, to solve Problem~\ref{prob:sparse_ham_recover}, we only need to recover the coefficient in $\epsilon$-support $S_\epsilon\coloneqq\{u\in V\setminus\{0\}:|\lambda_u|\ge\epsilon\}$ up to $\epsilon$, and output $\wh{\lambda}_u=0$ for all other $u\in V\setminus S_\epsilon$. 

\subsection{The access models}

\begin{figure}[htbp]
\centering
\includegraphics[width=0.75\linewidth]{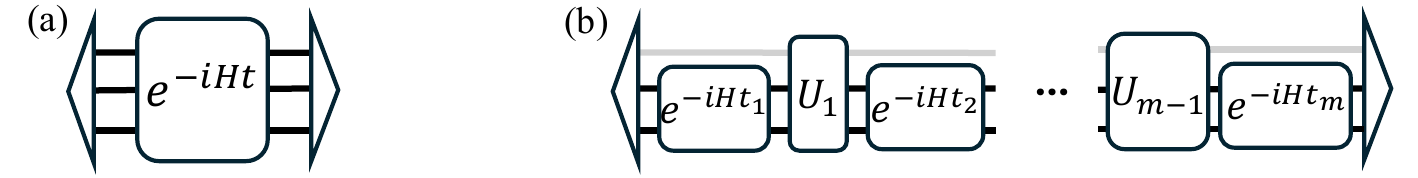}
\caption{(a) Control-free ancilla-free Hamiltonian learning protocols considered in this work. (b) The general Hamiltonian learning settings with ancillary qubits (grey qubits) and quantum control.}
\label{fig:control_free_protocol}
\end{figure}

In this section, we consider how we get access to the Hamiltonian and how we perform experiments to learn the Hamiltonian. 
The most general Hamiltonian learning protocols can be described by Figure~\ref{fig:control_free_protocol} (b).
In each experiment, the learner can prepare an input state on all working qubits (black qubits) and ancillary qubits (grey qubits), queries the Hamiltonian $m$ times with evolution time $t_1,...,t_m$ on the $n$ working qubits interleaved with $m-1$ unitary control $U_1,...,U_{m-1}$ between each two evolution time queries on all qubits, and make the final measurement on all qubits.
However, while hardware-aware controls provide useful algorithmic flexibility on near-term devices~\cite{cheng2025epoc,liang2024napa}, allowing arbitrary interleaved controls and ancillary qubits in a Hamiltonian-learning access model imposes stronger experimental assumptions.
In this paper, we therefore consider protocols without ancillary qubits and quantum control (as shown in Figure~\ref{fig:control_free_protocol} (a)). 
Moreover, we assume that the input state and measurement are restricted to a product state and a product measurement.
Quantitatively, a single experimental can be described as follows. 
\begin{definition}[Control-free, ancilla-free, and product IO access model]\label{def:access_model}
We consider learning an unknown Hamiltonian $H$ with a series of experiments described as follows:
In each experiment (shot), the learner can only
\begin{enumerate} 
\item prepare an arbitrary product state, 
\item perform an evolution $U(t)=e^{-iHt}$ for some chosen duration  $t\ge 0$,
\item and perform a single-qubit measurement.
\end{enumerate}
The learner can then collect all the measurement outcomes and reconstruct the coefficient of the Hamiltonian via classical postprocessing of complexity polynomial in $n$ and $1/\epsilon$.
\end{definition}
We remark that our protocol only require Pauli product input states and Pauli product measurement, which is an even more restricted setting than Definition~\ref{def:access_model}.
The cost of a protocol in this model is measured by the total evolution time $T_{\mathrm{tot}}\coloneqq \sum_{\text{all shots}} t_{\text{shot}}$ over all shots (experiments).

\subsection{Same-Pauli and Cross-Pauli measurements}\label{ssec:sectors}
Suppose $P$ and $Q$ are two Pauli strings. 
We use the Heisenberg notation to describe the time evolution of a Pauli string $P$ as $P(t)\coloneqq e^{iHt}Pe^{-iHt}$.
Given different choices of Pauli bases for the initial state and the final measurements, we classify the access model in Definition~\ref{def:access_model} into two types. 
The first type is referred to as the same-Pauli measurement, where the system is prepared under a Pauli string $P$ and read out under the same Pauli basis. 
The corresponding statistical signal can be represented by the auto-correlation function:
\begin{align}\label{eq:auto-correlation}
C_P(t) \coloneqq 2^{-n} \Tr(P(t)P) = 2^{-n} \Tr(e^{iHt}P e^{-iHt}P)
\end{align}
The first-order derivative of $C_P(t)$ at $t=0$ is proportional to $\Tr([H,P]P)$, which is always zero for any Pauli string $P$. 
Its second-order derivative is proportional to $\Tr([H,[H,P]]P)$, indicating that the initial evolution of the system is entirely driven by terms that anti-commute with $P$. 
Since the first order derivative is zero, we can only read the coefficients $\lambda_u^2$ corresponding to the quadratic terms $t^2$, resulting in sign-free signals. 
Therefore, we use same-Pauli measurements only to recover $\wh{S}_\epsilon$.

Another type is called the cross-Pauli measurement, where the system is prepared under one Pauli string $P$ but read out under a different Pauli basis $Q$. 
The corresponding signals are characterized by cross-correlation functions
\begin{equation}\label{eq:cross-correlation}
F_{P,Q}(t) \coloneqq 2^{-n} \Tr(P(t)Q) = 2^{-n} \Tr(e^{iHt}P e^{-iHt}Q).
\end{equation}
Here, the first-order derivative at $t=0$ is proportional to $\Tr([H,P]Q)$, which is non-zero. 
This indicates that the signal changes linearly with time $t$ at $t=0$, thus preserving the sign of $\lambda_u$. 
Therefore, cross-measurements can be used to estimate precise individual parameters after we have characterized $\wh{S}_\epsilon$.

\subsection{The trace identity}\label{ssec:trace}
We will need the following result to connect product-state experiments with the trace functions.
\begin{proposition}[Trace rule]\label{prop:trace-main}
Given Pauli strings $P,Q$, sample uniformly from an eigenbasis $\{|\psi\rangle\}$ of $Q$, let $q\in\{\pm1\}$ be the eigenvalue of the prepared state, evolve by $e^{-iHt}$, and measure $P$ with outcome $m\in\{\pm1\}$. Then, we have
\begin{align}
\E[qm]=2^{-n}\Tr(P(t)Q).
\end{align}
\end{proposition}
This rule will be used for both same-Pauli and cross-Pauli measurements.
\section{Sampling subroutines}\label{sec:tricks}
This section provides several sampling subroutines as essential building blocks of our algorithm. 
To begin with, we decompose the Hamiltonian learning task into several subproblems.
The first subproblem is how to design experiments with measurement results corresponding to the short-time derivatives of correlation functions. 
As the Hamiltonian has a bounded operator norm $\Lambda$, these can be addressed by constructing a kernel (see Section~\ref{ssec:time-sample} for details). 
The second subproblem is how to choose the input states and measurement bases to utilize the information revealed by same-Pauli measurements and cross-Pauli measurements. 
To this end, we first use same-Pauli measurements to determine the support in the structure learning stage (see Section~\ref{ssec:displacement}), and then use cross-Pauli measurements to accurately estimate the coefficients in the coefficient learning stage (see Section~\ref{ssec:parity-sample} for details). 
Concretely, when we represent Pauli strings in symplectic forms, same-Pauli measurements recover a candidate set of size independent of $n$ that contains the $\epsilon$-support of the Hamiltonian. 
Cross-Pauli measurements with randomly selected local Pauli bases recover all coefficients in the $\epsilon$-support set by designing the parity variable and applying weight corrections. We also state the two subroutines in the form of pseudocode.

\subsection{Random time sampling}\label{ssec:time-sample}
We now show how to obtain an unbiased estimation of the derivatives of the correlation function with a finite number of samples by constructing kernels. 
Consider same-Pauli measurements first, and the cross-Pauli measurement case can be addressed by a similar method. 
First, we diagonalize the Hamiltonian as $H|e_j\rangle=E_j|e_j\rangle$ with $|E_j|\le\Lambda$.
We have $\langle e_j|P(t)|e_k\rangle=e^{i(E_j-E_k)t}\langle e_j|P|e_k\rangle$, and thus $C_P(t)=2^{-n}\sum_{j,k}|\langle e_j|P|e_k\rangle|^2e^{i(E_j-E_k)t}$.
This indicates that $C_P(t)$ has a Fourier support of $[-2\Lambda,2\Lambda]$, which is band-limited. 
One can construct a kernel to recover derivatives of $C_P(t)$ at zero as a weighted average of values at finite times. 
In particular, as we can rewrite $C_P(t)=\int_{-2\Lambda}^{2\Lambda}\cos(\omega t)d\mu(\omega)$ as an integral within the band limitation for some probability measure $\mu$.  
Differentiating under the integral, we derive that
\begin{align}
C_P''(0)=-\int_{-2\Lambda}^{2\Lambda}\omega^2d\mu(\omega).
\end{align}
To obtain the sampling time, we need to rewrite the above equation to return to the time domain. 
This follows the same principle as the Whittaker-Kotelnikov-Shannon sampling theorem and its derivative formulae~\cite{whittaker1915functions,kotelnikov1933transmission,shannon1949communication,butzer2012shannon}.
We fix a smooth cutoff $\varphi\in C_c^\infty(\mathbb R)$ satisfying $\varphi(\xi)=1\ (|\xi|\le1)$ and $\varphi(\xi)=0\ (|\xi|\ge2)$. We define $\Phi_\Lambda(\omega)\coloneqq\omega^2 \varphi(\omega/(2\Lambda))$. The cutoff is equal to one on the true frequency window $|\omega|\le2\Lambda$, and it smoothly goes to zero before $|\omega|=4\Lambda$. 
The smooth transition is useful because a sharp cutoff would create sinc-type tails in the kernel, whereas the $C^\infty$ cutoff ensures the integrability. 
This smooth cutoff is important because it makes the kernel integrable. We define the kernel $K_\Lambda(t)\coloneqq-\frac{1}{2\pi}\int_{\mathbb R}\Phi_\Lambda(\omega)e^{i\omega t}d\omega$.
We can thus obtain the Fourier inversion:
\begin{align}
\int_{\mathbb R}K_\Lambda(t)C_P(t)dt=2\int_0^\infty K_\Lambda(t)C_P(t)dt=-\int_{-2\Lambda}^{2\Lambda}\omega^2d\mu(\omega)=C_P''(0).
\end{align}
We also need to regularize the kernel by defining
\begin{align}
q_\Lambda(t)\coloneqq\frac{2|K_\Lambda(t)|}{\|K_\Lambda\|_1},\qquad t\ge0.
\end{align}
is the probability density corresponding to the sampling time. 
We are now ready to provide the random time sampling routine with its performance guarantee.

\begin{proposition}[Unbiasedness for estimating $C_P''(0)$]\label{prop:even-sample-main}
Sample $\tau\sim q_\Lambda$, and then perform a single same-Pauli measurement with respect to $P_v$ with evolution time $\tau$. 
Let $q\in\{\pm1\}$ be the eigenvalue of the prepared state of $P_v$, evolve by $e^{-iH\tau}$, and measure $P_v$ with outcome $m\in\{\pm1\}$.
Define $X_{v,\tau}\coloneqq qm$ and the weighted sample $D_v\coloneqq\|K_\Lambda\|_1\sign(K_\Lambda(\tau))X_{v,\tau}$.
Then the estimator $D_v$ is an unbiased estimator of $C_{P_v}''(0)$, i.e., $\E[D_v]=C_{P_v}''(0)$.
\end{proposition}

Similarly, for cross-Pauli measurements, we can also construct a kernel for estimating the first-order derivative. 
We first define $L_\Lambda(t)\coloneqq\frac{1}{2\pi}\int_{\mathbb R}\Psi_\Lambda(\omega)e^{-i\omega t}d\omega$, where $\Psi_\Lambda(\omega)\coloneqq i\omega\varphi(\omega/(2\Lambda))$, and then regularize it as
\begin{align}\label{eq:p_lambda}
p_\Lambda(t)\coloneqq\frac{2|L_\Lambda(t)|}{\|L_\Lambda\|_1}, \qquad t\ge 0.
\end{align}
Due to the different structure of the first-order derivative, we need additional parity designs, which will be covered in Section~\ref{ssec:parity-sample}.
\subsection{Displacement sampling}
\label{ssec:displacement}

We now discuss how to locate the support set using same-Pauli measurements. 
Note that $C_P''(0)=-4\sum_{u:\{P_u,P\}=0}\lambda_u^2$, indicating that the same-Pauli measurement signal of a Pauli string is the sum of squares of the coefficients of all terms that anti-commute with it. 
For simplicity, we denote $P_{(x\mid z)} = i^{x\cdot z}\prod_{i=1}^n X_i^{x_i} Z_i^{z_i}$ in the symplectic form.
Let $|a\rangle = |a_1,\dots,a_n\rangle$ be a computational basis state of $Z$ such that $Z^z |a\rangle = (-1)^{z\cdot a} |a\rangle$ and $X^x |a\rangle = |a+x\rangle$. 
Therefore, we have $P_{(x\mid z)} |a\rangle=i^{x\cdot z}(-1)^{z\cdot a}|a+x\rangle$.
Similarly, we define $|b\rangle_X = |b_1,\dots,b_n\rangle_X$ as a $X$-basis eigenvector. Since $Z^z|b\rangle_X=|b+z\rangle_X$ and $X^x|b+z\rangle_X=(-1)^{x\cdot(b+z)}|b+z\rangle_X$, we have $P_{(x\mid z)}|b\rangle_X=i^{x\cdot z}(-1)^{x\cdot(b+z)}|b+z\rangle_X
=(-i)^{x\cdot z}(-1)^{x\cdot b}|b+z\rangle_X$.
The above results indicate that $X$ operators result in bit flips for $Z$-basis same-Pauli measurements, and $Z$ operators result in bit flips for $X$-basis same-Pauli measurements. 
We denote displacements $d$ and $e$ as the bit flip patterns of $Z$- and $X$-basis same-Pauli measurements, and define 
\begin{align}
W_x(d)\coloneqq\sum_{z\in\F^n}\lambda_{(d\mid z)}^2,\qquad
W_z(e)\coloneqq\sum_{x\in\F^n}\lambda_{(x\mid e)}^2.
\end{align}
as sums of $\lambda_u^2$ of all terms that have corresponding displacements of the $Z$- and $X$-basis.
We can show the following proposition.

\begin{proposition}[Distributions of $Z$- and $X$-basis displacements]\label{prop:displacements-main}
Let $D\in\F^n$ be the displacement for $Z$-basis same-Pauli measurement with a uniformly random eigenstate of $Z^n$ as input and $\tau\sim q_\Lambda$, and let $E\in\F^n$ be the displacement for $X$-basis same-Pauli measurements.  
Then for all nonzero $d,e\in\F^n$, we have
\begin{align}
\Pr[D=d]\ge \frac{2W_x(d)}{\|K_\Lambda\|_1},
\qquad
\Pr[E=e]\ge \frac{2W_z(e)}{\|K_\Lambda\|_1}.
\end{align}
\end{proposition}

Given the accuracy demand $\epsilon$, we define
\begin{align}
\Pi_x^{(\epsilon)}\coloneqq\{d\in\F^n:W_x(d)\ge\epsilon^2\},\qquad
\Pi_z^{(\epsilon)}\coloneqq\{e\in\F^n:W_z(e)\ge\epsilon^2\}.
\end{align}
Note that $\sum_d W_x(d)=\sum_e W_z(e)=\sum_u\lambda_u^2=2^{-n}\Tr(H^2)\le\Lambda^2$, each of these sets contains at most $\Lambda^2/\epsilon^2$ nonzero elements.
We then design a displacement sieve algorithm to reconstruct two sets that contain $\Pi_x^{(\epsilon)}$ and $\Pi_z^{(\epsilon)}$. The concrete procedure is given in Algorithm~\ref{alg:displacement-sieve}.
In the $Z$ basis, the observed bit flip $d=a+b$ estimates the $x$-projection; in the $X$ basis, the same procedure estimates the $z$-projection. 
The time parts in the pseudocode adopt the discrete-time assumption (see Appendix~\ref{app:lattice}).
\begin{algorithm}[htbp]
\SetAlgoLined
\caption{Displacement sieve from same-Pauli shots}
\label{alg:displacement-sieve}
\KwIn{detection scale $\epsilon$, failure probability $\eta$, norm bound $\Lambda$}
\KwOut{projection sets $\widehat\Pi_x,\widehat\Pi_z\subseteq\mathbb F_2^n$}
Set $p\leftarrow \epsilon^2/(4\kappa_0\Lambda^2)$, $r\leftarrow\lceil\Lambda^2/\epsilon^2\rceil$, $N\leftarrow\lceil C_0p^{-1}\log(4r)\rceil$, $L\leftarrow\lceil C_1\log(4/\eta)\rceil$, and $\tau_{\max}\leftarrow4\kappa_1N/\Lambda$ for some suitable constants $C_0,C_1$ and constants $\kappa_0,\kappa_1$ chosen in Appendix~\ref{app:quadratic}\;
Initialize $\widehat\Pi_x\leftarrow\varnothing$ and $\widehat\Pi_z\leftarrow\varnothing$\;
\For{$B\in\{Z,X\}$}{
  \For{$\ell=1,\ldots,L$}{
    Sample $t_1,\ldots,t_N\sim q_\Lambda$ independently and set $\Theta\leftarrow\sum_{j=1}^N t_j$\;
    \If{$\Theta>\tau_{\max}$}{abort this round and continue to the next round\;}
    Initialize counts $A_h\leftarrow0$ for all observed nonzero displacements $h\in\mathbb F_2^n$\;
    \For{$j=1,\ldots,N$}{
      Prepare a uniformly random product eigenstate $|a_j\rangle_B$ in the $B$ basis\;
      Evolve for time $t_j$ and measure all qubits in the same $B$ basis, obtaining $b_j\in\mathbb F_2^n$\;
      Set $h_j\leftarrow a_j+b_j$ over $\mathbb F_2$\;
      \If{$h_j\neq0$}{$A_{h_j}\leftarrow A_{h_j}+1$\;}
    }
    \eIf{$B=Z$}{
      $\widehat\Pi_x\leftarrow\widehat\Pi_x\cup\{d\neq0:A_d\ge Np\}$\;
    }{
      $\widehat\Pi_z\leftarrow\widehat\Pi_z\cup\{e\neq0:A_e\ge Np\}$\;
    }
  }
}
\Return{$\widehat\Pi_x,\widehat\Pi_z$}\;
\end{algorithm}

\begin{lemma}[Displacement sieve algorithm]\label{lem:projection-main}
For any $0<\epsilon\le\Lambda$ and failure probability $\eta\in(0,1)$, there is a protocol using same-Pauli measurements that outputs sets $\wh\Pi_x,\wh\Pi_z\subseteq\F^n$ with size $O(\tfrac{\Lambda^2}{\epsilon^2}\log\tfrac{1}{\eta})$ such that, $\Pi_x^{(\epsilon)}\setminus\{0\}\subseteq\wh\Pi_x$ and $\Pi_z^{(\epsilon)}\setminus\{0\}\subseteq\wh\Pi_z$, with probability at least $1-\eta$ in deterministic total evolution time at most $O(\frac{\Lambda}{\epsilon^2}\log\tfrac{4\Lambda}{\epsilon}\log\tfrac{4}{\eta})$.
\end{lemma}
After performing the displacement sieve algorithm, we formed a candidate set by combining these components using the Cartesian product.
\begin{align}
\wh\cU_{\mathrm{proj}}\coloneqq(\wh\Pi_x\cup\{0\})\times(\wh\Pi_z\cup\{0\})\setminus\{0\}.
\end{align}
Every coefficient with $|\lambda_{(x\mid z)}|\ge\epsilon$ belongs to this candidate set. 
Therefore, the final support is searched inside \textbf{$\widehat\cU_{\mathrm{proj}}$}. 
If the displacement sieve is run with failure probability $\eta$, then the deterministic size bound is $|\widehat\cU_{\mathrm{proj}}|
\le
O\left(\frac{\Lambda^4}{\epsilon^4}\log^2\frac{1}{\eta}\right).$

\subsection{Parity sampling}\label{ssec:parity-sample}
To estimate the values of the coefficients given $\wh\Pi_x,\wh\Pi_z$, we use cross-Pauli measurements. 
Fix a Pauli string $c=(c_1,\ldots,c_n)\in\{X,Y,Z\}^n$.
At each qubit $i$, we choose the two remaining Pauli bases $a_i,b_i$ in cyclic order $(c_i,a_i,b_i)\in\{(X,Y,Z),(Y,Z,X),(Z,X,Y)\}$.
The letter $c_i$ is the basis whose coefficient we want to estimate through a cross-Pauli measurement experiment with input basis $a_i$ and output basis $b_i$. 
Given $A,B\subseteq[n]$, we define $Q_A\coloneqq\bigotimes_{i=1}^n Q_i^{(A)}$ with $Q_i^{(A)}=a_i$ if $i\in A$ and $Q_i^{(A)}=I$ if $i\notin A$.
Similarly, we define $P_B\coloneqq\bigotimes_{i=1}^n P_i^{(B)}$ with $P_i^{(B)}=b_i$ if $i\in B$ and $P_i^{(B)}=I$ if $i\notin B$.
For each $(A,B)$, we define an associated Pauli label $u_c(A,B)$ with $(P_{u_c(A,B)})_i$ being $I$, $a_i$, $b_i$, and $c_i$ when $i\notin A\cup B$, $i\in A\setminus B$, $i\in B\setminus A$, and $i\in A\cap B$.

We consider the first-order derivative of the corresponding cross-correlation $F^{(c)}_{A,B}(t)\coloneqq2^{-n}\Tr(P_B(t)Q_A)$. 
A naive observation is that $\frac12\Bigl(F^{(c)}_{A,B}(t)-F^{(c)}_{A,B}(-t)\Bigr)$ gives an accurate estimation of the derivative. 
However, our access model in Definition~\ref{def:access_model} does not allow inverse evolutions. 
To address this issue, we notice that $P_{u_c(A,B)}P_B=(-1)^{|A\cap B|}P_BP_{u_c(A,B)}$, which indicates that $(F^{(c)}_{A,B})'(0)$ is nonzero only when $|A\cap B|$ is odd.
Based on this fact, we design the measurement subroutine to measure the corresponding first-order derivative when $|A\cap B|$ is odd.
First, we sample $\tau$ from the cross-Pauli measurements kernel density $p_\Lambda$ as in Eq.~\eqref{eq:p_lambda} and sample $S\in\{\pm1\}$ uniformly. 
If $S=+1$, we prepare a random local $a_i$ eigenstate on each qubit, evolve under the Hamiltonian for time $\tau$, and measure along the basis $b_i$ on every qubit, obtaining outcomes $m_i\in\{\pm1\}$. 
If $S=-1$, we reverse $b_i$ and $a_i$ as in the case of $S=+1$, obtaining outcomes $n_i\in\{\pm1\}$. 
We define the parity variable $Z^{(c)}_{A,B}=\prod_{i\in A}s_i\prod_{i\in B}m_i$ if $S=+1$ and $Z^{(c)}_{A,B}=-\prod_{i\in B}r_i\prod_{i\in A}n_i$ if $S=-1$,
where $s_i,r_i$ are the signs of the eigenstates of the input quantum states for $a_i$ when $S=+1$ or $b_i$ when $S=-1$. 
We then define the weighted parity sample
\begin{align}
G^{(c)}_{A,B}\coloneqq\|L_\Lambda\|_1\sign(L_\Lambda(T))Z^{(c)}_{A,B},
\end{align}
and show that it is an unbiased estimator of $(F^{(c)}_{A,B})'(0)$ as shown in the following proposition.
\begin{proposition}[Unbiasedness for estimating $(F^{(c)}_{A,B})'(0)$]\label{prop:parity-main}
Let $\kappa\coloneqq|A\cap B|$. 
Then $G^{(c)}_{A,B}$ is an unbiased estimator for $(F^{(c)}_{A,B})'(0)$, i.e., $\E[G^{(c)}_{A,B}]=(F^{(c)}_{A,B})'(0)$.
We also have $(F^{(c)}_{A,B})'(0)=2\sigma_c(A,B)\lambda_{u_c(A,B)}$ for odd $\kappa$ and $0$ for even $\kappa$, where $\sigma_c(A,B)\coloneqq i(-i)^\kappa\in\{\pm1\}$ when $\kappa$ is odd.  
In addition, we have $|G^{(c)}_{A,B}|\le2\ell_0\Lambda$ for a constant $\ell_0>0$.
\end{proposition}

Given Proposition~\ref{prop:parity-main}, from one
experiment, we can get a random sample of $Z^{(c)}_{A,B}$ for every pair $(A,B)$ with odd $|A\cap B|$, and each $Z^{(c)}_{A,B}$ corresponds to a label $u=u_c(A,B)$. 
The resulting random variable provides an unbiased sample of $\lambda_u$. 
However, not every label appears in one experiment and different coefficients are sampled with different frequencies. 
A label $u\neq 0$ contributes a sample only when it matches the basis $c$ on an odd number of indices $\mathcal V(c)\coloneqq\{u\in V\setminus\{0\}:|\{i:u_i=c_i\}|\text{ is odd}\}$.
Let $(A_c(u),B_c(u))$ be a pair with $u=u_c(A,B)$, and define
\begin{align}
Y^{(c)}_u\coloneqq\tfrac12\sigma_c(A_c(u),B_c(u))G^{(c)}_{A_c(u),B_c(u)}.
\end{align}
Then $\E[Y^{(c)}_u]=\lambda_u$ and $|Y^{(c)}_u|\le \ell_0\Lambda$, so each visible label contributes one unbiased sample when $u$ is visible.
To achieve global unbiasedness, we need to apply the following correction. 
We choose $c$ uniformly randomly from $\{X,Y,Z\}^n$, and define $q(u)\coloneqq\Pr[u\in \mathcal V(c)]=\tfrac{1-(1/3)^{w(u)}}{2}\ge\tfrac13$, where $w(u)=|\{i:u_i\neq I\}|$.
We then construct the estimator
\begin{align}
\widetilde{Y}_u
\coloneqq q(u)^{-1}\mathbbm{1}[u\in V(C)]Y^{(C)}_u.
\end{align}
Then $\E[\widetilde{Y}_u]=\lambda_u$ and
$|\widetilde{Y}_u|\le 3\ell_0\Lambda$, so that the estimator becomes unconditionally unbiased for every label $u$. Algorithm~\ref{alg:parity-coeff} gives the explicit coefficient estimation routine.

\begin{lemma}[Coefficient estimation]\label{lem:linear-oracle}
For any $0<\epsilon\le\Lambda$, given a nonempty candidate class $\wh\cU_{\mathrm{proj}}\subseteq V\setminus\{0\}$ and failure probability $\eta\in(0,1)$, there is a protocol using cross-Pauli measurements that outputs $\{\wh\lambda_u\}_{u\in \wh\cU_{\mathrm{proj}}}$ satisfying $\max_{u\in \wh\cU_{\mathrm{proj}}}|\wh\lambda_u-\lambda_u|\le\epsilon$ with probability at least $1-\eta$ in deterministic total evolution time at most $O(\tfrac{\Lambda}{\epsilon^2}\log(4|\wh\cU_{\mathrm{proj}}|)\log\tfrac{2}{\eta})$.
\end{lemma}

\begin{algorithm}[h!]
\SetAlgoLined
\caption{Parity coefficient estimation on a candidate set}
\label{alg:parity-coeff}
\KwIn{candidate set $\mathcal U\subseteq V\setminus\{0\}$, accuracy $\alpha$, failure probability $\eta$, norm bound $\Lambda$}
\KwOut{coefficient estimates $\{\widehat\lambda_u\}_{u\in\mathcal U}$}
\If{$\mathcal U=\varnothing$}{return the empty set of estimates\;}
Set $N\leftarrow\lceil C_0\Lambda^2\alpha^{-2}\log(16|\mathcal U|)\rceil$, $R\leftarrow\lceil16\log(2/\eta)\rceil$, and $\tau_{\max}\leftarrow4\ell_1N/\Lambda$ for some suitable constants $C_0$ and $\ell_1$ chosen in Appendix~\ref{app:linear}\;
For every $u\in\mathcal U$, precompute $q(u)=(1-(1/3)^{w(u)})/2$\;
\For{$r=1,\ldots,R$}{
  Sample $C_s\sim\mathrm{Unif}\{X,Y,Z\}^n$ and $T_s\sim p_\Lambda$ independently for $s=1,\ldots,N$; set $\Theta\leftarrow\sum_{s=1}^N T_s$\;
  \If{$\Theta>\tau_{\max}$}{set $\widetilde\lambda_r(u)\leftarrow0$ for every $u\in\mathcal U$, abort this block, and continue to the next block\;}
  Initialize $S_r(u)\leftarrow0$ for all $u\in\mathcal U$\;
  \For{$s=1,\ldots,N$}{
    Let $C_s=(c_1,\ldots,c_n)$ and choose $(a_i,b_i)$ by the cyclic rule $(c_i,a_i,b_i)\in\{(X,Y,Z),(Y,Z,X),(Z,X,Y)\}$\;
    Sample an orientation $\rho_s\sim\mathrm{Unif}\{\pm1\}$\;
    \eIf{$\rho_s=+1$}{
      Prepare random local $a_i$-eigenstates with signs $s_i$, evolve for time $T_s$, and measure in the $b_i$ bases with outcomes $m_i$\;
    }{
      Prepare random local $b_i$-eigenstates with signs $r_i$, evolve for time $T_s$, and measure in the $a_i$ bases with outcomes $n_i$\;
    }
    \ForEach{$u\in\mathcal U$ such that $u\in\mathcal V(C_s)$}{
      Find the unique pair $(A_{C_s}(u),B_{C_s}(u))$ with $u=u_{C_s}(A_{C_s}(u),B_{C_s}(u))$; write $A\leftarrow A_{C_s}(u)$ and $B\leftarrow B_{C_s}(u)$\;
      Set $\kappa\leftarrow|A\cap B|$ and $\sigma\leftarrow i(-i)^\kappa=(-1)^{(\kappa-1)/2}$; visibility ensures that $\kappa$ is odd\;
      \eIf{$\rho_s=+1$}{
        $Z\leftarrow\prod_{i\in A}s_i\prod_{i\in B}m_i$\;
      }{
        $Z\leftarrow-\prod_{i\in B}r_i\prod_{i\in A}n_i$\;
      }
      $G\leftarrow\|L_\Lambda\|_1\operatorname{sign}(L_\Lambda(T_s))Z$\;
      $Y\leftarrow\frac12\sigma G$\;
      $S_r(u)\leftarrow S_r(u)+q(u)^{-1}Y$\;
    }
  }
  Set $\widetilde\lambda_r(u)\leftarrow S_r(u)/N$ for every $u\in\mathcal U$\;
}
\Return{$\widehat\lambda_u\leftarrow\operatorname{median}_{r=1}^R\widetilde\lambda_r(u)$ for every $u\in\mathcal U$}\;
\end{algorithm}

\section{The Algorithm and Performance Analysis}\label{sec:upper}

\subsection{Main algorithm}
\label{ssec:main-algo}
We will now combine the same-Pauli measurements stage in Lemma~\ref{lem:projection-main} and the cross-Pauli measurements stage in Lemma~\ref{lem:linear-oracle} into a complete protocol to solve Problem ~\ref{prob:sparse_ham_recover} as in Algorithm~\ref{algo:upper}.

\begin{algorithm}[htbp]\label{algo:upper}
\SetAlgoLined
\caption{Control-free reconstruction at detection scale $\epsilon$}
\KwIn{target detection scale $\epsilon>0$, failure probability
$\delta\in(0,1)$}
\KwOut{estimate $\wh S_\epsilon$ of the $\epsilon$-detectable support
and estimator $\wh H_\epsilon=\sum_{u\in\wh S_\epsilon}
\wh\lambda_u P_u$}
Set internal scales $\epsilon_{\mathrm{lin}}\leftarrow\epsilon/6$ and
threshold $\theta\leftarrow\epsilon/2$;
Run the displacement sieve (Lemma~\ref{lem:projection-main}) with scale
$\epsilon$ and failure $\delta/2$ to obtain $\wh\Pi_x$, $\wh\Pi_z$;
Form the Cartesian candidate set
$\wh\cU_{\mathrm{proj}}\leftarrow
(\wh\Pi_x\cup\{0\})\times(\wh\Pi_z\cup\{0\})\setminus\{0\}$
\eIf{$\wh\cU_{\mathrm{proj}}=\varnothing$}{
$\wh S_\epsilon\leftarrow\varnothing$;\ \
$\wh H_\epsilon\leftarrow 0$
}{
Run Lemma~\ref{lem:linear-oracle} on $\wh\cU_{\mathrm{proj}}$ at accuracy
$\epsilon_{\mathrm{lin}}$ and failure $\delta/2$ to obtain
$\{\wh\lambda_u\}_{u\in\wh\cU_{\mathrm{proj}}}$
$\wh S_\epsilon\leftarrow
\{u\in\wh\cU_{\mathrm{proj}}:|\wh\lambda_u|>\theta\}$;\ \
$\wh H_\epsilon\leftarrow\sum_{u\in\wh S_\epsilon}\wh\lambda_u P_u$
}
\end{algorithm}

The performance guarantee of Algorithm~\ref{algo:upper} is given by the following theorem.

\begin{theorem}[Performance guarantee of Algorithm~\ref{algo:upper}]\label{thm:main}
Let $H=\sum_{u\in S}\lambda_uP_u$ with $\|H\|\le\Lambda$, and fix $0<\epsilon\le\Lambda$ and $\delta\in(0,1)$. With probability at least
$1-\delta$, Algorithm~\ref{algo:upper} outputs
$(\wh S_\epsilon,\wh H_\epsilon)$ satisfying
\begin{align}
S_\epsilon\subseteq\wh S_\epsilon\subseteq
\{u\in V\setminus\{0\}:|\lambda_u|>\epsilon/3\},\qquad
\max_{u\in S_\epsilon}|\wh\lambda_u-\lambda_u|\le\epsilon/6,
\end{align}
where $S_\epsilon\coloneqq\{u\in V\setminus\{0\}:|\lambda_u|\ge\epsilon\}$.
The total evolution time of Algorithm~\ref{algo:upper} satisfies
\begin{align}
T_{\mathrm{tot}}^{\det}\le C\frac{\Lambda}{\epsilon^2}\log\frac{\Lambda}{\epsilon}\log\frac{1}{\delta}\log\log\frac{1}{\delta}.
\end{align}
\end{theorem}

We remark that any input value $\Lambda\ge \|H\|$ preserves the correctness of our algorithm. 
It only affects runtime via the kernel interval length. 
If $\Lambda$ is unknown, it can be found by a binary search with only a constant-factor overhead in the total evolution time.

\subsection{Implementation under the discrete time assumption}\label{ssec:lattice}
Our displacement sieve sampling subroutine in Lemma~\ref{lem:projection-main} implicitly assumes we can randomly sample an arbitrary time $\tau$ and evolve under $\tau$, which requires the ability to perform continuous time evolution of a Hamiltonian.  
However, real-world physical devices typically have a minimum resolvable time $t_0$. 
We will now adapt our algorithm to an even more restrictive discretization assumption where we only allow resolution time chosen from a discrete-time grid $\mathcal T_{t_0}\coloneqq\{0,t_0,2t_0,\ldots\}$.

First, for any Hamiltonian learning of $H$ when $\|H\|\le\Lambda$, there is a fundamental barrier that requires a resolution time of at most $t_0 <\pi/\Lambda$. We provide a complete proof in the appendix.
Second, we can discretize the sampling rules with only a constant overhead on the total evolution time. 
The key observation is that the correlation function only contains frequencies in $[-2\Lambda,2\Lambda]$, while the kernels $K_\Lambda$ and $L_\Lambda$ are constructed using a smooth cutoff and therefore have frequency support $[-4\Lambda,4\Lambda]$. 
As a result, the product $K_\Lambda(t)f(t)$ has a bandwidth of at most $[-6\Lambda,6\Lambda]$. 
Therefore, as long as the time grid spacing satisfies $t_0 < \frac{\pi}{6\Lambda}$, we have
\begin{align}
\int_{\mathbb R} K_\Lambda(t) f(t)dt=t_0 \sum_{m\in\mathbb Z} K_\Lambda(m t_0) f(m t_0).
\end{align}
In this scenario, we sample an integer $m\ge 0$ with probability proportional to $t_0 |K_\Lambda(m t_0)|$, and run the same experiment at time $\tau=m t_0$. 
If one wants to limit the maximum evolution time per experiment, we can set a truncation parameter $R>0$ and truncate the sum at $m \le R/(2\Lambda t_0)$. 
The resulting error decays faster than any polynomial in $R$, while ensuring a resolution time upper bound of $O(1/\Lambda)$ for each experiment.
We further show that a resolution time of $\Omega(1/\Lambda)$ is necessary in the grid setting.

\subsection{SPAM-robustness of the main algorithm after calibration}\label{sec:spam_robustness}
We now discuss the robustness of Algorithm~\ref{algo:upper} to state-preparation-and-measurement errors. 
The high-level intuition is that SPAM errors affect the two stages of the algorithm in different ways. 
The displacement stage is a support-detection procedure based on rare bit-flip patterns, while the coefficient stage is a linear-estimation procedure based on signed parity signals. 
Therefore, calibration is used in two different forms.

In the displacement stage, the ideal guarantee says that every heavy $Z$- or $X$-basis displacement appears with probability on the order of $\epsilon^2/\Lambda^2$.
This probability is already the signal scale for detecting a coefficient of size $\epsilon$. 
Hence, an uncalibrated error that creates or removes displacement labels at the same scale would be indistinguishable from a true Hamiltonian signal. 
For this reason, we assume a calibrated bound on the full recorded displacement label.
Namely, there is a known number $\xi_1\ge0$ such that, for every block $j$, every history transcript on previous preparation and measurement records in the block so far (denoted as $\mathcal F_{j-1}$), and every nonzero displacement label,
\begin{align}
\Pr[\widetilde D_{Z,j}=d\mid \mathcal F_{j-1}]\ge\Pr[D_Z=d]-\xi_1,\quad
\Pr[\widetilde E_{X,j}=e\mid \mathcal F_{j-1}]\ge\Pr[E_X=e]-\xi_1 .
\end{align}
Here $D_Z$ and $E_x$ denote the ideal displacement labels, and $\widetilde D_{Z,j}$ and $\widetilde E_{X,j}$ denote the calibrated recorded labels in the presence of SPAM error. 
This condition is deliberately stated for the full $n$-qubit displacement record, thus any local SPAM noise calibration must first be converted into this full-record guarantee.
As mentioned earlier, we have $\Pr[D_Z=d]$ and $\Pr[E_x=e]$ of the scaling $\epsilon^2/\Lambda^2$.
Therefore, as long as we are guaranteed that $\zeta_1\lesssim \epsilon^2/\Lambda^2$, the calibrated record $\Pr[\widetilde D_{Z,j}=d\mid \mathcal F_{j-1}]$ and $\Pr[\widetilde E_{X,j}=e\mid \mathcal F_{j-1}]$ are still of the same scaling as ideal ones.

The coefficient-estimation stage is more direct. 
We use the following calibrated local depolarizing model.  
For a single qubit, let the preparation and measurement error be local depolarization noise on each qubit as
\begin{align}
\mathcal P(\rho)=r_{\rm p}\rho+(1-r_{\rm p})\frac{I}{2},
\qquad
\mathcal M(\rho)=r_{\rm m}\rho+(1-r_{\rm m})\frac{I}{2},
\end{align}
where $0<r_{\rm p},r_{\rm m}\le 1$ are known calibration numbers.  
The preparation noise $\mathcal P^{\otimes n}$ is applied after the ideal product-state preparation, and the measurement noise $\mathcal M^{\otimes n}$ is applied immediately before the ideal terminal Pauli measurement. 
Under the calibrated local depolarizing SPAM error model, the intended preparation and measurement of Pauli strings $Q$ and $P$ do not change the form of the trace signal. They only multiply it by known reliability factors $r_{\rm p}^{w(Q)}r_{\rm m}^{w(P)}$, where $w(P)$ and $w(Q)$ are the weights of $P$ and $Q$, i.e., the number of non-identity qubits.
The result of this SPAM noise factor is variance.
If the candidate set is $U$, the worst reliability factor over all visible parity blocks is denoted by $\zeta_2(U)$. 
Dividing by the calibrated SPAM noise factors increases the single-shot magnitude by at most $\zeta_2(U)^{-1}$, and hence the sample complexity of the coefficient-estimation stage is multiplied by $\zeta_2(U)^{-2}$. 
When every candidate label has Pauli weight at most $k$, we have
\begin{align}
\zeta_2(U)\geq(r_{\rm p}r_{\rm m})^k.
\end{align}
Thus, for a constant $k$ and a constant calibrated SPAM noise rate $r_{\rm p},r_{\rm m}$, this is only a constant-factor overhead. 
If a $k$-local ansatz is known, one may restrict the coefficient-estimation candidate set to labels of weight at most $k$, which does not remove any true $k$-local Hamiltonian term but prevents high-weight projection false positives from increasing the SPAM error variance.

We summarize the calibrated guarantee below.

\begin{theorem}[SPAM-robustness of Algorithm~\ref{algo:upper} after calibration]\label{thm:spam_robust_main}
Let $H=\sum_{u\in V\setminus\{0\}}\lambda_u P_u$ with $\|H\|\le\Lambda$, and fix $0<\epsilon\le\Lambda$ and $\delta\in(0,1)$.  
Assume the calibrated displacement condition satisfies $\xi_1\lesssim\tfrac{\epsilon^2}{2\kappa_0\Lambda^2}$.
Run the SPAM-robust projection stage with failure probability $\delta/2$ and form $\wh\cU\coloneqq\bigl((\wh\Pi_x\cup\{0\})\times(\wh\Pi_z\cup\{0\})\bigr)\setminus\{0\}$, and then run the SPAM-robust coefficient estimation stage on $\wh\cU$ with target accuracy $\epsilon/6$ and failure probability $\delta/2$.  
Finally, we set $\wh S_\epsilon\coloneqq\left\{u\in\wh\cU:|\wh\lambda_{\wh\cU}(u)|>\frac{\epsilon}{2}\right\}$, and $\wh H_\epsilon\coloneqq\sum_{u\in\wh S_\epsilon}\wh\lambda_{\wh\cU}(u)P_u.$
Then, with probability at least $1-\delta$,
\begin{align}
S_\epsilon\subseteq\wh S_\epsilon\subseteq\{u\in V\setminus\{0\}:|\lambda_u|>\epsilon/3\},\qquad
\max_{u\in S_\epsilon}|\wh\lambda_{\wh\cU}(u)-\lambda_u|\le\epsilon/6 .
\end{align}
For every realized candidate set $\wh\cU$, the following pathwise scheduled-time bound holds:
\begin{align}
T\le C\left[\frac{1}{\Lambda(\rho_\epsilon-\xi_1)}\log\frac{16\Lambda^2}{\epsilon^2}\log\frac8\delta+\frac{\Lambda}{\epsilon^2\zeta_2(\wh\cU)^2}\log(4\max\{|\wh\cU|,1\})\log\frac4\delta\right].
\end{align}
In particular, if every label in the realized candidate set $\wh\cU$ has Pauli weight at most $k$, then
\begin{align}
T=\widetilde O\left(\frac{\Lambda}{\epsilon^2}\left(1+(r_{\rm p}r_{\rm m})^{-2k}\right)\right).
\end{align}
\end{theorem}

In Appendix~\ref{app:spam}, we give an analysis of the performance of the algorithm against calibrated-SPAM noise, and the detailed performance analysis.

\section{Lower Bounds}\label{sec:lower}
Here, we show that our protocol is optimal over all possible control-free protocols. 
We will illustrate this by proving two lower bounds. 

First, the complexity of learning any Hamiltonian without control or structured prior knowledge will inevitably exceed the SQL~\cite{dutkiewicz2024advantage}. 
Consider the Hamiltonian family $H_\theta=\Lambda(\cos\theta Z+\sin\theta X)$ parametrized by $\theta$. 
Near $\theta=0$, the coefficient of $X$ is $\Lambda\sin\theta\approx\Lambda\theta$. 
To distinguish two coefficients by $\epsilon$, one must distinguish two parameters separated by $\epsilon/\Lambda$. 
In a controll-free experiment with probe length $t$, the Fisher information obtained scales at most linearly in $\Lambda t$. 
This means that if we want to estimate the coefficient of $X$, it takes at least $\Omega (\Lambda/\epsilon^2)$ time.

\begin{lemma}[SQL lower bound]\label{lem:sql-lower}
For every $0<\epsilon\le\Lambda/4$, any control-free protocol with deterministic total evolution time $T_{\mathrm{tot}}$ that estimates the $X$-coefficient of $H_\theta=\Lambda(\cos\theta Z+\sin\theta X)$, at $|\theta|\le\pi/6$, to accuracy $\epsilon$ with a success probability $\ge 2/3$ must satisfy $T_{\mathrm{tot}}=\Omega (\Lambda/\epsilon^2)$.
\end{lemma}

Proving a tighter lower bound is nontrivial because proofs using information theory tools usually only control how fast one can estimate a single coefficient, and can hardly constrain the strategy that learns many support locations in parallel. 
To improve the learning difficulty, we designed a hard family of $M$-sparse Hamiltonians with $M$ pair-wise anticommuting Pauli observables. 
According to Proposition 9 of Ref.~\cite{hrubevs2016families}, $M\leq 2n+1$.
We index each Pauli string in the Hamiltonian by $Q_{a,x_a}$ randomly and secretly chosen with $a=1,2,...,M$.
The Hamiltonian is written as
\begin{align}
H_{x,\sigma}=\frac{\Lambda}{\sqrt M}\sum_{a=1}^M (-1)^{\sigma_a}Q_{a,x_a},
\end{align}
where $\sigma_a\in\{0,1\}$ is a random sign used to hide global interference between blocks. 
Hence, we have $H_{x,\sigma}^2=\Lambda^2 I$ and $\|H_{x,\sigma}\|=\Lambda$. 
Let $M=\Theta(\Lambda^2/\epsilon^2)$ so that each non-zero coefficient equal to $\Lambda/\sqrt M$, which is still larger than $\epsilon$. 
Therefore, every non-zero term is detectable for a Hamiltonian learning algorithm. 
Analyzing the experiment after averaging over the uniformly random choices of signs $\sigma$, we find that any control-free experiment of duration $\tau$ either carries no information about $x$ with probability $\cos^2(\Lambda\tau)$, or provides information about only one uniformly random block with probability $\sin^2(\Lambda\tau)$. 
Since $\sin^2(\Lambda\tau)\le \Lambda\tau$, a probe of length $\tau$ can collect at most $\Lambda\tau$ bits of block information. 
Therefore, the learner faces a coupon-collector problem, where there are $M$ hidden blocks to discover, and each informative shot reveals at most one random block.  
Seeing all $M$ blocks requires hitting $\Omega(M\log M)$ informative blocks. 
After simple deductions, we can get a lower bound on total evolution time of $T_{\mathrm{tot}}=\Omega\left(\frac{\Lambda}{\epsilon^2} \log\frac{\Lambda}{\epsilon}\right)$.

\begin{theorem}[Unified lower bound]\label{thm:unified-log}
Given accuracy demand $\epsilon$ and norm bound $\Lambda$ with $0<\epsilon\le\Lambda/16$ and $\Lambda^2/\epsilon^2\leq O(n)$, any control-free protocol that solves Problem~\ref{prob:sparse_ham_recover} with probability $\ge 2/3$ must have total evolution time
\begin{align}
T_{\mathrm{tot}}\ \ge\Omega\left(\frac{\Lambda}{\epsilon^2}\log\frac{\Lambda}{\epsilon}\right).
\end{align}
\end{theorem}

\section{Numerical simulations}\label{sec:numeric}
\begin{figure}[H]
\centering
\includegraphics[width=0.98\linewidth]{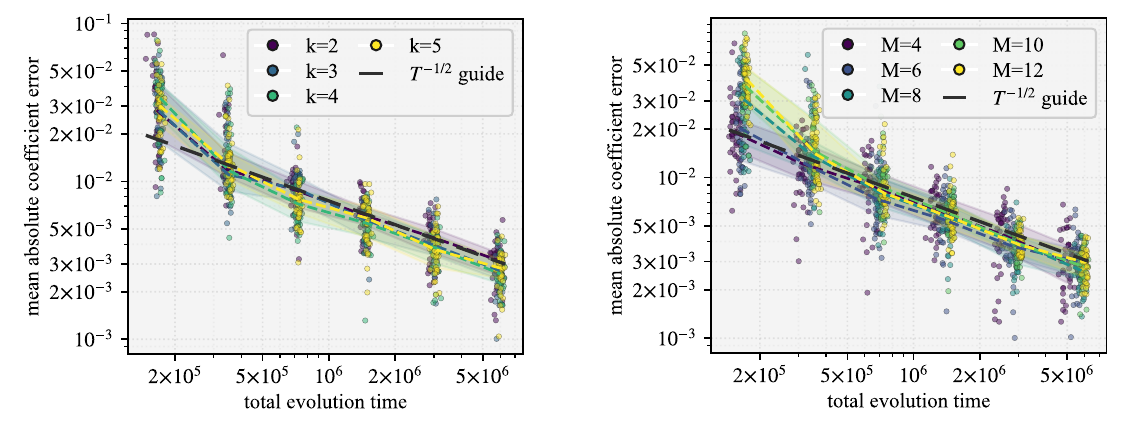}
\caption{Mean absolute coefficient error versus total evolution time for Algorithm~\ref{algo:upper} in learning random ansatz-free sparse Hamiltonians (left) and random $k$-local Hamiltonians (right). 
The black dashed line is a $T^{-1/2}$ guide.}
\label{fig:numerics-end-to-end}
\end{figure}

Here, we validate the complete protocol in Algorithm~\ref{algo:upper} by numerical simulations. 
Each data point runs the same two-stage reconstruction pipeline as Algorithm~\ref{algo:upper}: the same-Pauli displacement sieve first constructs the candidate set $\wh\cU_{\rm proj}$, and the cross-Pauli parity estimator then estimates all coefficients in that candidate set. 
The plotted error metric is the mean absolute coefficient error given a support set $S$ defined as
\begin{align}
\mathrm{MAE}_S=\frac{1}{|S|}\sum_{u\in S}|\widehat\lambda_u-\lambda_u|,
\end{align}
where missed true support elements are counted with $\widehat\lambda_u=0$. 
It penalizes both missed support and coefficient-estimation error.

For each Hamiltonian instance, we normalize $\|H\|\le \Lambda$ with $\Lambda=1$, run Algorithm~\ref{algo:upper} at a fixed detection scale $\epsilon$, and vary the sampling budget through the constants multiplying the projection and parity-estimation block sizes. 
For the unrestricted sparse panel, we use $n=5$, $M\in\{4,6,8,10,12\}$, six budget values, and $40$ independent trials per $(M,\mathrm{budget})$ pair, for a total of $5\times6\times40=1200$ complete reconstructions. 
For the exact-locality panel, we use $n=6$, $k\in\{2,3,4,5\}$, fixed sparsity $M_{\rm loc}=8$, the same six budget values, and $40$ independent trials per $(k,\mathrm{budget})$ pair, for a total of $4\times6\times40=960$ complete reconstructions. 
The observed error follows the expected $T^{-1/2}$ trend for both unrestricted sparse and $k$-local data-generating Hamiltonians.

\section{Discussion}\label{sec:conclusion}

In this work, we show that ansatz-free Hamiltonian learning can be performed in an in situ regime: our protocol is control-free, ancilla-free, and uses only Pauli product state preparation and measurement, while achieving the optimal total evolution time $\Theta(\Lambda/\epsilon^2)$, which is essentially the same total evolution time scale as SQL metrology, despite the additional burden of support recovery.
However, our protocol still requires probe times with characteristic scale $O(1/\Lambda)$ even in the discretized setting. 
It is natural to ask whether one can further relax this requirement to a constant time resolution, or conversely prove that some form of short-time resolution is fundamentally necessary for control-free Hamiltonian learning, analogous to the resolution barrier recently observed in ansatz-free Lindbladian learning~\cite{ivashkov2026ansatz}.
In addition, our lower bound is stated in terms of $\Lambda$ and $\epsilon$.
It would be very interesting to determine whether there is a lower bound that depends explicitly on the sparsity $M$ under any settings to decide the optimality of the previous algorithms along this line~\cite{hu2025ansatz,sinha2025improved}.


\section*{Acknowledgments}
We thank Zhiding Liang, Nikita Romanov, and Sisi Zhou for valuable feedback on the manuscript.
We thank Richard Allen, Sitan Chen, Senrui Chen, Hong-Ye Hu, Hsin-Yuan Huang, and Muzhou Ma for helpful discussions during this work.


{\small
\bibliographystyle{unsrt}
\bibliography{control_free_ham_lind_learning_ref}    

\begin{thebibliography}{10}
\bibitem{whittaker1915functions}
E.~T. Whittaker.
\newblock Xviii. on the functions which are represented by the expansions of the interpolation-theory.
\newblock {\em Proceedings of the Royal Society of Edinburgh}, 35:181--194, 1915.
\end{thebibliography}
}

\appendix
\section{Extended Preliminaries}\label{app:prelim}
\subsection{Hamiltonian and time evolution}
Here, we provide more primary results on Hamiltonian and time evolution.
In physics, the Hamiltonian describes the dynamics of a system. 
In quantum mechanics, it is a Hermitian operator that determines how a state evolves over time. 
The evolution is governed by the Schr\"{o}dinger equation
\begin{align}
i\hbar \frac{\partial}{\partial t}|\Psi(t)\rangle = \hat H|\Psi(t)\rangle.
\end{align}
When $H$ is time-independent and we set $\hbar=1$, the solution is
\begin{align}
|\Psi(t)\rangle = e^{-iHt}|\Psi(0)\rangle.
\end{align}
For an $n$-qubit system, $H$ is a $2^n\times2^n$ Hermitian matrix. 
We measure the size of the Hamiltonian by its operator norm
\begin{align}
\|H\|\coloneqq \sup_{\|\psi\|_2=1}\|H|\psi\rangle\|.
\end{align}
Since $H$ is Hermitian, $\norm{H}$ is simply the absolute value of the largest eigenvalue of
$H$ as $ \|H\|=\max_j |E_j|$, where $\{E_j\}$ are the eigenvalues of $H$.

We also use the Hilbert-Schmidt (Frobenius) norm for operators, which is defined as
\begin{align}
\|A\|_{\HS}^2\coloneqq \Tr(A^\dagger A).
\end{align}
Furthermore, for any Hamiltonian on an $n$-qubit system, the operator norm and the Hilbert-Schmidt norms satisfy $\|H\|^2\le\|H\|_{HS}^2 \le 2^{n}\|H\|^2$.

The above description on the evolution of the state $|\Psi(t)\rangle$ over time is known as the Schr\"{o}dinger picture. 
An alternative approach to describe a Hamiltonian dynamics is to keep the state constant and transfer the time dependency to the observable being measured, which is known as the Heisenberg picture. 
\begin{definition}[Heisenberg picture]Let $H=H^\dagger$ be a time-independent Hamiltonian and let $U(t)\coloneqq e^{-iHt}$ be the corresponding unitary time-evolution operator. 
For an observable $A=A^\dagger$, its Heisenberg evolution is defined by
\begin{align}
A_H(t)\coloneqq U(t)^\dagger A U(t)= e^{iHt} A e^{-iHt}.
\end{align}
\end{definition}
Let $|\Psi(t)\rangle = U(t)|\Psi(0)\rangle$. Then
\begin{align}
\langle \Psi(t)|A|\Psi(t)\rangle=\langle \Psi(0)|U(t)^\dagger A U(t)|\Psi(0)\rangle=\langle \Psi(0)|A_H(t)|\Psi(0)\rangle .
\end{align}
Therefore, measuring $A$ at time $t$ in the Schr\"odinger picture is equivalent to measuring the $A_H(t)$ on the initial state.
For example, in our learning protocol, our direct learning object is the Pauli operator.
We thus write the observed time-dependent signal as $P_H(t) = e^{iHt} P e^{-iHt}$, and we omit $H$ and write it as $P(t)$ when $H$ is clear from the context.

\subsection{Symplectic algebra}
Here, we provide an extended introduction to Pauli observables and symplectic algebra.
Note that the four single-qubit Pauli operators $\{I, X, Y, Z\}$, along with their phases, form a group: 
\begin{align}
\mathcal{P}_1 = \{ \pm I, \pm iI, \pm X, \pm iX, \pm Y, \pm iY, \pm Z, \pm iZ \}
\end{align}
where multiplication in this group is matrix multiplication. 
We call $\mathcal{P}_1$ the single-qubit Pauli group. 
Accordingly, by replacing the group elements with tensor products of single-qubit Pauli operators, we obtain the multi-qubit Pauli group $\mathcal{P}_n$. 
Due to the presence of phase, the Pauli group $\mathcal{P}_n$ is not an abelian group. 
To simplify the operation, we can divide it by the center of the Pauli group $\{ +I, -I, +iI, -iI \}$, thus obtaining the quotient group
\begin{align}
V \coloneqq \mathcal{P}_n / \{\pm I, \pm iI\} \cong (\mathbb{F}_2^{2n}, \oplus).
\end{align}
$V$ is a set of $2n$-bit binary vectors, and the corresponding group operation is the bitwise XOR addition. 
We call a vector $u$ in the quotient group a label of the corresponding $P_u$ in the original Pauli group.
Concretely, we represent each vector as
\begin{align}
u=(x,z)\in \mathbb F_2^n\times \mathbb F_2^n,
\end{align}
where the two binary strings $x$ and $z$ record the $X$- and $Z$-components of the Pauli operator. 
Because $Y = iXZ$, the effect of $Y$ can be decomposed into the sequential effects of $X$ and $Z$. 
On qubit $j$, the pair $(x_j,z_j)$ specifies
\begin{align}
(0,0)\leftrightarrow I,\qquad
(1,0)\leftrightarrow X,\qquad
(0,1)\leftrightarrow Z,\qquad
(1,1)\leftrightarrow Y,
\end{align}
Therefore, this label fully preserves the tensor product pattern of the Pauli operator while discarding its global phase. 

Symplectic labels allow us to represent commutation and anti-commutation without multiplying matrices directly. 
We define $\omega(u,v)\coloneqq x\cdot z' + z\cdot x' \pmod 2$, where the dot products are over $\mathbb F_2$. 
This number records the parity of the local anti-commutations between the two Pauli strings: if $\omega(u,v)=0$, then $P_u$ and $P_v$ commute; if $\omega(u,v)=1$, then they anticommute. 
For each $u\in V$, we can also define the Walsh character
\begin{align}
\chi_u(v)\coloneqq (-1)^{\omega(u,v)}\in\{\pm1\},
\end{align}
and thus
\begin{align}
\chi_u(v)=+1\iff [P_u,P_v]=0,
\qquad
\chi_u(v)=-1\iff \{P_u,P_v\}=0.
\end{align}
Hence, we have
\begin{align}
\mathbbm{1}[\{P_u,P_v\}=0]=\frac{1-\chi_u(v)}{2}.
\end{align}
Formally, under the symplectic notation, we have the following fact.
\begin{fact}[Pauli product rule]\label{fac:pauli-product-app}
For every $u,v\in V$, there exists a phase $\zeta(u,v)\in\{\pm1,\pm i\}$ such
that $P_uP_v=\zeta(u,v)P_{u+v}$. If $\omega(u,v)=0$, then
$\zeta(u,v)\in\{\pm1\}$. If $\omega(u,v)=1$, then $\zeta(u,v)=is(u,v)$ for
some sign $s(u,v)\in\{\pm1\}$, and
$[P_u,P_v]=2is(u,v)P_{u+v}$.
\end{fact}

Since $P_u$ can be a Pauli string formed by tensor products, we can also expand it into a qubit-by-qubit form $P_u=\bigotimes_{i=1}^n P_{u_i}$, for $P_{u_i}\in\{I,X,Y,Z\}$.
The addition $u+v$ on it is still the bitwise XOR addition over $\mathbb F_2$, and $(u+v)_i$ denotes the $i$th local Pauli letter.
For each fixed label $u$, one can assign a $\{\pm1\}$ value to every $v\in V$ depending on their commutation relation. 
When these values are summed over a subspace, the result is either a sum of all $1$ or cancellation out completely to $0$.
Quantitatively, we have the following lemma.

\begin{lemma}[Orthogonality of symplectic characters]\label{lem:char-sum-app}
If $K\le V$ has dimension $r$, then for every $u\in V$,
\begin{align}
\sum_{v\in K}\chi_u(v)=
\begin{cases}
2^r,& u\in K^\perp,\\
0,& u\notin K^\perp.
\end{cases}
\end{align}
\end{lemma}

\begin{proof}
If $u\in K^\perp$, then $\chi_u(v)=1$ for all $v\in K$, so the sum is $|K|=2^r$. 
If $u\notin K^\perp$, choose $v_0\in K$ such that $\omega(u,v_0)=1$. 
Then the map $v\mapsto v+v_0$ is a bijection of $K$, and
\begin{align}
\chi_u(v+v_0)=\chi_u(v)\chi_u(v_0)=-\chi_u(v),
\end{align}
resulting in the sum canceling out in pairs.
\end{proof}

\subsection{Proof of Proposition~\ref{prop:trace-main}}\label{app:trace}

Let $\rho_\psi\coloneqq |\psi\rangle\langle\psi|$. 
Consider the conditional probability of $m$ given $\psi$
\begin{align}
\E[m\mid\psi]=\Tr(\rho_\psi P(t)).
\end{align}
Therefore, we have
\begin{align}
\E[qm]=2^{-n}\sum_\psi q(\psi)\Tr(\rho_\psi P(t))=2^{-n}\Tr\left(P(t)\sum_\psi q(\psi)\rho_\psi\right).
\end{align}
Since $\{|\psi\rangle\}$ is an eigenbasis of $Q$ and $Q|\psi\rangle=q(\psi)|\psi\rangle$, we have $Q=\sum_\psi q(\psi)\rho_\psi$.
Namely
\begin{align}
\E[qm]=2^{-n}\Tr(P(t)Q).
\end{align}

\section{Same-Pauli Measurement}\label{app:quadratic}

Same-Pauli measurements require obtaining the short-time second-order derivative of the auto-correlation function in Eq.~\eqref{eq:auto-correlation} at the initial time. 
We will now show how to obtain the support of the Hamiltonian using these measurements.

\subsection{Complementarity of anti-commuting probes}

In this subsection, we will prove the following three points: 
\begin{enumerate}
\item The short-time second-order derivative of the auto-correlation function is proportional to the sum of the squared coefficients of all anti-commutative terms with the probe operator in the Hamiltonian.
\item The raw signal obtained from the same-Pauli measurement can be written as a linear sum of the squares of the coefficients in the Hamiltonian.
\item Any single probe suffers from invisibility of some Hamiltonian terms. 
As a result, we need complementary bases to extract all Pauli terms in the Hamiltonian.
\end{enumerate}

\begin{proposition}[Same-Pauli signal]\label{prop:curvature}
For every Hermitian Pauli string $P$ and Pauli Hamiltonian $H=\sum_{u\in V}\lambda_u P_u$, we have
\begin{align}
C_P''(0)=-4\sum_{u:\{P_u,P\}=0}\lambda_u^2.
\end{align}
\end{proposition}
\begin{proof}
Note that $\tfrac{d}{dt}P(t)=i[H,P(t)]$ and $P''(0)=-[H,[H,P]]$, we have
\begin{align}
C_P''(0)=2^{-n}\Tr(P''(0)P)=-2^{-n}\Tr([H,[H,P]]P).
\end{align}
Write $X\coloneqq [H,P]$, we have
\begin{align}
\begin{split}
\Tr([H,[H,P]]P)&=\Tr((HX-XH)P)=
\Tr(HXP)-\Tr(XHP)=
\Tr(XPH)-\Tr(XHP)\\
&=\Tr(X(PH-HP))=\Tr(X[P,H])=-\Tr(X[H,P])\\
&=-\Tr(X^2).
\end{split}
\end{align}
Because $H$ and $P$ are Hermitian, $X=[H,P]$ is anti-Hermitian, so $X^\dagger=-X$ and therefore
\begin{align}
-\Tr(X^2)=\Tr(X^\dagger X)=\|X\|_{HS}^2.
\end{align}
Thus
\begin{align}\label{eq:same-pauli-commutator}
2^{-n}\Tr([H,[H,P]]P)=2^{-n}\|[H,P]\|_{HS}^2.
\end{align}

As we can write $[H,P]$ as
\begin{align}
[H,P]=\sum_{u\in S}\lambda_u[P_u,P]=\sum_{u:\{P_u,P\}=0}2\lambda_uP_uP.
\end{align}
For $u\neq u'$, the Hilbert-Schmidt inner product of the corresponding terms is
\begin{align}
\Tr((P_uP)^\dagger(P_{u'}P))=\Tr(PP_uP_{u'}P)=\Tr(P_uP_{u'}),
\end{align}
where we used fact that $P=P^\dagger=P^{-1}$ and cyclicity of trace. 
Since $P_uP_{u'}$ is a phase times a non-identity Pauli string when $u\neq u'$, it is traceless. 
For $u=u'$, the trace equals $\Tr(I)=2^n$. 
Hence the family $\{P_uP:\{P_u,P\}=0\}$ is orthogonal in Hilbert-Schmidt inner product. 
Therefore
\begin{align}
\|[H,P]\|_{HS}^2=4\sum_{u:\{P_u,P\}=0}\lambda_u^2\|P_uP\|_{HS}^2=4\cdot 2^n\sum_{u:\{P_u,P\}=0}\lambda_u^2,
\end{align}
After substituting the above into Eq.~\eqref{eq:same-pauli-commutator}, we have 
\begin{align}
C_P''(0)=-4\sum_{u:\{P_u,P\}=0}\lambda_u^2.
\end{align}
\end{proof}

Based on the above principle, we can extend the representation of the same-Pauli measurement signal to the entire $V$ and give the form of a linear sum. 
We denote
\begin{align}
x_u\coloneqq \lambda_u^2\ge 0,
\qquad\text{and}\qquad
W_{\mathrm{tot}}\coloneqq \sum_{u\in V}x_u=2^{-n}\Tr(H^2).
\end{align}
For each $v\in V$, set $P_v$ and define
\begin{align}
y(v)\coloneqq \sum_{u\in V}\mathbbm{1}[\{P_u,P_v\}=0]x_u=-\frac14 C_{P_v}''(0).
\end{align}
according to Proposition~\ref{prop:curvature}.
We also introduce the notation
\begin{align}
g(v)\coloneqq W_{\mathrm{tot}}-2y(v)=\sum_{u\in V}x_u\chi_u(v).
\end{align}
In this way, each probe direction is assigned a number, and the Walsh characters form the coordinate system that converts those numbers back into squared Pauli coefficients.

\begin{proposition}[Walsh inversion for results from auto-correlation function values]\label{prop:even-duality}
For every nonzero $u\in V$,
\begin{align}
x_u=2^{-2n}\sum_{v\in V}g(v)\chi_u(v)=\frac{1}{2^{2n+1}}\sum_{v\in V} C_{P_v}''(0)\chi_u(v).
\end{align}
\end{proposition}
\begin{proof}
As we have shown in the main text that
\begin{align}
2^{-2n}\sum_{v\in V}\chi_u(v)\chi_{u'}(v)=\mathbbm{1}\{u=u'\},
\end{align}
and we recall that
\begin{align}
g(v)=\sum_{u'\in V}x_{u'}\chi_{u'}(v),
\end{align}
we obtain
\begin{align}
\begin{split}
2^{-2n}\sum_{v\in V}g(v)\chi_u(v)&=2^{-2n}\sum_{v\in V}\sum_{u'\in V}x_{u'}\chi_{u'}(v)\chi_u(v)\\
&=\sum_{u'\in V}x_{u'}
\left(2^{-2n}\sum_{v\in V}\chi_{u'}(v)\chi_u(v)\right)\\
&=x_u.
\end{split}
\end{align}
This proves the first step. 
For the second step, we use
\begin{align}
y(v)=-\frac14 C_{P_v}''(0),
\qquad
g(v)=W_{\mathrm{tot}}-2y(v)=W_{\mathrm{tot}}+\frac12 C_{P_v}''(0).
\end{align}
Therefore, we have
\begin{align}
2^{-2n}\sum_{v\in V}g(v)\chi_u(v)=2^{-2n}W_{\mathrm{tot}}\sum_{v\in V}\chi_u(v)+\frac{1}{2^{2n+1}}\sum_{v\in V}C_{P_v}''(0)\chi_u(v).
\end{align}
When $u\neq0$, the character $\chi_u$ is nontrivial, so $\sum_{v\in V}\chi_u(v)=0$.
This gives the claimed result.
\end{proof}
Proposition~\ref{prop:even-duality} indicates there is no information loss in the full field. 
However, one can never query all $4^n$ Pauli directions at a same time. 
The following theorem explains what a whole subspace of probes sees and what it misses.

\begin{lemma}\label{lem:complementarity}
Let $K\le V$ have dimension $r$, and define $\cK(K)\coloneqq \frac{1}{2^{r+1}}\sum_{v\in K}|C_{P_v}''(0)|$ and $\cS(K)\coloneqq \sum_{u\in K^\perp}x_u$.
Then
\begin{align}
\cK(K)=\sum_{u\notin K^\perp}x_u,
\qquad
\cK(K)+\cS(K)=W_{\mathrm{tot}}.
\end{align}
\end{lemma}
\begin{proof}
By Proposition~\ref{prop:curvature}, each $C_{P_v}''(0)$ is non-positive, and hence $|C_{P_v}''(0)|=4y(v)$.
Using $y(v)=\sum_{u\in V}\frac{1-\chi_u(v)}{2}x_u$, we get
\begin{align}
\sum_{v\in K}|C_{P_v}''(0)|=4\sum_{v\in K}y(v)=2\sum_{u\in V}x_u\left(|K|-\sum_{v\in K}\chi_u(v)\right).
\end{align}
By Lemma~\ref{lem:char-sum-app}, the character sum equals $2^r$ if $u\in K^\perp$ and equals $0$ otherwise. 
Thus all $u\in K^\perp$ contribute zero, while each $u\notin K^\perp$ contributes $2^{r+1}x_u$. Dividing by $2^{r+1}$ gives
\begin{align}
\cK(K)=\sum_{u\notin K^\perp}x_u.
\end{align}
We then obtain the second step in the claimed result by the fact that
\begin{align}
W_{\mathrm{tot}}=\sum_{u\notin K^\perp}x_u+\sum_{u\in K^\perp}x_u.
\end{align}
\end{proof}
In short, a family of same-Pauli probes cannot see the coefficients in $K^\perp$, namely the terms commuting with all probes in the family.
This is why the protocol uses complementary product bases: the $Z$ basis reveals the $X$-projection of a Pauli label, while the $X$ basis reveals the $Z$-projection.

\subsection{Proof of Proposition~\ref{prop:even-sample-main}}
By assuming $\norm{H}\leq \Lambda$, all correlation functions become band-limited. 
Here, we explain why a fixed kernel can combine finite discrete values to obtain a derivative at zero. 
We first show that every same-Pauli auto-correlation function is a bounded-frequency signal that we can recover its second-order derivative from its finite-time values.

\begin{proposition}[A spectral representation of the auto-correlation function]\label{prop:even-spectrum}
Assume $\|H\|\le \Lambda$. For each $v\in V$, there exists a symmetric probability measure $\mu_v$ supported on $[-2\Lambda,2\Lambda]$ such that
\begin{align}
C_{P_v}(t)=\int_{-2\Lambda}^{2\Lambda}\cos(\omega t)d\mu_v(\omega).
\end{align}
Consequently,
\begin{align}
C_{P_v}''(0)=-\int_{-2\Lambda}^{2\Lambda}\omega^2d\mu_v(\omega).
\end{align}
\end{proposition}
\begin{proof}
We defined the normalized Hilbert-Schmidt inner product $\langle A,B\rangle_{\HS}\coloneqq 2^{-n}\Tr(A^\dagger B)$.
The commutator map $\cA_H\coloneqq \ad_H$ is self-adjoint on this Hilbert space because
\begin{align}
\langle A,[H,B]\rangle_{\HS}
=2^{-n}\Tr((A^\dagger H-HA^\dagger)B)
=\langle [H,A],B\rangle_{\HS}.
\end{align}
As $H|e_j\rangle=E_j|e_j\rangle$, we have
\begin{align}
\cA_H(|e_j\rangle\langle e_k|)=(E_j-E_k)|e_j\rangle\langle e_k|.
\end{align}
Since $|E_j|\le\Lambda$, the spectrum of $\cA_H$ lies in $[-2\Lambda,2\Lambda]$.

Let $E(d\omega)$ be the spectral measure of $\cA_H$, and define
\begin{align}
\mu_v(\Omega)\coloneqq \langle P_v,E(\Omega)P_v\rangle_{\HS}.
\end{align}
This is a probability measure supported on $[-2\Lambda,2\Lambda]$. 
Moreover,
\begin{align}
C_{P_v}(t)=\langle P_v,e^{it\cA_H}P_v\rangle_{\HS}
=\int_{-2\Lambda}^{2\Lambda}e^{i\omega t}d\mu_v(\omega).
\end{align}
The function $C_{P_v}(t)$ is real and even. 
We may average the frequency measure with its reflection $\omega$ and $-\omega$ and get:
\begin{align}
\frac12\bigl(C_{P_v}(t)+C_{P_v}(-t)\bigr)=C_{P_v}(t).
\end{align}
After this averaging, the measure is symmetric, so the sine terms cancel and only the cosine part remains:
\begin{align}
C_{P_v}(t)=\int_{-2\Lambda}^{2\Lambda}\cos(\omega t)d\mu_v(\omega).
\end{align}
Since the measure is supported on a bounded interval, we can differentiate inside the integral, giving
\begin{align}
C_{P_v}''(0)=-\int_{-2\Lambda}^{2\Lambda}\omega^2d\mu_v(\omega).
\end{align}
\end{proof}
In other words, we only need to choose $K_\Lambda$ such that on the entire frequency range where the signal can live, and any convolution against $K_\Lambda$ is equal to taking a second-order derivative at zero.
To build such a kernel with good decay, we fix a smooth
cutoff $\varphi\in C_c^\infty(\mathbb R)$ satisfying
\begin{align}
\varphi(\xi)=1\quad(|\xi|\le1),
\qquad
\varphi(\xi)=0\quad(|\xi|\ge2).
\end{align}
We set $B_\Lambda\coloneqq 2\Lambda$ and $\Phi_\Lambda(\omega)\coloneqq \omega^2\varphi(\omega/B_\Lambda)$, and define
\begin{align}
K_\Lambda(t)\coloneqq -\frac{1}{2\pi}\int_{\mathbb R}\Phi_\Lambda(\omega)e^{i\omega t}d\omega .
\end{align}
The cutoff does not change anything on the true spectral window as $\varphi(\omega/B_\Lambda)=1$ whenever $|\omega|\le B_\Lambda=2\Lambda$.
We thus have
\begin{align}
\Phi_\Lambda(\omega)=\omega^2
\qquad
 (|\omega|\le2\Lambda).
\end{align}
Beyond this range, the $\Phi_\Lambda(\omega)$ gently fades to zero rather than stopping abruptly. 
This is just a mathematical technique to let the integral behave well. 
Any construct that meets the requirements will not affect the complexity of the protocol in terms of order. 
As a simple example, we may define:
\begin{align}
\rho(s)\coloneqq 
\begin{cases}
e^{-1/s},& s>0,\\
0,& s\le 0,
\end{cases}
\qquad
\theta(s)\coloneqq \frac{\rho(s)}{\rho(s)+\rho(1-s)} .
\end{align}
Then $\theta\in C^\infty(\mathbb R)$, $\theta(s)=0$ for $s\le0$, and $\theta(s)=1$ for $s\ge1$. 
We then define the even cutoff:
\begin{align}
\varphi(\xi)\coloneqq \theta\left(\frac{4-\xi^2}{3}\right).
\end{align}
It is easy to verify that the example kernel constructed in this way meets the above requirements. The next lemma illustrates that the effect on the bounded frequency window is equivalent to taking the second-order derivative at zero.
\begin{lemma}[Even kernel for the second-order derivative of the auto-correlation function]\label{lem:even-kernel}
If $f(t)=\int_{-2\Lambda}^{2\Lambda}\cos(\omega t)d\nu(\omega)$ for a finite symmetric signed measure $\nu$, then
\begin{align}
f''(0)=\int_{\mathbb R}K_\Lambda(t)f(t)dt=2\int_0^\infty K_\Lambda(t)f(t)dt.
\end{align}
\end{lemma}
\begin{proof}
Since $K_\Lambda(t)$ is integrable over the entire $\mathbb R$ and $f$ is bounded, we can interchange the order of inner and outer integrals according to Fubini's theorem. 
Using the definition of $K_\Lambda$, we have
\begin{align}
\int_{\mathbb R}K_\Lambda(t)f(t)dt=\int_{-2\Lambda}^{2\Lambda}\left(\int_{\mathbb R}K_\Lambda(t)e^{i\omega t}dt\right)d\nu(\omega).
\end{align}
The inner transform equals $-\Phi_\Lambda(\omega)$, which is $-\omega^2$ on $[-2\Lambda,2\Lambda]$. 
Hence, we have
\begin{align}
\int_{\mathbb R}K_\Lambda(t)f(t)dt=-\int_{-2\Lambda}^{2\Lambda}\omega^2d\nu(\omega)=f''(0).
\end{align}
The second equality follows because both $K_\Lambda$ and $f$ are even.
\end{proof}

Now we have a method to express the derivative at zero as a weighted average of the values of $C_{P_v}(t)$ at positive time points. 
The remaining problem is that $K_\Lambda(t)$ is a signed function, not a probability density. 
Therefore, we need to sample time from the absolute value $|K_\Lambda(t)|$, and then put the sign of $K_\Lambda(t)$ back into the estimator.
Because the cutoff $\varphi$ is fixed once and for all, the kernels $K_\Lambda$'s are just rescalings of one specific fixed kernel. 
Hence, there exist absolute constants $\kappa_0,\kappa_1>0$, depending only on
$\varphi$, such that
\begin{align}
\|K_\Lambda\|_1=4\kappa_0\Lambda^2,
\qquad
\frac{2}{\|K_\Lambda\|_1}\int_0^\infty t|K_\Lambda(t)|dt=\frac{\kappa_1}{2\Lambda}.
\end{align}
where $\|K_\Lambda\|_1\coloneqq \int_{\mathbb R}|K_\Lambda(t)|dt$ is the function $L^1$ norm. 
The first equality controls the size of one sample. 
The second equality means that the average sampled evolution time is of order $1/\Lambda$.

Since $K_\Lambda$ is even, we can define a probability density on $[0,\infty)$ by
\begin{align}
q_\Lambda(t)\coloneqq \frac{2|K_\Lambda(t)|}{\|K_\Lambda\|_1}=\frac{|K_\Lambda(t)|}{2\kappa_0\Lambda^2}.
\end{align}
Now sample $\tau\sim q_\Lambda$. 
For a same-Pauli shot with probe $P_v$ and time $\tau$, let $X_v\in\{\pm1\}$ be the product of the prepared eigenvalue and the final measurement outcome.
By Proposition~\ref{prop:trace-main}, we have
\begin{align}
\E[X_v\mid T=t]=C_{P_v}(t).
\end{align}
Define
\begin{align}
D_v\coloneqq \|K_\Lambda\|_1\operatorname{sign}(K_\Lambda(T))X_v=4\kappa_0\Lambda^2\operatorname{sign}(K_\Lambda(T))X_v,\qquad
Y_v\coloneqq -\frac14D_v.
\end{align}

We now convert the kernel into a randomized experiment to let one same-Pauli shot give an unbiased estimate of the squared coefficient as claimed in Proposition~\ref{prop:even-sample-main}.
\begin{proposition}[Sample for same-Pauli measurements]\label{prop:even-sample}
For every $v\in V$, we have
\begin{align}
\E[D_v]=C_{P_v}''(0),
\qquad
\E[Y_v]=y(v),
\qquad
|Y_v|\le \kappa_0\Lambda^2,
\qquad
\E[\tau]=\frac{\kappa_1}{2\Lambda}.
\end{align}
\end{proposition}
\begin{proof}
Conditioning on $\tau$ and using the density $q_\Lambda$,
\begin{align}
\E[D_v]=\int_0^\infty \|K_\Lambda\|_1\sign(K_\Lambda(t))C_{P_v}(t)
\frac{2|K_\Lambda(t)|}{\|K_\Lambda\|_1}dt=2\int_0^\infty K_\Lambda(t)C_{P_v}(t)dt
=C_{P_v}''(0),
\end{align}
where the last equality is from Lemma~\ref{lem:even-kernel}. 
Therefore, we have
\begin{align}
\E[Y_v]=-\frac14C_{P_v}''(0)=y(v),
\qquad
|Y_v|\le \frac14\|K_\Lambda\|_1=\kappa_0\Lambda^2,
\end{align}
and the mean-time identity is exactly the definition of $\kappa_1$.
\end{proof}

\subsection{Proof of Proposition~\ref{prop:displacements-main}}
As discussed in Lemma~\ref{lem:complementarity}, we need a pair of complementary bases to fully recover all the squared coefficient information. 
Since Pauli labels are described by the symplectic form, a natural selection is the $Z$ product basis and the $X$ product basis. 
However, a direct search over the full Pauli label space $V=\mathbb F_2^{2n}$ with $|V|=4^n$ would be too expensive. 
We therefore design a displacement sieve method on the Boolean hypercube $\mathbb F_2^n$: instead of learning a full label $u=(x\mid z)$ at once, the same-Pauli measurement first finds which $x$- and which $z$-coordinates appear in Hamiltonian terms.

In the $Z$ product basis, a basis vector is labeled by $a\in\mathbb F_2^n$. 
For a Pauli label $u=(x\mid z)$, the $x$ part records which $Z$-basis labels are flipped by $P_u$. 
Indeed, the $i$-th $Z$-basis label is flipped exactly when $P_u$ anticommutes with the corresponding $Z_i$ observable, i.e., when
\begin{align}
\omega(u,(0\mid e_i))=x_i.
\end{align}
Thus
\begin{align}
P_{(x\mid z)}|a\rangle_Z\propto |a+x\rangle_Z.
\end{align}
So if we prepare $|a\rangle_Z$, evolve, measure again in the $Z$ basis, and obtain $|b\rangle_Z$, then the observed displacement $d\coloneqq a+b$ points to Pauli terms with $x=d$.  

We will first prove that when the time is sampled from the same kernel density $q_\Lambda$ defined above, this observed displacement probability has a lower bound.

\begin{lemma}[Displacement bound]\label{lem:transition-lower}
Let $\{|a\rangle_Z\}_{a\in\F^n}$ be any orthonormal basis indexed by $\F^n$, and fix a nonzero displacement $d\in\F^n$. Define
\begin{align}
f_d(t)\coloneqq 2^{-n}\sum_{a\in\F^n}\bigl|\langle a+d|e^{-iHt}|a\rangle_Z\bigr|^2.
\end{align}
Then $f_d$ is even, nonnegative, and spectrally supported in $[-2\Lambda,2\Lambda]$. 
Moreover, we have
\begin{align}
\E_{\tau\sim q_\Lambda}[f_d(\tau)]\ge\frac{1}{2\kappa_0\Lambda^2}2^{-n}\sum_{a\in\F^n}|\langle a+d|H|a\rangle_Z|^2.
\end{align}
\end{lemma}
\begin{proof}
Let $H=\sum_j E_j|\phi_j\rangle\langle\phi_j|$ with $|E_j|\le\Lambda$. For each $a$,
\begin{align}
\langle a+d|e^{-iHt}|a\rangle=\sum_j \langle a+d|\phi_j\rangle_Z\langle\phi_j|a\rangle_Z e^{-iE_j t}.
\end{align}
The squared modulus has frequencies $E_j-E_k\in[-2\Lambda,2\Lambda]$. 
Taking the average value of $a$ preserves this spectral support and makes it nonnegative.

We next show that $f_d$ is even. 
In fact, the transition probability from $a$ to $a+d$ at time $-t$ equals the transition probability from $a+d$ to $a$ at time $t$ since $U(-t)=U(t)^\dagger$:
\begin{align}
\begin{split}
f_d(-t)&=2^{-n}\sum_a|\langle a+d|e^{iHt}|a\rangle_Z|^2\\
&=2^{-n}\sum_a|\langle a|e^{-iHt}|a+d\rangle_Z|^2
=2^{-n}\sum_b|\langle b+d|e^{-iHt}|b\rangle_Z|^2=f_d(t).
\end{split}
\end{align}
Because $d\neq0$ and $\langle a+d|a\rangle=0$, denote $\alpha_a(t)\coloneqq \langle a+d|e^{-iHt}|a\rangle$, then $\alpha_a(0)=0$ and $\alpha_a'(0)=-i\langle a+d|H|a\rangle$. 
Therefore, we have
\begin{align}
\frac{d^2}{dt^2}|\alpha_a(t)|^2\Big|_{t=0}=2|\langle a+d|H|a\rangle_Z|^2,\qquad
f_d''(0)=2^{1-n}\sum_a|\langle a+d|H|a\rangle_Z|^2.
\end{align}
Applying Lemma~\ref{lem:even-kernel} to $f_d$ gives
\begin{align}
2\int_0^\infty K_\Lambda(t)f_d(t)dt=f_d''(0).
\end{align}
Since $f_d(t)\ge0$, we have
\begin{align}
\int_0^\infty |K_\Lambda(t)|f_d(t)dt\ge\left|\int_0^\infty K_\Lambda(t)f_d(t)dt\right|=\frac12f_d''(0).
\end{align}
Using $q_\Lambda(t)=2|K_\Lambda(t)|/\|K_\Lambda\|_1$ and $\|K_\Lambda\|_1=4\kappa_0\Lambda^2$ leads to the claim.
\end{proof}

For $u=(x\mid z)$, write $x(u)=x$ and $z(u)=z$. We define the projection intensities
\begin{align}
W_x(d)\coloneqq \sum_{z\in\F^n}\lambda_{(d\mid z)}^2,
\qquad
W_z(e)\coloneqq \sum_{x\in\F^n}\lambda_{(x\mid e)}^2.
\end{align}
A $Z$-basis displacement shot samples a uniformly random computational basis state $|A\rangle_Z$, evolves for $\tau\sim q_\Lambda$, measures every qubit in the $Z$ basis, obtains $B\in\F^n$, and records $D=A+B$.
We now prove the following two propositions, which together prove Proposition~\ref{prop:displacements-main}. 

\begin{proposition}[$Z$-basis displacements]\label{prop:z-traj}
For every nonzero $d\in\F^n$, we have
\begin{align}
\Pr[D=d]\ge \frac{W_x(d)}{2\kappa_0\Lambda^2}.
\end{align}
\end{proposition}
\begin{proof}
For a fixed $t$, we have
\begin{align}
\Pr[D=d\mid T=t]=2^{-n}\sum_{a\in\F^n}|\langle a+d|e^{-iHt}|a\rangle_Z|^2.
\end{align}
Lemma~\ref{lem:transition-lower} gives
\begin{align}
\Pr[D=d]\ge\frac{1}{2\kappa_0\Lambda^2}2^{-n}\sum_a|\langle a+d|H|a\rangle_Z|^2.
\end{align}
For the canonical Pauli representative, there is a unit phase $\zeta_Z(d,z)$, independent of $a$, such that
\begin{align}
P_{(d\mid z)}|a\rangle_Z=\zeta_Z(d,z)(-1)^{z\cdot a}|a+d\rangle_Z.
\end{align}
Thus
\begin{align}
\langle a+d|H|a\rangle_Z=\sum_{z\in\F^n}\lambda_{(d\mid z)}\zeta_Z(d,z)(-1)^{z\cdot a}.
\end{align}
Averaging the squared modulus over uniform $a$ and using Walsh orthogonality on $\F^n$ gives
\begin{align}
2^{-n}\sum_a|\langle a+d|H|a\rangle_Z|^2 =\sum_{z\in\F^n}\lambda_{(d\mid z)}^2=W_x(d).
\end{align}
Substitution proves the proposition.
\end{proof}
The experiment based on the $X$ basis is the same as that based on the $Z$ basis, except that the roles of $x$ and $z$ are reversed. 
The Pauli $Z$ component flips the eigenstates of the $X$ basis, so the observed displacement is a $z$ projection.

\begin{proposition}[$X$-basis displacements]\label{prop:x-traj}
For every nonzero $e\in\F^n$, we have
\begin{align}
\Pr[E=e]\ge \frac{W_z(e)}{2\kappa_0\Lambda^2}.
\end{align}
\end{proposition}
\begin{proof}
Define the $X$-basis displacement shot exactly as above, but prepare and measure in the $X$ basis and record $E=A+B$. 
Lemma~\ref{lem:transition-lower} gives
\begin{align}
\Pr[E=e]\ge\frac{1}{2\kappa_0\Lambda^2}2^{-n}\sum_a|\langle a+e|H|a\rangle_X|^2.
\end{align}
For each fixed $(x\mid e)$ there is a unit phase $\zeta_X(x,e)$, independent of $a$, such that
\begin{align}
P_{(x\mid e)}|a\rangle_X=\zeta_X(x,e)(-1)^{x\cdot a}|a+e\rangle_X.
\end{align}
Therefore, we have
\begin{align}
\langle a+e|H|a\rangle_X=\sum_{x\in\F^n}\lambda_{(x\mid e)}\zeta_X(x,e)(-1)^{x\cdot a}.
\end{align}
Walsh orthogonality over $a$ gives
\begin{align}
2^{-n}\sum_a|\langle a+e|H|a\rangle_X|^2=\sum_{x\in\F^n}\lambda_{(x\mid e)}^2=W_z(e),
\end{align}
which proves the claim.
\end{proof}

\subsection{Proof of Lemma~\ref{lem:projection-main}}\label{app:support-thm}
According to the definitions, $W_x(d)=\sum_z\lambda_{(d\mid z)}^2$ is the total squared coefficient of all Pauli terms whose $x$-part is $d$, and  $W_z(e)$ is the total squared coefficient of all Pauli terms whose $z$-part is $e$. 
We now show that all these sums together are bounded by $\Lambda^2$, so there can be only a few coordinates $d$ or $e$ whose sum is at least $\epsilon^2$. 
\begin{proposition}[Projection counts]\label{prop:projection-counts}
One has
\begin{align}
\sum_{d\in\F^n}W_x(d)=\sum_{e\in\F^n}W_z(e)=W_{\mathrm{tot}}\le \Lambda^2.
\end{align}
Consequently, for every target detection scale $\epsilon>0$, we have
\begin{align}
|\Pi_x^{(\epsilon)}\setminus\{0\}|\le \frac{\Lambda^2}{\epsilon^2},
\qquad
|\Pi_z^{(\epsilon)}\setminus\{0\}|\le \frac{\Lambda^2}{\epsilon^2}.
\end{align}
\end{proposition}
\begin{proof}
If we sum $W_x(d)$ over all $d\in\F^n$, every squared Pauli coefficient is counted exactly once:
\begin{align}
\sum_{d\in\F^n} W_x(d)=\sum_{d\in\F^n}\sum_{z\in\F^n}\lambda_{(d\mid z)}^2=\sum_{u\in V}\lambda_u^2=W_{\mathrm{tot}}.
\end{align}
Similarly, we have $\sum_{e\in\F^n} W_z(e)=W_{\mathrm{tot}}$.
Recall that $W_{\mathrm{tot}}=2^{-n}\Tr(H^2)$.
We have $W_{\mathrm{tot}}=2^{-n}\Tr(H^2)\le\|H\|^2\le\Lambda^2$.

Now let $A_x\coloneqq \Pi_x^{(\epsilon)}\setminus\{0\}$.
For every $d\in A_x$, the definition of $\Pi_x^{(\epsilon)}$ gives $W_x(d)\ge\epsilon^2$.
Hence
\begin{align}
|A_x|\epsilon^2\le\sum_{d\in A_x}W_x(d)\le\sum_{d\in\F^n}W_x(d)=W_{\mathrm{tot}}\le\Lambda^2.
\end{align}
Dividing by $\epsilon^2$ gives $|\Pi_x^{(\epsilon)}\setminus\{0\}|\le\frac{\Lambda^2}{\epsilon^2}$.
Similarly, we have $|\Pi_z^{(\epsilon)}\setminus\{0\}|\le\frac{\Lambda^2}{\epsilon^2}$.
\end{proof}

Now, we are ready to prove Lemma~\ref{lem:projection-main}, which we rewrite as the following for convenience.

\begin{lemma}[Displacement sieve]\label{lem:projection-presieve}
There is an absolute constant $C>0$ that satisfies the following conditions. For any $0<\epsilon\le\Lambda$ and $\eta\in(0,1)$, there is a protocol that outputs sets $\widehat\Pi_x,\widehat\Pi_z\subseteq\mathbb F_2^n$ with probability at least $1-\eta$ and size $O(\tfrac{\Lambda^2}{\epsilon^2}\log\tfrac{1}{\eta})$ within deterministic total evolution time at most
\begin{align}
C\frac{\Lambda}{\epsilon^2}\log\frac{4\Lambda}{\epsilon}\log\frac{4}{\eta}.
\end{align}
\end{lemma}
\begin{proof}
We will complete the proof based on repeated sampling. We first show that every coordinate
$d$ with $W_x(d)\ge\epsilon^2$ appears often enough in the $Z$-basis displacement experiment. The same argument will then apply to the $X$ basis.

By Proposition~\ref{prop:z-traj}, every nonzero
$d\in\Pi_x^{(\epsilon)}$ appears in one $Z$-basis displacement shot with
probability at least $\frac{\epsilon^2}{2\kappa_0\Lambda^2}$. 
Similarly, by Proposition~\ref{prop:x-traj}, every nonzero $e\in\Pi_z^{(\epsilon)}$ appears in one $X$-basis displacement shot with
the same lower bound. 
Set $p\coloneqq \frac{\epsilon^2}{4\kappa_0\Lambda^2}$, then every nonzero coordinate in
$\Pi_x^{(\epsilon)}$ or $\Pi_z^{(\epsilon)}$ has appearance probability at least $2p$ in the corresponding basis.

Next, Proposition~\ref{prop:projection-counts} gives
\begin{align}
|\Pi_x^{(\epsilon)}\setminus\{0\}|\le\frac{\Lambda^2}{\epsilon^2},
\qquad
|\Pi_z^{(\epsilon)}\setminus\{0\}|\le\frac{\Lambda^2}{\epsilon^2}.
\end{align}
We let
\begin{align}
r\coloneqq \left\lceil\frac{\Lambda^2}{\epsilon^2}\right\rceil,\qquad
N\coloneqq \left\lceil C_0p^{-1}\log(4r)\right\rceil,\qquad
L\coloneqq \left\lceil C_1\log\frac{4}{\eta}\right\rceil,
\end{align}
where $C_0,C_1>0$ are sufficiently large absolute constants.

We now describe one round of experiments in the $Z$ basis. 
First sample $N$ times $t_1,\ldots,t_N\sim q_\Lambda$ independently, and let $\Theta\coloneqq \sum_{j=1}^N t_j$ represent the total evolution time planned for this round. 
To prevent the random sampling time from becoming too long, we need to set an upper limit on the total evolution time. 
A reasonable choice is $8$ times the average total time $\tau_{\max}\coloneqq 4\kappa_1\frac{N}{\Lambda}$.
If $\Theta>\tau_{\max}$, the round is aborted before any quantum evolution and outputs
the empty set. 
If $\Theta\le\tau_{\max}$, the round performs the $N$ $Z$-basis displacement shots with time $t_1,\ldots,t_N$. 
For every nonzero displacement $d$, let $A_d$ be the number of times $d$ appears. 
This round outputs all nonzero $d$ such that $A_d\ge Np$.
The $X$-basis round is defined in the same way. The final sets $\widehat\Pi_x$ and $\widehat\Pi_z$ are the unions of the outputs of $L$ independent rounds in the corresponding basis.

Now fix one $Z$-basis round and one coordinate $d\in\Pi_x^{(\epsilon)}\setminus\{0\}$.
Since $d$ appears in each shot with probability at least $2p$, the count $A_d$ has a mean of at least $2Np$. 
The Chernoff bound gives
\begin{align}
\Pr[A_d<Np]\le \exp(-Np/4).
\end{align}
As there are at most $r$ nonzero coordinates $d$ with $W_x(d)\ge\epsilon^2$, we have
\begin{align}
\Pr\left[\text{some } d\in\Pi_x^{(\epsilon)}\setminus\{0\}\text{ has } A_d<Np\right]\le r\exp(-Np/4).
\end{align}
By choosing $C_0$ large enough, this probability is at most $1/16$. 
Hence, with probability at least $15/16$, the round records every element of $\Pi_x^{(\epsilon)}\setminus\{0\}$, provided it is not aborted.

It remains to bound the probability of abortion. 
By Proposition~\ref{prop:even-sample}, a sampled time length from $q_\Lambda$ has the mean value $\frac{\kappa_1}{2\Lambda}$.
Hence, we have
\begin{align}
\E[\Theta]=N\frac{\kappa_1}{2\Lambda}.
\end{align}
By Markov's inequality,
\begin{align}
\Pr[\Theta>\tau]\le\frac{N\kappa_1/(2\Lambda)}{4\kappa_1N/\Lambda}=\frac18.
\end{align}
Then one $Z$-basis round succeeds, meaning that it is not aborted, and it
records every element of $\Pi_x^{(\epsilon)}\setminus\{0\}$, with probability at least $1-\frac{1}{16}-\frac18=\frac{13}{16}$.
The same proof applies to one $X$-basis round and the set $\Pi_z^{(\epsilon)}\setminus\{0\}$.

Since the rounds are independent, the probability that all $L$ $Z$-basis rounds fail is at most $\left(\frac{3}{16}\right)^L$.
By choosing $C_1$ large enough, we have $\left(\frac{3}{16}\right)^L\le\frac{\eta}{2}$.
The same bound holds for the $X$ basis. 
A union bound over the two bases gives
\begin{align}
\Pi_x^{(\epsilon)}\setminus\{0\}\subseteq\widehat\Pi_x,
\qquad
\Pi_z^{(\epsilon)}\setminus\{0\}\subseteq\widehat\Pi_z
\end{align}
with probability at least $1-\eta$.

We next bound the output size. 
In any non-aborted round, every recorded coordinate appears at least $Np$ times, while the total number of observations is exactly $N$. 
Hence, one round can record at most $\frac{N}{Np}=\frac{1}{p}$ coordinates. 
There are $L$ rounds in the $Z$ basis and $L$ rounds in the $X$ basis, so
\begin{align}
|\widehat\Pi_x|+|\widehat\Pi_z|\le\frac{2L}{p}\le C\frac{\Lambda^2}{\epsilon^2}\log\frac{4}{\eta}.
\end{align}

Finally, we bound the deterministic total evolution time. 
Each round uses at most $\tau$ evolution time. Therefore, the total time over both bases are at most
\begin{align}
2L\tau =8\kappa_1L\frac{N}{\Lambda}.
\end{align}
Using $p^{-1}=4\kappa_0\frac{\Lambda^2}{\epsilon^2}$ and $\log(4r)=O(\log\tfrac{4\Lambda}{\epsilon})$, we get
\begin{align}
2L\tau\le C\frac{\Lambda}{\epsilon^2}\log\frac{4\Lambda}{\epsilon}\log\frac{4}{\eta}
\end{align}
for some large enough constant $C$ as claimed.
\end{proof}
Assuming that we succeed in Lemma~\ref{lem:projection-presieve}, the Cartesian candidate set of the output
\begin{align}
\widehat\cU_{\mathrm{proj}}\coloneqq (\widehat\Pi_x\cup\{0\})\times(\widehat\Pi_z\cup\{0\})\setminus\{0\}
\end{align}
contains every $u=(x\mid z)$ with $|\lambda_u|\ge\epsilon$: if $x\neq0$, then $W_x(x)\ge\lambda_u^2\ge\epsilon^2$, and if $x=0$ then $x$ belongs to
$\widehat\Pi_x\cup\{0\}$. 
The same applies to $z$.

\section{Cross-Pauli Measurement}\label{app:linear}
\subsection{Frame and parity blocks}
The same-Pauli measurement finds where the coefficients with magnitude at least $\epsilon$ are. 
In this section, we use the cross-Pauli measurement to estimate their signs and magnitudes. 
The idea is to choose, on every qubit, one local direction $c_i$ whose coefficient letter we want to detect, and use the two remaining directions as prepare and measure bases. 
Finally, we use the input and output to construct parity variables, making them equal to the first-order derivative of the cross-correlation function in Eq.~\eqref{eq:cross-correlation}.

\begin{definition}[Target Pauli basis]
A target Pauli basis is a Pauli string $c=(c_1,\ldots,c_n)\in\{X,Y,Z\}^n$.
For each qubit $i$, define $(a_i,b_i)$ by the cyclic convention
\begin{align}
(c_i,a_i,b_i)\in\{(X,Y,Z),(Y,Z,X),(Z,X,Y)\}.
\end{align}
Thus $a_i b_i=i c_i$, $b_i c_i=i a_i$, and $c_i a_i=i b_i$.
For subsets $A,B\subseteq[n]$, define
\begin{align}
\begin{split}
Q_A&\coloneqq \bigotimes_{i=1}^n Q_i^{(A)},
\qquad
Q_i^{(A)}\coloneqq 
\begin{cases}
a_i,& i\in A,\\
I,& i\notin A,
\end{cases}\\
P_B&\coloneqq \bigotimes_{i=1}^n P_i^{(B)},
\qquad
P_i^{(B)}\coloneqq 
\begin{cases}
b_i,& i\in B,\\
I,& i\notin B.
\end{cases}
\end{split}
\end{align}
The label $u_c(A,B)\in V$ is defined sitewise by
\begin{align}
(u_c(A,B))_i=
\begin{cases}
I,& i\notin A\cup B,\\
a_i,& i\in A\setminus B,\\
b_i,& i\in B\setminus A,\\
c_i,& i\in A\cap B.
\end{cases}
\end{align}
The visible set of the target Pauli basis is
\begin{align}
\mathcal V(c)\coloneqq \{u\in V\setminus\{0\}: |\{i:u_i=c_i\}|\text{ is odd}\}.
\end{align}
For every $u\in\mathcal V(c)$, there is a unique pair $(A_c(u),B_c(u))$ such that $u=u_c(A_c(u),B_c(u))$. 
In fact, the set $A$ records where the input basis is used, and the set $B$ records where the output basis is used. 
When $A$ and $B$ overlap, the local product of the input and output directions leaves exactly the missing direction $c_i$.
\end{definition}

\begin{definition}[Parity Blocks]
Fix a target Pauli basis $c$. 
For $A,B\subseteq[n]$, set
\begin{align}
F^{(c)}_{A,B}(t)\coloneqq 2^{-n}\Tr(P_B(t)Q_A).
\end{align}
A target Pauli basis shot is performed as follows. 
First sample
\begin{align}
\tau\sim p_\Lambda,
\qquad
S\sim\mathrm{Unif}\{\pm1\},
\end{align}
independently, where $p_\Lambda$ is defined in Eq.~\eqref{eq:p_lambda} and again below in Lemma~\ref{lem:odd-kernel}. 
If $S=+1$, prepare independent random eigenstates of the $a_i$ bases, with eigenvalue signs $s_i\in\{\pm1\}$, evolve for time $\tau$, and measure every qubit in the $b_i$ basis, obtaining signs $m_i$. 
If $S=-1$, swap the two bases: prepare $b_i$ eigenstates with signs $r_i$, evolve for time $\tau$, and measure in the $a_i$ basis, obtaining signs $n_i$.

For $A,B\subseteq[n]$, define
\begin{align}
s_A\coloneqq \prod_{i\in A}s_i,
\quad
m_B\coloneqq \prod_{i\in B}m_i,
\quad
r_B\coloneqq \prod_{i\in B}r_i,
\quad
n_A\coloneqq \prod_{i\in A}n_i,
\end{align}
and the parity variable
\begin{align}
Z^{(c)}_{A,B}\coloneqq 
\begin{cases}
s_A m_B,& S=+1,\\
-r_B n_A,& S=-1.
\end{cases}
\end{align}
\end{definition}

We now prove the following two propositions.

\begin{proposition}[Extracts the odd part of cross-correlation functions]\label{prop:block-trace}
For every target Pauli basis $c$, subsets $A,B\subseteq[n]$, and time $t\ge 0$,
\begin{align}
\E[Z^{(c)}_{A,B}\mid \tau=t]=\frac12\Bigl(F^{(c)}_{A,B}(t)-F^{(c)}_{A,B}(-t)\Bigr).
\end{align}
\end{proposition}
\begin{proof}
If $S=+1$, the random local eigenvalue signs define a uniformly random eigenvector of $Q_A$ with eigenvalue $s_A$. 
By Proposition~\ref{prop:trace-main},
\begin{align}
\E[s_A m_B\mid T=t,S=+1]=2^{-n}\Tr(P_B(t)Q_A)=F^{(c)}_{A,B}(t).
\end{align}
If $S=-1$, the same trace rule, now with preparation observable $P_B$ and measurement observable $Q_A$, gives
\begin{align}
\E[-r_B n_A\mid T=t,S=-1]=-2^{-n}\Tr(Q_A(t)P_B).
\end{align}
By the cyclicity of trace,
\begin{align}
2^{-n}\Tr(Q_A(t)P_B)=2^{-n}\Tr(P_B(-t)Q_A)=F^{(c)}_{A,B}(-t).
\end{align}
Averaging the two equally likely signs $S=\pm1$ proves the claimed result.
\end{proof}

\begin{proposition}[A spectral representation of the cross-correlation function]\label{prop:odd-spectrum}
Assume $\|H\|\le \Lambda$. 
For every Hermitian Pauli pair $(P,Q)$, there exists a finite complex measure $\nu_{P,Q}$ supported on $[-2\Lambda,2\Lambda]$ such that
\begin{align}
F_{P,Q}(t)=\int_{-2\Lambda}^{2\Lambda}e^{i\omega t}d\nu_{P,Q}(\omega),
\qquad
|\nu_{P,Q}|([-2\Lambda,2\Lambda])\le 1.
\end{align}
Consequently,
\begin{align}
F_{P,Q}'(0)=\int_{-2\Lambda}^{2\Lambda}i\omega d\nu_{P,Q}(\omega).
\end{align}
\end{proposition}
\begin{proof}
As in Proposition~\ref{prop:even-spectrum}, the superoperator $\cA_H\coloneqq \ad_H$ is self-adjoint on the normalized Hilbert-Schmidt space, i.e., given $H|e_j\rangle=E_j|e_j\rangle$ with $|E_j|\le\Lambda$, then
\begin{align}
\cA_H(|e_j\rangle\langle e_k|)=(E_j-E_k)|e_j\rangle\langle e_k|.
\end{align}
Thus the spectrum of $\cA_H$ is contained in $[-2\Lambda,2\Lambda]$. 
Let $E(d\omega)$ be its spectral measure.
Define
\[
\nu_{P,Q}(\Omega)\coloneqq \langle Q,E(\Omega)P\rangle_{\HS}.
\]
Then, we have
\begin{align}
F_{P,Q}(t)
=\langle Q,e^{it\cA_H}P\rangle_{\HS}
=\int_{-2\Lambda}^{2\Lambda}e^{i\omega t}d\nu_{P,Q}(\omega).
\end{align}
To bound the total variation, let $\{\Omega_j\}$ be any finite measurable partition of $[-2\Lambda,2\Lambda]$. 
Since the ranges of the spectral projections $E(\Omega_j)$ are orthogonal, we have
\begin{align}
\begin{split}
\sum_j |\nu_{P,Q}(\Omega_j)|
&=\sum_j \bigl|\langle Q,E(\Omega_j)P\rangle_{\HS}\bigr|\\
&=\sum_j \bigl|\langle E(\Omega_j)Q,E(\Omega_j)P\rangle_{\HS}\bigr|\\
&\le
\left(\sum_j \|E(\Omega_j)Q\|_{\HS}^2\right)^{1/2}
\left(\sum_j \|E(\Omega_j)P\|_{\HS}^2\right)^{1/2}\\
&=\|Q\|_{\HS}\|P\|_{\HS}=1.
\end{split}
\end{align}
Taking the supremum over all finite partitions gives
\begin{align}
|\nu_{P,Q}|([-2\Lambda,2\Lambda])\le1.
\end{align}
Differentiation under the integral is allowed because the support is bounded. 
After differentiation, we obtain the claimed result.
\end{proof}
Fix the same smooth cutoff $\varphi$ as in the even-kernel construction and set
\begin{align}
\Psi_\Lambda(\omega)\coloneqq i\omega\varphi(\omega/(2\Lambda)).
\end{align}
Now we can define the odd kernel
\begin{align}
L_\Lambda(t)\coloneqq \frac{1}{2\pi}\int_{\mathbb R}\Psi_\Lambda(\omega)e^{-i\omega t}d\omega.
\end{align}
The odd kernel above satisfies the following lemma.

\begin{lemma}[Odd kernel for the first-order derivative of the cross-correlation function]\label{lem:odd-kernel}
If $f(t)=\int_{-2\Lambda}^{2\Lambda}e^{i\omega t}d\nu(\omega)$ for a finite complex measure $\nu$, then
\begin{align}
f'(0)=\int_{\mathbb R}L_\Lambda(t)f(t)dt=\int_0^\infty L_\Lambda(t)\bigl(f(t)-f(-t)\bigr)dt.
\end{align}
\end{lemma}
\begin{proof}
Because $L_\Lambda\in L^1(\mathbb R)$ and $|f(t)|\le |\nu|([-2\Lambda,2\Lambda])$, we apply Fubini's theorem:
\begin{align}
\int_{\mathbb R}L_\Lambda(t)f(t)dt=\int_{-2\Lambda}^{2\Lambda}\left(\int_{\mathbb R}L_\Lambda(t)e^{i\omega t}dt\right)d\nu(\omega).
\end{align}
The inner Fourier transform equals $\Psi_\Lambda(\omega)$, and on
$[-2\Lambda,2\Lambda]$ we have $\varphi(\omega/(2\Lambda))=1$, hence $\Psi_\Lambda(\omega)=i\omega$.
Therefore, we have
\begin{align}
\int_{\mathbb R}L_\Lambda(t)f(t)dt=\int_{-2\Lambda}^{2\Lambda}i\omega d\nu(\omega)=f'(0).
\end{align}
Since $L_\Lambda$ is odd, we have
\begin{align}
\int_{\mathbb R}L_\Lambda(t)f(t)dt=\int_0^\infty L_\Lambda(t)\bigl(f(t)-f(-t)\bigr)dt.
\end{align}
\end{proof}

As explained earlier, the kernel we will use is just a rescaling of the fixed kernel. 
So there exist absolute constants $\ell_0,\ell_1>0$, depending only on the cutoff, namely
\begin{align}
\|L_\Lambda\|_1=2\ell_0\Lambda,\qquad
\frac{2}{\|L_\Lambda\|_1}\int_0^\infty t|L_\Lambda(t)|dt=\frac{\ell_1}{2\Lambda}.
\end{align}
Now we define the odd sampling density as
\begin{align}
p_\Lambda(t)\coloneqq \frac{2|L_\Lambda(t)|}{\|L_\Lambda\|_1},
\qquad t\ge0.
\end{align}
Similar to the same-Pauli measurement, we define the weighted parity sample value as
\begin{align}
G^{(c)}_{A,B}\coloneqq \|L_\Lambda\|_1\sign(L_\Lambda(T))Z^{(c)}_{A,B}.
\end{align}

\subsection{Proof of Proposition~\ref{prop:parity-main}}
We first prove that the sampled parity block has an expectation equal to the first-order derivative at the initial time. 
By Proposition~\ref{prop:block-trace} and
Lemma~\ref{lem:odd-kernel}, we have
\begin{align}
\begin{split}
\E\bigl[G^{(c)}_{A,B}\bigr]&=\int_0^\infty\|L_\Lambda\|_1\sign(L_\Lambda(t))\cdot\frac12\bigl(F^{(c)}_{A,B}(t)-F^{(c)}_{A,B}(-t)\bigr)\cdot\frac{2|L_\Lambda(t)|}{\|L_\Lambda\|_1}dt\\
&=\int_0^\infty L_\Lambda(t)\bigl(F^{(c)}_{A,B}(t)-F^{(c)}_{A,B}(-t)\bigr)dt\\
&=\bigl(F^{(c)}_{A,B}\bigr)'(0).
\end{split}
\end{align}

It remains to compute this derivative. Since $P_B'(0)=i[H,P_B]$, we have
\begin{align}
\bigl(F^{(c)}_{A,B}\bigr)'(0)=\frac{i}{2^n}\sum_{u\in V}\lambda_u\Tr([P_u,P_B]Q_A).
\end{align}
If $P_u$ commutes with $P_B$, the term contributes nothing. If it anticommutes with
$P_B$, then $[P_u,P_B]=2P_uP_B$.
Therefore, the only possible nonzero trace is controlled by $\Tr(P_uP_BQ_A)$.
Looking qubit-by-qubit, as we have
\begin{align}
(P_u)_i=
\begin{cases}
I,& i\notin A\cup B,\\
a_i,& i\in A\setminus B,\\
b_i,& i\in B\setminus A,\\
c_i,& i\in A\cap B.
\end{cases}
\end{align}
Hence, at most one label can contribute, namely $u=u_c(A,B)$. 
For this label, $P_u$ and $P_B$ anticommute exactly on the index $i\in A\cap B$, for $(P_u)_i=c_i$, and $(P_B)_i=b_i$. 
Therefore, $[P_u,P_B]\neq0$ means $k=|A\cap B|\text{ is odd}$.
If $k$ is even, the derivative is zero. 
If $k$ is odd, then
\begin{align}
\bigl(F^{(c)}_{A,B}\bigr)'(0)=\frac{2i}{2^n}\lambda_{u_c(A,B)}\Tr(P_{u_c(A,B)}P_BQ_A).
\end{align}
It remains only to evaluate the phase. 
For $u=u_c(A,B)$, we have
\begin{align}
(P_u)_i(P_B)_i(Q_A)_i=
\begin{cases}
I,& i\notin A\cup B,\\
a_i^2=I,& i\in A\setminus B,\\
b_i^2=I,& i\in B\setminus A,\\
c_i b_i a_i=-iI,& i\in A\cap B.
\end{cases}
\end{align}
This indicates that 
\begin{align}
P_{u_c(A,B)}P_BQ_A=(-i)^kI,
\qquad
\Tr(P_{u_c(A,B)}P_BQ_A)=2^n(-i)^k.
\end{align}
We can thus compute the first-order derivative at zero time as
\begin{align}
\bigl(F^{(c)}_{A,B}\bigr)'(0)=2i(-i)^k\lambda_{u_c(A,B)}=2\sigma_c(A,B)\lambda_{u_c(A,B)}.
\end{align}

Combining the definition of $G^{(c)}_{A,B}$ and the identity $\|L_\Lambda\|_1=2\ell_0\Lambda$, we finally get $|G^{(c)}_{A,B}|\le 2\ell_0\Lambda$ as $|Z^{(c)}_{A,B}|\le1$.
For a visible label $u$, the block with $(A,B)=(A_c(u),B_c(u))$ has expectation $\E\bigl[G^{(c)}_{A,B}\bigr]=2\sigma_c(A,B)\lambda_u$. Since $\sigma_c(A,B)\in\{\pm1\}$.
Multiplying this block by $\sigma_c(A,B)/2$ removes the known sign and the factor $2$, leaving an unbiased sample of $\lambda_u$. 
Specifically, for $u\in\mathcal V(c)$, we define
\begin{align}
Y_u^{(c)}\coloneqq \frac12\sigma_c(A_c(u),B_c(u))
G^{(c)}_{A_c(u),B_c(u)}.
\end{align}
Then Proposition~\ref{prop:parity-main} gives
\begin{align}\label{eq:visible-frame-identity}
\E[Y_u^{(c)}]=\lambda_u,
\qquad
|Y_u^{(c)}|\le \ell_0\Lambda.
\end{align}
Therefore, proposition~\ref{prop:parity-main} is proved.

\subsection{Visibility correction}
For a fixed target Pauli basis $c$, the block we constructed above only gives a useful sample for labels in $\mathcal V(c)$. 
Concretely, this means the label has an odd number of indices where its Pauli letter is the target Pauli basis letter $c_i$. 
In this case, the corrected block $Y_u^{(c)}$ has mean $\lambda_u$. 
But if the number is even, Proposition~\ref{prop:parity-main} has shown that the corresponding first-order signal is zero, so this target Pauli basis is blind to that label. 
To address this issue, we choose the target Pauli basis at random. 
We need to make sure that every nonzero label becomes visible with a known probability $q(u)$. 
To achieve this, we rescale the sample by $q(u)^{-1}$ on the shots where it is visible. 
This will compensate for the shots where the label is invisible, so the averaged sample still outputs $\lambda_u$.

\begin{proposition}[Visibility probability of a random target Pauli basis]\label{prop:coverage}
Let $c\in\{X,Y,Z\}^n$ be a uniformly random chosen Pauli string with no identity site. 
For every fixed nonzero label $u\in V\setminus\{0\}$, we have
\begin{align}
\Pr[u\in\mathcal V(c)]=\frac{1-(1/3)^{w(u)}}{2}\ge \frac13.
\end{align}
where $w(u)$ is the Pauli weight of $u$, i.e.\ the number of qubits on which $u_i\neq I$. 
Then for independent random target Pauli bases $c^{(1)},\dots,c^{(R)}$, we have
\begin{align}
\Pr\bigl[u\notin\mathcal V(c^{(r)})\text{ for all }r=1,\dots,R\bigr]\le\left(\frac23\right)^R.
\end{align}
\end{proposition}
\begin{proof}
On every qubit where $u_i\neq I$, the random target Pauli basis letter satisfies $c_i=u_i$
with probability $1/3$. 
Therefore, the number of matches is
\begin{align}
X\sim\mathrm{Bin}(w(u),1/3).
\end{align}
The label is visible exactly when $X$ is odd. 
For a binomial random variable $X$ with success probability $p$, we have
\begin{align}
\Pr[X\text{ is odd}]= \frac{1-(1-2p)^{w(u)}}{2}.
\end{align}
Taking $p=1/3$ gives
\begin{align}
\Pr[u\in\mathcal V(c)]=\frac{1-(1/3)^{w(u)}}{2}.
\end{align}
Since $u\neq0$, we have $w(u)\ge1$, and the probability is at least $1/3$.
The claimed result can be derived from the independence of every single site.
\end{proof}
The method for correcting for missed target Pauli bases is simple. 
For a variable $u$, each observation has a probability of $q(u)$ of seeing it.
Therefore, we only look at the data that actually captures $u$, and then scale the calculated value up by $q(u)^{-1}$.
We define this explicitly below.

\begin{definition}[Visibility correction]\label{def:ht-correction}
Fix a nonzero label $u\in V\setminus\{0\}$ and let $C\sim\mathrm{Unif}\{X,Y,Z\}^n$ be an independent random target Pauli basis. 
The random variable
\begin{align}
\widetilde Y_u\coloneqq q(u)^{-1}\mathbbm{1}\{u\in\mathcal V(C)\}Y_u^{(C)}
\end{align}
is called the visibility-corrected sample for the label $u$. 
\end{definition}
The following proposition will demonstrate that even if we only receive a sample on the visible target Pauli bases, the average value obtained after visibility correction will equal the true coefficient.
\begin{proposition}[Visibility correction preserves the mean value of the coefficient]\label{prop:ht-correction}
For every nonzero label $u\in V\setminus\{0\}$, we have $\E[\widetilde Y_u]=\lambda_u$ and $|\widetilde Y_u|\le 3\ell_0\Lambda$.
\end{proposition}
\begin{proof}
By Proposition~\ref{prop:coverage}, $q(u)>0$. 
Using the identity\eqref{eq:visible-frame-identity}, we have
\begin{align}
\E[\widetilde Y_u]=
q(u)^{-1}
\sum_{c:u\in\mathcal V(c)}\Pr[C=c]\E[Y_u^{(c)}]=q(u)^{-1}\sum_{c:u\in\mathcal V(c)}\Pr[C=c]\lambda_u=\lambda_u.
\end{align}
Moreover, as $q(u)\ge1/3$, we have
\begin{align}
|\widetilde Y_u|\le q(u)^{-1}\ell_0\Lambda\le 3\ell_0\Lambda.
\end{align}
\end{proof}

\subsection{Proof of Lemma~\ref{lem:linear-oracle}}
Previously, we only discussed the case of one unbiased sampling. 
However, similar to the situation in the same-Pauli measurement, each sample may take an unreasonably long time. 
To prevent the program from running indefinitely, we break the computation task into several blocks and also set an abortion rule. 
We then take the median of the results from all blocks, which yields a higher accuracy and success rate.

\begin{definition}[One estimation block]
Fix a candidate set $\cU\subseteq V\setminus\{0\}$. 
If $\cU=\varnothing$, the estimator returns no values and uses zero evolution time. 
Therefore, we assume that $|\cU|\ge1$. 
For target accuracy $\epsilon>0$, we set
\begin{align}
N\coloneqq \left\lceil C_0\frac{\Lambda^2}{\epsilon^2}\log\bigl(16|\cU|\bigr)\right\rceil,
\end{align}
where $C_0>0$ is a sufficiently large absolute constant.

One estimation block consists of $N$ independent planned target Pauli basis shots. 
For shot $s$, we sample
\begin{align}
C_s\sim\mathrm{Unif}\{X,Y,Z\}^n,
\qquad
T_s\sim p_\Lambda,
\end{align}
independently, let $\Theta\coloneqq \sum_{s=1}^{N}T_s$ be the planned total time, and set the deterministic limit $8$ times the expected value of the planned total time as $\tau\coloneqq 4\ell_1\frac{N}{\Lambda}$.
If $\Theta>\tau$, we abort the block before any quantum evolution and return nothing on $\cU$. 
Otherwise, we execute the planned shots. 
For every $u\in\cU$, we define
\begin{align}
W_s(u)\coloneqq q(u)^{-1}\mathbbm{1}\{u\in\mathcal V(C_s)\}Y_u^{(C_s)},\qquad s=1,\ldots,N.
\end{align}
The block then outputs
\begin{align}
\wt\lambda(u)\coloneqq \frac1N\sum_{s=1}^{N}W_s(u),\qquad u\in\cU .
\end{align}
\end{definition}

We are now ready to prove Lemma~\ref{lem:linear-oracle}, which is restated below for convenience.

\begin{lemma}[Coefficient estimation on a candidate set]\label{lem:odd-oracle}
There exists an absolute constant $C>0$ such that the following holds. 
Let $\cU\subseteq V\setminus\{0\}$ be a candidate set, let $0<\epsilon\le\Lambda$, and let $\eta\in(0,1)$. 
If $\cU=\varnothing$, the protocol returns no value and uses zero time. 
If $|\cU|\ge1$, run $R\coloneqq \left\lceil 16\log\frac{2}{\eta}\right\rceil$ independent estimation blocks. 
If a block is aborted, assign it the value $0$ for every $u\in\cU$. 
For each $u\in\cU$, let $\widehat\lambda_{\cU}(u)$ be the median of the $R$ reported values. 
Then, with probability at least $1-\eta$,
\begin{align}
\max_{u\in\cU}|\widehat\lambda_{\cU}(u)-\lambda_u|\le \epsilon.
\end{align}
Moreover, the total evolution time is at most
\begin{align}
C\frac{\Lambda}{\epsilon^2}\log\bigl(4|\cU|\bigr)\log\frac{2}{\eta}
\end{align}
when $|\cU|\ge1$, and is zero when $\cU=\varnothing$.
\end{lemma}
\begin{proof}
The theorem clearly holds for the case of an empty set, so we consider the case where $|\cU|\ge1$.
Fix $u\in\cU$, the single-shot variable $W_s(u)$ is exactly the visibility-corrected sample from Definition~\ref{def:ht-correction}.
Proposition~\ref{prop:ht-correction} gives
\begin{align}
\E[W_s(u)]=\lambda_u,
\qquad
|W_s(u)|\le 3\ell_0\Lambda.
\end{align}
On the planned sample space, we define the raw average $\bar\lambda(u)\coloneqq \frac1N\sum_{s=1}^{N}W_s(u)$.
If the block is not aborted, this is exactly the value reported by the block.
After applying Hoeffding's inequality, we get
\begin{align}
\Pr\Bigl[|\bar\lambda(u)-\lambda_u|>\epsilon\Bigr]\le2\exp\left(-\frac{N\epsilon^2}{18\ell_0^2\Lambda^2}\right).
\end{align}
By choosing $C_0$ sufficiently large, the right-hand side is at most $1/(16|\cU|)$. 
A union bound over all $u\in\cU$ yields
\begin{align}
\Pr\left[\max_{u\in\cU}|\bar\lambda(u)-\lambda_u|>\epsilon\right]\le\frac{1}{16}.
\end{align}

The planned time of one block satisfies $\E[\Theta]=N\frac{\ell_1}{2\Lambda}$.
Since $\tau=4\ell_1\frac{N}{\Lambda}=8N\frac{\ell_1}{2\Lambda}$, Markov's inequality gives
\begin{align}
\Pr[\Theta>\tau]\le \frac18.
\end{align}
Therefore, a single block succeeds simultaneously for all labels in $\cU$ with probability at least $1-\frac1{16}-\frac18=\frac{13}{16}$.

Let $I_j$ be the indicator that block $j$ fails. 
The variables $I_j$ are independent and satisfy $\E [I_j]\le \frac{3}{16}$.
If fewer than half of the blocks fail, then more than half of the block outputs are $\epsilon$-accurate for every $u\in\cU$. 
Therefore, by taking the median for each coordinate, the result guarantees the precision of $\epsilon$ for $u \in U$ under all circumstances. 
Hence, we have
\begin{align}
\Pr\left[\max_{u\in\cU}|\widehat\lambda_{\cU}(u)-\lambda_u|>\epsilon
\right]\le\Pr\left[\sum_{j=1}^{R}I_j\ge\frac{R}{2}\right].
\end{align}
By the choice of $R$, the Chernoff bound gives
\begin{align}
\Pr\left[\sum_{j=1}^{R}I_j\ge\frac{R}{2}\right]\le\exp\left(-\frac{25}{128}R\right)\le\eta.
\end{align}

It remains to bound the total evolution time. Each block uses at most $\tau$, so the total time is at most $R\tau$.
Using $0<\epsilon\le\Lambda$ and $|\cU|\ge1$, we have
\begin{align}
N\le C'\frac{\Lambda^2}{\epsilon^2}\log(16|\cU|)
\end{align}
for an absolute constant $C'$.
We also note that $R\le C'\log\frac{2}{\eta}$ after enlarging $C'$. 
Therefore, we have
\begin{align}
R\tau=4\ell_1R\frac{N}{\Lambda}\le C\frac{\Lambda}{\epsilon^2}\log\bigl(4|\cU|\bigr)\log\frac{2}{\eta}
\end{align}
after enlarging the absolute constant $C$. 
This proves the theorem.
\end{proof}

\section{Proof of Theorem~\ref{thm:main}}
We have now clearly demonstrated two different measurements and their functions. 
The same-Pauli displacement sieve gives a small Cartesian candidate set that contains all detectable labels. 
The cross-Pauli estimator then estimates all coefficients in that candidate set. 
Now we combine them to give the complete process of our protocol. 

We note that the displacement sieve is only a projection sieve. 
At the population level, it identifies the large $x$- and $z$-projections
\begin{align}
\Pi_x^{(\epsilon)}\coloneqq \{d\in\F^n:W_x(d)\ge\epsilon^2\},\qquad
\Pi_z^{(\epsilon)}\coloneqq \{e\in\F^n:W_z(e)\ge\epsilon^2\}.
\end{align}
These projections define the Cartesian candidate set
\begin{align}
\cU_{\mathrm{proj}}\coloneqq (\Pi_x^{(\epsilon)}\cup\{0\})\times(\Pi_z^{(\epsilon)}\cup\{0\})\setminus\{0\}.
\end{align}
The true $\epsilon$-detectable support $S_\epsilon\coloneqq\{u\in V\setminus\{0\}:|\lambda_u|\ge\epsilon\}$ is contained in this candidate set $S_\epsilon\subseteq \cU_{\mathrm{proj}}$.
This means $\cU_{\mathrm{proj}}$ can be larger than $S_\epsilon$. 
Therefore, in the second stage, we need to estimate the coefficient $\lambda_u$ in the candidate set $\widehat\cU_{\mathrm{proj}}\coloneqq (\widehat\Pi_x\cup\{0\})\times(\widehat\Pi_z\cup\{0\})\setminus\{0\}$.
We define the precision of the second stage as $\epsilon_{\mathrm{lin}}$, i.e.,
\begin{align}
\max_{u\in\widehat\cU_{\mathrm{proj}}}|\widehat\lambda_u-\lambda_u|\le\epsilon_{\mathrm{lin}}.
\end{align} 
Thus, we obtain an estimation of the true support $S_\epsilon$ as
\begin{align}
\widehat S_\epsilon\coloneqq \{u\in\widehat\cU_{\mathrm{proj}}:|\widehat\lambda_u|>\theta\}.
\end{align} 
The remaining question is what values should $\theta$ and $\epsilon_{\mathrm{lin}}$ take to satisfy the requirement. 
We only need to note that we ultimately need every coefficient with $|\lambda_u|\ge\epsilon$ to pass the final threshold. 
Therefore, a reasonable choice is to take
\begin{align}
\epsilon_{\mathrm{lin}}\coloneqq \frac{\epsilon}{6},
\qquad
\theta\coloneqq \frac{\epsilon}{2}.
\end{align}
Now, we are ready to show Theorem~\ref{thm:main}, which we restate as the following for convenience.
\begin{theorem}[Upper bound restated]\label{thm:mainA}
Let
\begin{align}
H=\sum_{u\in S}\lambda_uP_u,
\qquad
S\subseteq V\setminus\{0\},
\qquad
\|H\|\le \Lambda,
\end{align}
and fix $0<\epsilon\le\Lambda$ and $\delta\in(0,1)$. Then Algorithm~\ref{algo:upper} outputs $(\wh S_\epsilon,\wh H_\epsilon)$ such that, with probability at least $1-\delta$,
\begin{align}
S_\epsilon\subseteq\wh S_\epsilon\subseteq\{u\in V\setminus\{0\}:|\lambda_u|>\epsilon/3\},\qquad
\max_{u\in S_\epsilon}|\wh\lambda_u-\lambda_u|\le \epsilon/6.
\end{align}
Moreover, for an absolute constant $C>0$,
\begin{align}
T_{\mathrm{tot}}^{\det}\le C\frac{\Lambda}{\epsilon^2}\log\frac{\Lambda}{\epsilon}\log\frac{1}{\delta}\log\log\frac{1}{\delta}.
\end{align}
\end{theorem}
\begin{proof}
First, we run the displacement sieve in Lemma~\ref{lem:projection-presieve} with scale $\epsilon$ and failure probability $\delta/2$. Let the outputs be $\wh\Pi_x$ and $\wh\Pi_z$, and define
\begin{align}
\wh\cU_{\mathrm{proj}}\coloneqq (\wh\Pi_x\cup\{0\})\times(\wh\Pi_z\cup\{0\})\setminus\{0\}.
\end{align}
On the projection success event, which has probability at least $1-\delta/2$, every label in $S_\epsilon$ is contained in $\wh\cU_{\mathrm{proj}}$.
Indeed, if $u=(x\mid z)$ and $|\lambda_u|\ge\epsilon$, then the corresponding projection intensity is at least $\epsilon^2$ in every nonzero coordinate. 
Hence, the sieve includes the nonzero $x$ and $z$ coordinates, while zero coordinates are also included.

We also need the size of the candidate set. The proof of Lemma~\ref{lem:projection-presieve} gives the deterministic bound
\begin{align}
|\wh\Pi_x|+|\wh\Pi_z|
\le
C_1\frac{\Lambda^2}{\epsilon^2}\log\frac{8}{\delta},
\end{align}
for an absolute constant $C_1$. 
Therefore, whenever $\wh\cU_{\mathrm{proj}}$ is nonempty, we have
\begin{align}
|\wh\cU_{\mathrm{proj}}|
\le
\bigl(|\widehat\Pi_x|+1\bigr)\bigl(|\widehat\Pi_z|+1\bigr)-1
\le
\left(C_1\frac{\Lambda^2}{\epsilon^2}\log\frac{8}{\delta}+1\right)^2.
\end{align}
Taking the logarithm keeps the product inside the logarithm as a sum. Hence
\begin{align}\label{eq:candidate-log-bound}
\log(4|\wh\cU_{\mathrm{proj}}|)
\le
C_2\left(
\log\frac{4\Lambda}{\epsilon}
+
\log\log\frac{8}{\delta}
\right).
\end{align}
Now consider any candidate set.
If $\wh\cU_{\mathrm{proj}}=\varnothing$, then there are no Pauli labels to estimate in the second stage. 
The theorem is clearly true. 
Otherwise, apply Lemma~\ref{lem:odd-oracle} on this fixed candidate set with accuracy $\epsilon/6$ and failure probability $\delta/2$. 
Together with Eq.~\eqref{eq:candidate-log-bound}, it gives
\begin{align}
\Pr\left[\max_{u\in\wh\cU_{\mathrm{proj}}}|\wh\lambda_u-\lambda_u|\le \frac{\epsilon}{6}\middle|\wh\cU_{\mathrm{proj}}\right]\ge 1-\frac{\delta}{2}.
\end{align}
Let $E_{\mathrm{proj}}$ be the projection success event and let $E_{\mathrm{lin}}$ be the event that the linear estimates are $\epsilon/6$-accurate on the candidate set. 
The above conditional statement holds for every candidate set, so
\begin{align}
\Pr(E_{\mathrm{proj}}\cap E_{\mathrm{lin}})\ge 1-\delta.
\end{align}

It remains to check the threshold. 
On $E_{\mathrm{proj}}\cap E_{\mathrm{lin}}$, we take any $u\in S_\epsilon$. 
Then $u\in\wh\cU_{\mathrm{proj}}$ and
\begin{align}
|\wh\lambda_u|\ge|\lambda_u|-\frac{\epsilon}{6}\ge\frac{5\epsilon}{6}>\frac{\epsilon}{2}=\theta.
\end{align}
Therefore, we have $u\in\wh S_\epsilon$. 
This proves $S_\epsilon\subseteq\wh S_\epsilon$.
Conversely, if $u\in\wh S_\epsilon$, then $u\in\wh\cU_{\mathrm{proj}}$ and
\begin{align}
|\lambda_u|\ge|\wh\lambda_u|-\frac{\epsilon}{6}>\frac{\epsilon}{2}-\frac{\epsilon}{6}=\frac{\epsilon}{3}.
\end{align}
Therefore, we have
\begin{align}
\wh S_\epsilon\subseteq\{u\in V\setminus\{0\}:|\lambda_u|>\epsilon/3\}.
\end{align}
Finally, we bound the deterministic evolution time. 
The projection stage uses
\begin{align}
T_{\mathrm{proj}}^{\det}\le C_3\frac{\Lambda}{\epsilon^2}\log\frac{4\Lambda}{\epsilon}\log\frac{8}{\delta}.
\end{align}
This stage uses accuracy $\epsilon/6$, so by Lemma~\ref{lem:odd-oracle} and Eq.~\eqref{eq:candidate-log-bound}, we have
\begin{align}
T_{\mathrm{lin}}^{\det}\le C_4\frac{\Lambda}{\epsilon^2}\log\frac{\Lambda^2}{\epsilon^2}\log\frac{4}{\delta}\log\log\frac{1}{\delta}.
\end{align}
Since $\log(4\Lambda/\epsilon)\le C\log(\Lambda^2/\epsilon^2)$ and $\log(8/\delta)\le C\log(1/\delta)$ after changing absolute constants, and $\log(\Lambda^2/\epsilon^2)=2\log(\Lambda/\epsilon)$, the sum of the two bounds gives
\begin{align}
T_{\mathrm{tot}}^{\det}\le C\frac{\Lambda}{\epsilon^2}\log\frac{\Lambda}{\epsilon}\log\frac{1}{\delta}\log\log\frac{1}{\delta}.
\end{align}
This proves the claimed result and thus Theorem~\ref{thm:main}.
\end{proof}
\section{Robustness to calibrated SPAM errors}\label{app:spam}

In this appendix, we consider the impact of calibrated SPAM errors. 
SPAM error has different effects on the displacement sieving stage and the parity estimation stage. 
In the displacement stage, a true coefficient of size $\epsilon$ produces a bit-flipping pattern with a probability of approximately $\epsilon^2/\Lambda^2$. 
Therefore, the calibration error rate of the recorded displacement tags must be lower than this. 
If a SPAM error produces erroneous bit-flipping patterns at the same rate as the true signal, then the original protocol alone cannot distinguish between signal and noise. 
In the coefficient stage, the calibrated SPAM error multiplies the parity signal by a known reliability factor. 
To maintain the unbiasedness of the estimator, we need to divide by these factors, which, however, increases variance, especially for candidate Pauli matrices with high weights.

\subsection{The calibrated SPAM error model}

For the coefficient-estimation stage, we use the following calibrated local depolarizing model.  
For a single qubit, let the preparation and measurement error be local depolarization noise on each qubit as
\begin{align}
\mathcal P(\rho)=r_{\rm p}\rho+(1-r_{\rm p})\frac{I}{2},
\qquad
\mathcal M(\rho)=r_{\rm m}\rho+(1-r_{\rm m})\frac{I}{2},
\end{align}
where $0<r_{\rm p},r_{\rm m}\le 1$ are known calibration numbers.  
The preparation noise $\mathcal P^{\otimes n}$ is applied after the ideal product-state preparation, and the measurement noise $\mathcal M^{\otimes n}$ is applied immediately before the ideal terminal Pauli measurement.  
The noise maps and all planned random settings are applied independently from shot to shot.  
For any Pauli string $P$,
\begin{align}
\mathcal P^{\otimes n}(P)=r_{\rm p}^{w(P)}P,
\qquad
(\mathcal M^{\otimes n})^\dagger(P)=r_{\rm m}^{w(P)}P,
\end{align}
where $w(P)$ is the weight of $P$ as the number of non-identity single-qubit Paulis in $P$.

The projection stage uses same-basis displacement records.  
We assume a calibrated bound for the full recorded displacement labels.  
Let $D_Z,E_X$ denote the ideal displacement labels generated by fresh $Z$- and $X$-basis displacement shots, including the fresh random time $T\sim q_\Lambda$.  
Let $\widetilde D_{Z,j},\widetilde E_{X,j}$ denote the recorded labels in the $j$-th shot in the presence of SPAM noise.  
We assume that there is a known number $\xi_1\ge0$ such that, for every block $j$, every history transcript on previous preparation and measurement records in the block so far (denoted as $\mathcal F_{j-1}$), and every nonzero displacement label,
\begin{align}\label{eq:spam-sequential-displacement}
\Pr[\widetilde D_{Z,j}=d\mid \mathcal F_{j-1}]\ge\Pr[D_Z=d]-\xi_1,\quad
\Pr[\widetilde E_{X,j}=e\mid \mathcal F_{j-1}]\ge\Pr[E_X=e]-\xi_1 .
\end{align}
For the displacement stage, the calibration bound is assumed shot by shot.
Here, $\xi_1$ bounds the full $n$-qubit displacement record instead of a single-qubit error rate. 
Therefore, any local calibration must first be converted to this full-record bound.

For Hermitian Pauli strings $P,Q$, recal that $F_{P,Q}(t)\coloneqq 2^{-n}\Tr(P(t)Q)$ with $P(t)=e^{iHt}Pe^{-iHt}$.
We have the following SPAM-noisy version of estimating $F_{P,Q}(t)$.

\begin{lemma}[SPAM-noise-scaled estimation of $F_{P,Q}(t)$]\label{lem:ham-spam-trace-rule}
Run the experiment with the intended preparation Pauli $Q$ and intended measurement Pauli $P$.  
Let $X\in\{\pm1\}$ be the usual product of the prepared eigenvalue sign and the measured eigenvalue sign.  
Under the calibrated depolarizing SPAM error model above, we have
\begin{align}
\E[X\mid t]=r_{\rm p}^{w(Q)}r_{\rm m}^{w(P)} F_{P,Q}(t).
\end{align}
Consequently, the rescaled variable $X^\sharp\coloneqq r_{\rm p}^{-w(Q)}r_{\rm m}^{-w(P)}X$ satisfies
\begin{align}
\E[X^\sharp\mid t]=F_{P,Q}(t),\qquad
|X^\sharp|\le r_{\rm p}^{-w(Q)}r_{\rm m}^{-w(P)}.
\end{align}
\end{lemma}

\begin{proof}
For the ideal experiment, the signed average over the prepared eigenbasis of $Q$ gives $2^{-n}Q$.  
After preparation noise, this signed operator becomes $2^{-n}\mathcal P^{\otimes n}(Q)=2^{-n}r_{\rm p}^{w(Q)}Q$.
Measurement noise before an ideal measurement of $P$ is equivalent, in the Heisenberg picture, to measuring $(\mathcal M^{\otimes n})^\dagger(P)=r_{\rm m}^{w(P)}P$.
Therefore, the observed signed expectation equals
\begin{align}
2^{-n}\Tr\left(r_{\rm m}^{w(P)}P(t)r_{\rm p}^{w(Q)}Q\right)=r_{\rm p}^{w(Q)}r_{\rm m}^{w(P)}F_{P,Q}(t).
\end{align}
Dividing by the known factor gives the last two claims.
\end{proof}

\subsection{SPAM error in the projection stage}

Let $\rho_\epsilon\coloneqq\min\big\{1,\tfrac{\epsilon^2}{2\kappa_0\Lambda^2}\big\}$.
Assume $0\le \xi_1<\rho_\epsilon$, and set $\rho_1\coloneqq\rho_\epsilon-\xi_1$.
We have the following proposition.

\begin{proposition}[SPAM-robust displacement stage]\label{prop:ham-spam-projection}
Fix $0<\epsilon\le \Lambda$ and $\eta\in(0,1)$.  
We set variables
\begin{align}
N_1\coloneqq\left\lceil C_1^{(0)}\rho_1^{-1}\log\frac{16\Lambda^2}{\epsilon^2}\right\rceil,\quad
L_1\coloneqq\left\lceil 16\log\frac4\eta\right\rceil,\quad
\tau_1\coloneqq 8N_1\frac{\kappa_1}{2\Lambda},
\end{align}
where $C_1^{(0)}$ is a sufficiently large absolute constant.
Replace the block length in the ideal displacement stage by $N_1$.  
Use $L_1$ independent blocks for each of the two bases $Z$ and $X$, and use the deterministic block cap $\tau_1$.
Let $\wh\Pi_x,\wh\Pi_z$ be the nonzero recorded displacement labels in the two bases.  
Then, with probability at least $1-\eta$,
\begin{align}
\Pi_x^{(\epsilon)}\setminus\{0\}\subseteq \wh\Pi_x,\qquad
\Pi_z^{(\epsilon)}\setminus\{0\}\subseteq \wh\Pi_z .
\end{align}
The deterministic projection-stage time is bounded by
\begin{align}
T_1\le C_1\frac{1}{\Lambda\rho_1}\log\frac{16\Lambda^2}{\epsilon^2}\log\frac4\eta,\quad\Rightarrow\quad
T_1\le C_1'\frac{\Lambda}{\epsilon^2}\log\frac{16\Lambda^2}{\epsilon^2}\log\frac4\eta
\end{align}
if $\xi_1\le \rho_\epsilon/2$ for constants $C_1,C_1'$.
\end{proposition}

\begin{proof}
We prove the $Z$-basis statement, and the $X$-basis statement follows an identical argument.  
For every nonzero $d\in\Pi_x^{(\epsilon)}$, Proposition~\ref{prop:z-traj} gives
\begin{align}
\Pr[D_Z=d]\ge\frac{W_x(d)}{2\kappa_0\Lambda^2}\ge\rho_\epsilon .
\end{align}
By Eq.~\eqref{eq:spam-sequential-displacement}, for every planned shot $j$,
\begin{align}
\Pr[\widetilde D_{Z,j}=d\mid \mathcal F_{j-1}]\ge\rho_\epsilon-\xi_1=\rho_1 .
\end{align}
Therefore, by the chain rule, a fixed $d\in\Pi_x^{(\epsilon)}\setminus\{0\}$ is missed by all $N_1$ planned shots in one block with probability at most
\begin{align}
\Pr[\widetilde D_{Z,1}\ne d,\ldots,\widetilde D_{Z,N_1}\ne d]\le(1-\rho_1)^{N_1}\le e^{-N_1\rho_1}.
\end{align}
The choice of $N_1$ makes this at most $\epsilon^2/(16\Lambda^2)$.  
Since $|\Pi_x^{(\epsilon)}\setminus\{0\}|\le\tfrac{\Lambda^2}{\epsilon^2}$, a union bound shows that the planned displacement samples in one $Z$-basis block contain every nonzero element of $\Pi_x^{(\epsilon)}$ with probability at least $15/16$.

The scheduled cumulative time in one block has the expectation $N_1\kappa_1/(2\Lambda)$.  
By Markov's inequality and the definition of $\tau_1$, the block aborts with probability at most $1/8$.  
Hence, a single block succeeds with probability at least $1-1/16-1/8=13/16$.  
With $L_1=\lceil16\log(4/\eta)\rceil$ independent blocks, the probability that no $Z$-basis block succeeds is at most $\eta/4$.  
The same estimate holds in the $X$ basis.  
A union bound over the two bases proves the two inclusions.

The time bound is the deterministic cap $2L_1\tau_1$.  
The final simplified bound follows from $\xi_1\le\rho_\epsilon/2$, which gives $\rho_1^{-1}\le2\rho_\epsilon^{-1}$, and from
\begin{align}
\rho_\epsilon^{-1}=\max\left\{1,\frac{2\kappa_0\Lambda^2}{\epsilon^2}\right\}\le C_\kappa\frac{\Lambda^2}{\epsilon^2}
\end{align}
for $0<\epsilon\le\Lambda$ and an absolute constant $C_\kappa$.
\end{proof}

\subsection{SPAM errors in the coefficient learning stage}

Fix a target Pauli basis $c\in\{X,Y,Z\}^n$.  
For $A,B\subseteq[n]$, keep the same Pauli strings $Q_A$ and $P_B$ used in the ideal stage.  
Define $a_{A,B}\coloneqq r_{\rm p}^{|A|}r_{\rm m}^{|B|}$ and $b_{A,B}\coloneqq r_{\rm p}^{|B|}r_{\rm m}^{|A|}$.
The parity variable under SPAM noise is
\begin{align}
Z^{(c),\sharp}_{A,B}\coloneqq\begin{cases}
a_{A,B}^{-1}s_A m_B, & S=+1,\\
-b_{A,B}^{-1}r_B n_A, & S=-1.
\end{cases}
\end{align}

\begin{lemma}[Identity on parity variable and weighted parity sample value under SPAM noise]\label{lem:ham-spam-frame}
For every target Pauli basis $c$, every $A,B\subseteq[n]$, and every $t\ge0$,
\begin{align}
\E[Z^{(c),\sharp}_{A,B}\mid T=t]=\frac12\left(F^{(c)}_{A,B}(t)-F^{(c)}_{A,B}(-t)\right).
\end{align}
Consequently, with $G^{(c),\sharp}_{A,B}\coloneqq\|L_\Lambda\|_1\sign(L_\Lambda(T))Z^{(c),\sharp}_{A,B}$ the corrected weighted parity sample value, one has
\begin{align}
\E[G^{(c),\sharp}_{A,B}]=\bigl(F^{(c)}_{A,B}\bigr)'(0).
\end{align}
\end{lemma}

\begin{proof}
Given $S=+1$, the intended preparation Pauli is $Q_A$, and the intended measurement Pauli is $P_B$.  
Lemma~\ref{lem:ham-spam-trace-rule} gives
\begin{align}
\E[s_A m_B\mid T=t,S=+1]=a_{A,B}F^{(c)}_{A,B}(t).
\end{align}
After division by $a_{A,B}$, the contribution is $F^{(c)}_{A,B}(t)$.  
Given $S=-1$, the intended preparation Pauli is $P_B$, and the intended measurement Pauli is $Q_A$.  
Lemma~\ref{lem:ham-spam-trace-rule} and cyclicity of trace give
\begin{align}
\E[r_B n_A\mid T=t,S=-1]=b_{A,B}F^{(c)}_{A,B}(-t).
\end{align}
The definition includes a minus sign, so this branch contributes $-F^{(c)}_{A,B}(-t)$.  
Averaging the two equally likely branches gives the first identity.  
The second follows the same kernel calculation as in the ideal stage.
\end{proof}

Let $\cU\subseteq V\setminus\{0\}$ be the candidate set used in the coefficient stage.  Define
\begin{align}
\zeta_2(\cU)\coloneqq\inf\left\{r_{\rm p}^{|A_c(u)|}r_{\rm m}^{|B_c(u)|},r_{\rm p}^{|B_c(u)|}r_{\rm m}^{|A_c(u)|}:u\in\cU,\ c\in\{X,Y,Z\}^n,\ u\in\mathcal V(c)\right\}.
\end{align}
If $\cU=\emptyset$, we set $\zeta_2(\cU)=1$. 
If every $u\in\cU$ has Pauli weight at most $k$, then $\zeta_2(\cU)\ge (r_{\rm p}r_{\rm m})^k$.  
If a $k$-local ansatz is known and desired, one may first replace $\cU$ by $\cU_{\le k}\coloneqq\{u\in\cU:w(u)\le k\}$, and without this filtering, projection false positives can have weight larger than that of the true terms.

\begin{proposition}[SPAM-robust coefficient estimation]\label{prop:ham-spam-linear}
Let $\cU\subseteq V\setminus\{0\}$ be finite and let $\zeta_2=\zeta_2(\cU)>0$.  
If $\cU=\emptyset$, the protocol returns no coefficients and uses zero time.  
Otherwise, run the coefficient estimation stage with variables under SPAM noise from Lemma~\ref{lem:ham-spam-frame}.  Then, for every target accuracy $\epsilon>0$ and failure probability $\eta\in(0,1)$, there is an estimator $\wh\lambda_{\cU}$ such that
\begin{align}
\Pr\left[\max_{u\in\cU}|\wh\lambda_{\cU}(u)-\lambda_u|\le \epsilon\right]\ge1-\eta,
\end{align}
and the deterministic time satisfies
\begin{align}
T^{\sharp}_{\rm lin}(\cU;\epsilon,\eta)\le C_2\frac{\Lambda}{\epsilon^2\zeta_2^2}\log(4\max\{|\cU|,1\})\log\frac2\eta.
\end{align}
If every $u\in\cU$ has Pauli weight at most $k$, then the variance overhead is at most $(r_{\rm p}r_{\rm m})^{-2k}$.
\end{proposition}

\begin{proof}
If $\cU=\emptyset$, the conclusion is straightforward. 
Assume $\cU\ne\emptyset$.  
For a visible label $u\in\mathcal V(c)$, define the corrected target Pauli sample exactly as in the ideal proof, but with $G^{(c),\sharp}_{A,B}$ in place of $G^{(c)}_{A,B}$:
\begin{align}
Y^{(c),\sharp}_u\coloneqq \frac12\sigma_c(A_c(u),B_c(u))G^{(c),\sharp}_{A_c(u),B_c(u)} .
\end{align}
If $u\notin\mathcal V(c)$, set $Y^{(c),\sharp}_u=0$.  Lemma~\ref{lem:ham-spam-frame} and the ideal case calculation give
\begin{align}
\E[Y^{(c),\sharp}_u]=\lambda_u\quad (u\in\mathcal V(c)),\qquad
|Y^{(c),\sharp}_u|\le\ell_0\Lambda\zeta_2^{-1}.
\end{align}
After the correction over randomly chosen $c$ as $\widetilde Y^\sharp_u\coloneqq q(u)^{-1}\mathbbm{1}\{u\in\mathcal V(C)\}Y^{(C),\sharp}_u$, we still have
\begin{align}
\E[\widetilde Y^\sharp_u]=\lambda_u,\qquad
|\widetilde Y^\sharp_u|\le3\ell_0\Lambda\zeta_2^{-1},
\end{align}
because $q(u)\ge1/3$.

The rest is the same Hoeffding and median-of-blocks argument as in the ideal stage, with the single-shot bound $3\ell_0\Lambda$ replaced by $3\ell_0\Lambda\zeta_2^{-1}$.  
Thus, the number of shots and the deterministic time are multiplied by $\zeta_2^{-2}$.  
If all labels in $\cU$ have weight at most $k$, then $|A_c(u)|\le k$ and $|B_c(u)|\le k$, so every preparation-measurement factor is at least $(r_{\rm p}r_{\rm m})^k$.  Hence, we have $\zeta_2^{-2}\le (r_{\rm p}r_{\rm m})^{-2k}$.
\end{proof}

\subsection{SPAM-robust Hamiltonian reconstruction}

We now combine the arguments in the previous two stages and obtain the following theorem.

\begin{theorem}[SPAM-robust Hamiltonian reconstruction]\label{thm:spam-main}
Let $H=\sum_{u\in V\setminus\{0\}}\lambda_u P_u$ with $\|H\|\le\Lambda$, and fix $0<\epsilon\le\Lambda$ and $\delta\in(0,1)$.  
Assume the calibrated displacement condition Eq.~\eqref{eq:spam-sequential-displacement} with $\xi_1<\rho_\epsilon=\min\big\{1,\tfrac{\epsilon^2}{2\kappa_0\Lambda^2}\big\}$.
Run the SPAM-robust projection stage with failure probability $\delta/2$ and form $\wh\cU\coloneqq\bigl((\wh\Pi_x\cup\{0\})\times(\wh\Pi_z\cup\{0\})\bigr)\setminus\{0\}$, and then run the SPAM-robust coefficient estimation stage on $\wh\cU$ with target accuracy $\epsilon/6$ and failure probability $\delta/2$.  
Finally, we set
\begin{align}
\wh S_\epsilon\coloneqq\left\{u\in\wh\cU:|\wh\lambda_{\wh\cU}(u)|>\frac{\epsilon}{2}\right\},\qquad
\wh H_\epsilon\coloneqq\sum_{u\in\wh S_\epsilon}\wh\lambda_{\wh\cU}(u)P_u.
\end{align}
Then, with probability at least $1-\delta$,
\begin{align}
S_\epsilon\subseteq\wh S_\epsilon\subseteq\{u\in V\setminus\{0\}:|\lambda_u|>\epsilon/3\},\qquad
\max_{u\in S_\epsilon}|\wh\lambda_{\wh\cU}(u)-\lambda_u|\le\epsilon/6 .
\end{align}
For every realized candidate set $\wh\cU$, the following pathwise scheduled-time bound holds:
\begin{align}
T^\sharp_{\rm tot}\le C\left[\frac{1}{\Lambda(\rho_\epsilon-\xi_1)}\log\frac{16\Lambda^2}{\epsilon^2}\log\frac8\delta+\frac{\Lambda}{\epsilon^2\zeta_2(\wh\cU)^2}\log(4\max\{|\wh\cU|,1\})\log\frac4\delta\right].
\end{align}
Moreover, $|\wh\Pi_x|,|\wh\Pi_z|\le L_1N_1$, so the logarithmic candidate-size factor can be made nonrandom by replacing $|\wh\cU|$ with $(L_1N_1+1)^2-1$.  
In particular, if $\xi_1\le \rho_\epsilon/2$ and every label in the realized candidate set $\wh\cU$ has Pauli weight at most $k$, then
\begin{align}
T^\sharp_{\rm tot}=\widetilde O\left(\frac{\Lambda}{\epsilon^2}\left(1+(r_{\rm p}r_{\rm m})^{-2k}\right)\right).
\end{align}
\end{theorem}

\begin{proof}
By Proposition~\ref{prop:ham-spam-projection}, with probability at least $1-\delta/2$, we have $\Pi_x^{(\epsilon)}\setminus\{0\}\subseteq\wh\Pi_x$ and $\Pi_z^{(\epsilon)}\setminus\{0\}\subseteq\wh\Pi_z$.
On this event, every $u=(x(u)\mid z(u))\in S_\epsilon$ belongs to $\wh\cU$.  
Indeed, if $x(u)\ne0$, then $W_x(x(u))\ge \lambda_u^2\ge \epsilon^2$, so $x(u)\in\Pi_x^{(\epsilon)}\subseteq\wh\Pi_x$.  
If $x(u)=0$, then it is included by the definition of $\wh\cU$.  
The same argument applies to $z(u)$.

Conditioned on the realized set $\wh\cU$, proposition~\ref{prop:ham-spam-linear}, applied with target accuracy $\epsilon/6$, gives
\begin{align}
\Pr\left[\max_{u\in\wh\cU}|\wh\lambda_{\wh\cU}(u)-\lambda_u|\le\epsilon/6\middle|\wh\cU\right]\ge1-\delta/2 .
\end{align}
A union bound gives simultaneous success of the projection and coefficient stages with probability at least $1-\delta$.

On this event, if $u\in S_\epsilon$, then $u\in\wh\cU$ and
\begin{align}
|\wh\lambda_{\wh\cU}(u)|\ge|\lambda_u|-\epsilon/6\ge5\epsilon/6>\epsilon/2,
\end{align}
so $u\in\wh S_\epsilon$. 
Conversely, if $u\in\wh S_\epsilon$, then
\begin{align}
|\lambda_u|\ge|\wh\lambda_{\wh\cU}(u)|-\epsilon/6>\epsilon/2-\epsilon/6=\epsilon/3.
\end{align}
This proves the support inclusions.  
The coefficient error bound on $S_\epsilon$ is exactly the coefficient recovery stage accuracy.

The pathwise time bound is the sum of the projection-stage bound in Proposition~\ref{prop:ham-spam-projection} and the coefficient-stage bound in Proposition~\ref{prop:ham-spam-linear}, after conditioning on the realized candidate set.  
The final displayed form follows from $\xi_1\le\rho_\epsilon/2$ and $\zeta_2(\wh\cU)\ge(r_{\rm p}r_{\rm m})^k$ when all labels in the realized $\wh\cU$ have weight at most $k$.
\end{proof}

\section{Discrete Time Assumption}\label{app:lattice}

\subsection{Sampling rules on the time grid}
The previous analysis allows each shot to use evolution times of arbitrary lengths randomly sampled from a distribution. 
This is convenient for writing the kernel estimators, but it is not essential. 
As real-world experiments cannot perform extremely short evolutions, we need to discuss the feasibility of our protocol under the discretization assumption.
To this end, we consider a more restrictive setting where we only allow resolution time chosen from a discrete-time grid $\mathcal T_{t_0}\coloneqq\{0,t_0,2t_0,\ldots\}$.

For a Hamiltonian with norm bound $\Lambda$, there exists a general upper bound on the minimal time resolution, namely $t_0<\pi/\Lambda$. 
As long as the grid size is less than the upper bound, our protocol can still work with only constant factors overhead over Theorem~\ref{thm:main}. 
The intuition is that signals are still band-limited, so we can use Poisson summation to convert the continuous kernel identities into exact discrete identities. 
In contrast, if the grid is too coarse, it is possible for two different Hamiltonians to have identical evolutions at every allowed time. 
In that case, support recovery is impossible. 
In other words, the resolution $t_0=O(1/\Lambda)$ is unavoidable under the setting where we only allow resolution time chosen from a discrete-time grid $\mathcal T_{t_0}\coloneqq\{0,t_0,2t_0,\ldots\}$.

\begin{definition}[Resolution time in the grid setting]
A resolution time is a number $t_0>0$ such that every allowed evolution time belongs to the grid
\begin{align}
\mathcal T_{t_0}\coloneqq \{0,t_0,2t_0,\ldots\}.
\end{align}
The discrete time control-free model is the same access model as Definition~\ref{def:access_model}, except that every shot must use a time in $\mathcal T_{t_0}$.
\end{definition}
For our discretized algorithm, we need to fix $t_0$ based on a frequency count. 
The signals produced by the Hamiltonian use only frequencies in $[-2\Lambda,2\Lambda]$. 
The kernels $K_\Lambda$ and $L_\Lambda$ use only frequencies in $[-4\Lambda,4\Lambda]$. 
Therefore, after multiplying a signal by a kernel, the product involves frequencies in $[-6\Lambda,6\Lambda]$.
Poisson summation compares a continuous integral with a grid sum by looking at the Fourier transform at the points
\begin{align}
0,\quad
\pm\frac{2\pi}{t_0},\quad
\pm\frac{4\pi}{t_0},\quad
\ldots .
\end{align}
The point $0$ gives the integral. 
We want every other point to lie outside $[-6\Lambda,6\Lambda]$, because the Fourier transform is zero there. 
It would already be enough to assume $\frac{2\pi}{t_0}>6\Lambda$, namely $t_0<\frac{\pi}{3\Lambda}$. 
In the following parts of this section, we will use the slightly stronger condition
\begin{align}\label{eq:lattice-condition}
0<t_0<\frac{\pi}{6\Lambda}.
\end{align}
This only changes absolute constants.

\subsection{Discretized kernel}
Before describing the specific discretized kernel, we further explain some principles of kernel construction. 
Using the even kernel as an example
\begin{align}
B_\Lambda\coloneqq 2\Lambda,\qquad
\Phi_\Lambda(\omega)\coloneqq \omega^2\varphi(\omega/B_\Lambda),\qquad
K_\Lambda(t)\coloneqq -\frac{1}{2\pi}\int_{\mathbb R}\Phi_\Lambda(\omega)e^{i\omega t}d\omega .
\end{align}
Define the fixed profile kernel
\begin{align}
K_\ast(s)\coloneqq 
-\frac{1}{2\pi}\int_{\mathbb R}\xi^2\varphi(\xi)e^{i\xi s}d\xi.
\end{align}
Then $K_\ast$ is independent of $\Lambda$. 
By changing the variables into $\omega=B_\Lambda\xi$, we get
\begin{align}
K_\Lambda(t)=-\frac{1}{2\pi}\int_{\mathbb R}(B_\Lambda\xi)^2\varphi(\xi)e^{iB_\Lambda\xi t}B_\Lambda d\xi=B_\Lambda^3K_\ast(B_\Lambda t).
\end{align}
Since $B_\Lambda=2\Lambda$, this is
\begin{align}\label{eq:K-grid-scaling}
K_\Lambda(t)=(2\Lambda)^3K_\ast(2\Lambda t).
\end{align}
The power $3$ comes from two sources: the factor $\omega^2$ contributes $B_\Lambda^2$, and the measure $d\omega$ contributes one more $B_\Lambda$.
Similarly, recall from the odd-kernel construction that
\begin{align}
\Psi_\Lambda(\omega)\coloneqq i\omega\varphi(\omega/B_\Lambda),\qquad
L_\Lambda(t)\coloneqq \frac{1}{2\pi}\int_{\mathbb R}\Psi_\Lambda(\omega)e^{-i\omega t}d\omega .
\end{align}
Define the fixed profile kernel
\begin{align}
L_\ast(s)\coloneqq \frac{1}{2\pi}\int_{\mathbb R}i\xi\varphi(\xi)e^{-i\xi s}d\xi .
\end{align}
Again using $\omega=B_\Lambda\xi$, we obtain
\begin{align}
L_\Lambda(t)=\frac{1}{2\pi}\int_{\mathbb R}iB_\Lambda\xi\varphi(\xi)e^{-iB_\Lambda\xi t}B_\Lambda d\xi=B_\Lambda^2L_\ast(B_\Lambda t).
\end{align}
Thus
\begin{align}\label{eq:L-grid-scaling}
L_\Lambda(t)=(2\Lambda)^2L_\ast(2\Lambda t).
\end{align}
The power $2$ appears because the first derivative kernel has only one factor of $\omega$, and $d\omega$ gives the other factor.

With these detailed principles, we can give an exact summation rule that replaces continuous integrals. 
For clarity, we assume that a function has frequencies in $[-B,B]$, i.e., its Fourier transform is zero outside $[-B,B]$.
\begin{proposition}[Summation on the time grid]\label{prop:lattice-quadrature} Assume Eq.~\eqref{eq:lattice-condition}, and let $h$ be a rapidly decaying function that uses only frequencies in $[-6\Lambda,6\Lambda]$. 
Then, we have
\begin{align}
\int_{\mathbb R}h(t)dt=t_0\sum_{m\in\mathbb Z}h(mt_0).
\end{align}
\end{proposition}

\begin{proof}
The Poisson summation formula gives
\begin{align}
t_0\sum_{m\in\mathbb Z}h(mt_0)=\sum_{\ell\in\mathbb Z}\widehat h\left(\frac{2\pi\ell}{t_0}\right).
\end{align}
The right-hand side checks $\widehat h$ at the points $0,\pm\frac{2\pi}{t_0},\pm\frac{4\pi}{t_0},\ldots$.
The point $0$ gives
\begin{align}
\widehat h(0)=\int_{\mathbb R}h(t)dt.
\end{align}
Now take any nonzero integer $\ell$. Since
$t_0<\pi/(6\Lambda)$, we have
\begin{align}
\left|\frac{2\pi\ell}{t_0}\right|\ge\frac{2\pi}{t_0}>12\Lambda.
\end{align}
This point is outside $[-6\Lambda,6\Lambda]$, while $h$ uses no frequencies outside this interval. 
Therefore
\begin{align}
\widehat h\left(\frac{2\pi\ell}{t_0}\right)=0
\qquad
(\ell\neq0).
\end{align}
So only the $\ell=0$ term remains and thus
\begin{align}
t_0\sum_{m\in\mathbb Z}h(mt_0)=\widehat h(0)=\int_{\mathbb R}h(t)dt.
\end{align}
\end{proof}
We now apply this rule to the signals used by the algorithm.
For a nonzero $d\in\F^n$, we define
\begin{align}
f_d^{(Z)}(t)\coloneqq 2^{-n}\sum_{a\in\F^n}\left|\langle a+d|e^{-iHt}|a\rangle_Z\right|^2 .
\end{align}
Also for a nonzero $e\in\F^n$, we define
\begin{align}
f_e^{(X)}(t)\coloneqq 2^{-n}\sum_{a\in\F^n}\left|\langle a+e|e^{-iHt}|a\rangle_X\right|^2 .
\end{align}
Thus $f_d^{(Z)}(t)$ and $f_e^{(X)}(t)$ are the fixed-time displacement probabilities in the $Z$ and $X$ bases.
For a target Pauli basis $c$ and subsets $A,B\subseteq[n]$, we define
\begin{align}
g_{A,B}^{(c)}(t)\coloneqq F_{A,B}^{(c)}(t)-F_{A,B}^{(c)}(-t).
\end{align}
which is used to extract the odd time inversion part.
We can prove the following proposition.
\begin{proposition}[Kernel identities on the time grid]\label{prop:lattice-quadrature-identities}
Assume Eq.~\eqref{eq:lattice-condition}. 
Then for every $v\in V$, every nonzero $d,e\in\F^n$, every target Pauli basis $c$, and every
$A,B\subseteq[n]$, we have
\begin{align}
C_{P_v}''(0)&=t_0K_\Lambda(0)C_{P_v}(0)+2t_0\sum_{m\ge1}K_\Lambda(mt_0)C_{P_v}(mt_0),\label{eq:lattice-auto-identity}\\
2W_x(d)&=2t_0\sum_{m\ge1}K_\Lambda(mt_0)f_d^{(Z)}(mt_0),\label{eq:lattice-z-displacement-identity}\\
2W_z(e)&=2t_0\sum_{m\ge1}K_\Lambda(mt_0)f_e^{(X)}(mt_0),\label{eq:lattice-x-displacement-identity}\\
\bigl(F_{A,B}^{(c)}\bigr)'(0)&=t_0\sum_{m\ge1}L_\Lambda(mt_0)g_{A,B}^{(c)}(mt_0).\label{eq:lattice-linear-identity}
\end{align}
\end{proposition}

\begin{proof}
We apply Proposition~\ref{prop:lattice-quadrature} to the four signals in the claimed proposition. 
By Proposition~\ref{prop:even-spectrum} and Lemma~\ref{lem:transition-lower}, the functions $C_{P_v}$, $f_d^{(Z)}$, and $f_e^{(X)}$ use only frequencies in $[-2\Lambda,2\Lambda]$. 
By Proposition~\ref{prop:odd-spectrum}, the same is true for $F_{A,B}^{(c)}$. 
The kernels use only frequencies in $[-4\Lambda,4\Lambda]$ by construction of the cutoff. Therefore, the four products $K_\Lambda C_{P_v}$, $K_\Lambda f_d^{(Z)}$, $K_\Lambda f_e^{(X)}$, and $L_\Lambda F_{A,B}^{(c)}$ use only frequencies in $[-6\Lambda,6\Lambda]$.
Applying the exact grid identity to these products and using Lemmas~\ref{lem:even-kernel} and~\ref{lem:odd-kernel} gives the displayed formulas. 
The first three formulas use the evenness of $K_\Lambda$, $C_{P_v}$, $f_d^{(Z)}$, and $f_e^{(X)}$. 
For the displacement formulas, we also use $f_d^{(Z)}(0)=0$ and $f_e^{(X)}(0)=0$ because $d,e\neq0$. 
Finally, the last formula uses that $L_\Lambda$ is odd, so the positive and negative grid points combine into
$F_{A,B}^{(c)}(t)-F_{A,B}^{(c)}(-t)$.
\end{proof}

In the continuous time setting, we proved two bounds: one for the size of a one weighted sample and the other for the average sampling time of a single shot. Now we give their versions under the discrete-time assumption.
\begin{lemma}[Bounds under discrete time assumption]\label{lem:lattice-grid-weights}
Assume Eq.~\eqref{eq:lattice-condition}. There exist constants $\kappa_0^{\rm grid}$, $\kappa_1^{\rm grid}$, $\ell_0^{\rm grid}$, and $\ell_1^{\rm grid}$, depending only on the cutoff $\varphi$, such that
\begin{align}
t_0|K_\Lambda(0)|+2t_0\sum_{m\ge1}|K_\Lambda(mt_0)|&\le4\kappa_0^{\rm grid}\Lambda^2,\label{eq:grid-K-mass}\\
2t_0\sum_{m\ge1}mt_0|K_\Lambda(mt_0)|&\le2\kappa_0^{\rm grid}\kappa_1^{\rm grid}\Lambda,\label{eq:grid-K-time}\\
2t_0\sum_{m\ge1}|L_\Lambda(mt_0)|&\le2\ell_0^{\rm grid}\Lambda,\label{eq:grid-L-mass}\\
2t_0\sum_{m\ge1}mt_0|L_\Lambda(mt_0)|&\le\ell_0^{\rm grid}\ell_1^{\rm grid}.\label{eq:grid-L-time}
\end{align}
\end{lemma}

\begin{proof}
Set $\Delta=2\Lambda t_0$. 
Then $0<\Delta<\pi/3$. 
Using Eq.~\eqref{eq:K-grid-scaling} and Eq.~\eqref{eq:L-grid-scaling}, we have
\begin{align}
t_0|K_\Lambda(0)|+2t_0\sum_{m\ge1}|K_\Lambda(mt_0)|=4\Lambda^2\left(\Delta|K_\ast(0)|+2\Delta\sum_{m\ge1}|K_\ast(m\Delta)|\right).
\end{align}
The other three displayed bounds reduce in the same way to sums of $K_\ast$ and $L_\ast$ on a grid of mesh $\Delta$. 
Since $K_\ast$ and $L_\ast$ decay rapidly, these rescaled sums and their first moments are bounded uniformly for $0<\Delta<\pi/3$. 
Enlarging the four constants gives
\eqref{eq:grid-K-mass}-\eqref{eq:grid-L-time}.
\end{proof}
Finally, we propose a discrete version of the projection sieving algorithm.

\begin{corollary}[Discrete displacement sieve]\label{cor:lattice-projection}
Assume Eq.~\eqref{eq:lattice-condition}. 
Define one discrete displacement shot as follows. 
Sample $M\in\{0,1,2,\ldots\}$ by
\begin{align}
\Pr[M=0]&\coloneqq 1-\frac{2t_0\sum_{m\ge1}|K_\Lambda(mt_0)|}{4\kappa_0^{\rm grid}\Lambda^2},\label{eq:grid-proj-null}\\
\Pr[M=m]&\coloneqq \frac{2t_0|K_\Lambda(mt_0)|}{4\kappa_0^{\rm grid}\Lambda^2},\qquad m\ge1.\label{eq:grid-proj-positive}
\end{align}
If $M=0$, record a null symbol and perform no evolution. 
If $M=m\ge1$, run the usual $Z$-basis or $X$-basis displacement shot at time $mt_0$.

Let $D^{\rm grid}$ be the displacement returned by the $Z$-basis experiment,
and let $E^{\rm grid}$ be the displacement returned by the $X$-basis
experiment. Then for every nonzero $d,e\in\F^n$,
\begin{align}
\Pr[D^{\rm grid}=d]\ge \frac{W_x(d)}{2\kappa_0^{\rm grid}\Lambda^2},\qquad
\Pr[E^{\rm grid}=e]\ge \frac{W_z(e)}{2\kappa_0^{\rm grid}\Lambda^2}.
\end{align}
Moreover,
\begin{align}
\E[T_{\rm proj}^{\rm grid}]\le\frac{\kappa_1^{\rm grid}}{2\Lambda}.
\end{align}
\end{corollary}

\begin{proof}
The distribution of $M$ is valid by Eq.~\eqref{eq:grid-K-mass}. 
We prove the $Z$-basis claim here, and the $X$-basis claim is identical. 
For a fixed $m\ge1$, the probability of displacement $d$ at time $mt_0$ is $f_d^{(Z)}(mt_0)$. 
Hence, we have
\begin{align}
\Pr[D^{\rm grid}=d]=\frac{1}{2\kappa_0^{\rm grid}\Lambda^2}\sum_{m\ge1}t_0|K_\Lambda(mt_0)|f_d^{(Z)}(mt_0).
\end{align}
Since $f_d^{(Z)}(t)\ge0$, we have
\begin{align}
\sum_{m\ge1}t_0|K_\Lambda(mt_0)|f_d^{(Z)}(mt_0)\ge\left|\sum_{m\ge1}t_0K_\Lambda(mt_0)f_d^{(Z)}(mt_0)\right|.
\end{align}
By Eq.~\eqref{eq:lattice-z-displacement-identity}, the last signed sum is $W_x(d)$. 
This proves the lower bound for $D^{\rm grid}$. 
The proof for $E^{\rm grid}$ uses Eq.~\eqref{eq:lattice-x-displacement-identity} and follows a similar approach.

The mean-time bound follows directly from Eq.~\eqref{eq:grid-K-time} as
\begin{align}
\E[T_{\rm proj}^{\rm grid}]=\sum_{m\ge1}\frac{2t_0|K_\Lambda(mt_0)|}{4\kappa_0^{\rm grid}\Lambda^2}(mt_0)\le\frac{\kappa_1^{\rm grid}}{2\Lambda}.
\end{align}
\end{proof}

We next replace one continuous-time coefficient estimation block with a discrete-time shot.
The parity block is exactly the same as before. 
The only change is the choice of evolution time.
Fix a target Pauli basis $c$ and sets $A,B\subseteq[n]$. We sample an integer $M\in\{0,1,2,\ldots\}$ by
\begin{align}
\Pr[M=0]&\coloneqq 1-\frac{2t_0\sum_{m\ge1}|L_\Lambda(mt_0)|}{2\ell_0^{\rm grid}\Lambda},\label{eq:grid-lin-null}\\
\Pr[M=m]&\coloneqq \frac{2t_0|L_\Lambda(mt_0)|}{2\ell_0^{\rm grid}\Lambda},\qquad m\ge1.\label{eq:grid-lin-positive}
\end{align}
The outcome $M=0$ is null. 
In that case, we output $0$ and do not evolve the system. 
If $M=m\ge1$, we run the same parity block experiment
as in Proposition~\ref{prop:block-trace} but at time $mt_0$.

Let $Z_{A,B,m}^{(c),{\rm grid}}$ be the parity variable produced by this block. 
We put back the sign of $L_\Lambda(mt_0)$ and define
\begin{align}
G_{A,B}^{(c),{\rm grid}}\coloneqq 
\begin{cases}
2\ell_0^{\rm grid}\Lambda\sign(L_\Lambda(mt_0))Z_{A,B,m}^{(c),{\rm grid}},& M=m\ge1,\\
0,& M=0.
\end{cases}
\end{align}
This is the discrete version of the continuous linear sample. 
The factor $2\ell_0^{\rm grid}\Lambda$ only cancels the normalization in Eq.~\eqref{eq:grid-lin-positive}.

We will also use the same visibility correction as before. 
For a label $u\neq0$ and a target Pauli basis $c$, define
\begin{align}
Y_u^{(c),{\rm grid}}\coloneqq 
\begin{cases}
\dfrac12\sigma_c(A_c(u),B_c(u))G_{A_c(u),B_c(u)}^{(c),{\rm grid}},& u\in\mathcal V(c),\\
0,& u\notin\mathcal V(c).
\end{cases}
\end{align}
If $C\sim\mathrm{Unif}\{X,Y,Z\}^n$ and $q(u)\coloneqq \Pr[u\in\mathcal V(C)]$, then the discrete sample after applying visibility correction is
\begin{align}
\widetilde Y_u^{\rm grid}\coloneqq q(u)^{-1}Y_u^{(C),{\rm grid}}.
\end{align}

\begin{proposition}[Discrete linear sample]\label{prop:lattice-linear}
Assume Eq.~\eqref{eq:lattice-condition}. For the grid shot defined above,
\begin{align}\label{eq:grid-G-properties}
\E\left[G_{A,B}^{(c),{\rm grid}}\right]=\bigl(F_{A,B}^{(c)}\bigr)'(0),\qquad
\left|G_{A,B}^{(c),{\rm grid}}\right|\le2\ell_0^{\rm grid}\Lambda,\qquad
\E[T_{\rm lin}^{\rm grid}]\le\frac{\ell_1^{\rm grid}}{2\Lambda}.
\end{align}
Moreover, for every nonzero label $u$, we have
\begin{align}
\E[\widetilde Y_u^{\rm grid}]=\lambda_u,\qquad
|\widetilde Y_u^{\rm grid}|\le 3\ell_0^{\rm grid}\Lambda.
\label{eq:grid-ht-properties}
\end{align}
\end{proposition}

\begin{proof}
The distribution of $M$ is valid by Eq.~\eqref{eq:grid-L-mass}. 
For $m\ge1$, Proposition~\ref{prop:block-trace} gives
\begin{align}
\E\left[Z_{A,B,m}^{(c),{\rm grid}}\mid M=m\right]=\frac12 g_{A,B}^{(c)}(mt_0).
\end{align}
Therefore, we have
\begin{align}
\begin{split}
\E\left[G_{A,B}^{(c),{\rm grid}}\right]&=\sum_{m\ge1}\frac{2t_0|L_\Lambda(mt_0)|}{2\ell_0^{\rm grid}\Lambda}\cdot2\ell_0^{\rm grid}\Lambda\sign(L_\Lambda(mt_0))\cdot\frac12 g_{A,B}^{(c)}(mt_0)\\
&=t_0\sum_{m\ge1}L_\Lambda(mt_0)g_{A,B}^{(c)}(mt_0)\\
&=\bigl(F_{A,B}^{(c)}\bigr)'(0),
\end{split}
\end{align}
where the last step is Eq.~\eqref{eq:lattice-linear-identity}. 
The bound on $G_{A,B}^{(c),{\rm grid}}$ follows from $|Z_{A,B,m}^{(c),{\rm grid}}|\le1$.

The mean time is also immediate from the definition of $M$ as
\begin{align}
\E[T_{\rm lin}^{\rm grid}]=\sum_{m\ge1}\frac{2t_0|L_\Lambda(mt_0)|}{2\ell_0^{\rm grid}\Lambda}(mt_0)\le\frac{\ell_1^{\rm grid}}{2\Lambda},
\end{align}
where we used Eq.~\eqref{eq:grid-L-time}.

It remains to check the visibility correction. 
If $u\in\mathcal V(c)$, the derivative calculation in Proposition~\ref{prop:parity-main} gives
\begin{align}
\bigl(F_{A_c(u),B_c(u)}^{(c)}\bigr)'(0)=2\sigma_c(A_c(u),B_c(u))\lambda_u.
\end{align}
Thus,
\begin{align}
\E\left[Y_u^{(c),{\rm grid}}\right]=\lambda_u,
\qquad
\left|Y_u^{(c),{\rm grid}}\right|\le\ell_0^{\rm grid}\Lambda
\end{align}
whenever $u\in\mathcal V(c)$. 
If $u\notin\mathcal V(c)$, then $Y_u^{(c),{\rm grid}}=0$ by definition.

Finally, Proposition~\ref{prop:coverage} gives $q(u)\ge1/3$. 
Hence, we have
\begin{align}
\begin{split}
\E[\widetilde Y_u^{\rm grid}]&=q(u)^{-1}\Pr[u\in\mathcal V(C)]\lambda_u=\lambda_u,\\
|\widetilde Y_u^{\rm grid}|&\le q(u)^{-1}\ell_0^{\rm grid}\Lambda\le3\ell_0^{\rm grid}\Lambda.
\end{split}
\end{align}
\end{proof}

The grid distributions above have infinite support, but their mean time is finite. 
If we want a deterministic limit on the duration of each shot, fix $R>0$ and set
\begin{align}
M_R\coloneqq \left\lfloor\frac{R}{2\Lambda t_0}\right\rfloor .
\end{align}
as the new cutoff threshold. Keep only outcomes $m\le M_R$ and move the remaining probability to the null outcome. 
Then every executed shot has time at most
\begin{align}
M_Rt_0\le \frac{R}{2\Lambda}.
\end{align}
The only new loss is the tail:
\begin{align}
2t_0\sum_{m>M_R}|K_\Lambda(mt_0)|
\quad\text{or}\quad
2t_0\sum_{m>M_R}|L_\Lambda(mt_0)|.
\end{align}
Using Eq.~\eqref{eq:K-grid-scaling}, Eq.~\eqref{eq:L-grid-scaling}, and $\Delta=2\Lambda t_0$, these tails are bounded by
\begin{align}
C_N\Lambda^2(1+R)^{-N}
\quad\text{and}\quad
C_N\Lambda(1+R)^{-N}
\end{align}
Therefore, $R$ controls the tradeoff.
A larger $R$ allows a longer maximum shot time $R/(2\Lambda)$, but discards less probability mass. 
A smaller $R$ gives a shorter maximum shot time, but discards more of the tail.

\subsection{A minimal resolution time lower bound in the grid setting}
The upper bound above assumes that the time grid spacing is smaller than a constant times $1/\Lambda$. 
We now show that this scale is necessary in the grid setting. 
\begin{proposition}[A universal hard family]\label{prop:nyquist}
Fix $t_0>0$. 
Suppose every allowed shot time belongs to $\mathcal T_{t_0}$ and the admissible Hamiltonian class contains the one-sparse Hamiltonians
\begin{align}
H_u\coloneqq \frac{\pi}{t_0}P_u,
\qquad u\in V\setminus\{0\}.
\end{align}
Then no protocol can identify the support label $u$ with success probability better than random guessing, even with arbitrarily many shots and arbitrary adaptivity.
\end{proposition}
\begin{proof}
For every integer $m\ge0$,
\begin{align}
e^{-iH_u(mt_0)}=e^{-im\pi P_u}=\cos(m\pi)I-i\sin(m\pi)P_u=(-1)^mI.
\end{align}
Therefore, every allowed experiment has exactly the same outcome distribution for every $u$. 
So the support label can not be identified.
\end{proof}
The proposition uses Hamiltonians of norm $\pi/t_0$. 
Therefore, simply requiring the norm bound to be greater than $\pi/t_0$ is not enough. 
We will also give a necessary resolution condition depending on the norm bound.

\begin{corollary}[The upper bound represented by norms]\label{cor:anti-aliasing}
If exact learning is required uniformly over all sparse Pauli Hamiltonians with $\|H\|\le\Lambda$ in the discrete-time model, then it is necessary that
\begin{align}
t_0<\frac{\pi}{\Lambda}.
\end{align}
\end{corollary}
\begin{proof}
If $t_0\ge\pi/\Lambda$, then $\pi/t_0\le\Lambda$, so the one-sparse Hamiltonians in Proposition~\ref{prop:nyquist} are not learnable. 
Exact uniform support learning is therefore impossible.
\end{proof}

\section{Lower Bound}\label{app:lower}

\subsection{Information theoretic tools}
The proof of the lower bound consists of two parts. First, in a control-free protocol, even estimating a very small coefficient already reaches the standard quantum limit in time complexity. 
Second, learning a large number of hidden support locations leads to a coupon collection problem, thus adding a log term to the final lower bound. 
We begin with two fundamental information-theoretic inequalities as facts from quantum Fisher information for single-parameter estimations (see, e.g., Refs.~\cite{wootters1981statistical,braunstein1994statistical,braunstein1996generalized}).

\begin{fact}[Total variation from Fisher information, see, e.g. ]\label{fact:tv-fi}
Let $P_\theta$ be the distribution of the full history of a protocol for a parametrization estimation problem on a single parameter $\theta$, and let $I(\theta)$ be its classical Fisher information. 
Then for any $a<b$,
\begin{align}
\TV(P_a,P_b) \le \frac12\int_a^b \sqrt{I(\xi)}d\xi .
\end{align}
\end{fact}

\begin{fact}[One-shot Fisher information bound]\label{fact:bc-bound}
Let $\rho_\theta(t)=e^{-iH_\theta t}\rho_0e^{iH_\theta t}$ be followed by any measurement and classical postprocessing. 
Then the classical Fisher information is bounded by
\begin{align}
I_{\mathrm{cl}}(\theta)\le F_Q(\rho_\theta(t)),
\end{align}
where $F_Q$ is the quantum Fisher information. 
If $H_\theta$ is differentiable and
\begin{align}
\mathcal G_\theta(t)\coloneqq \int_0^t e^{iH_\theta s}(\partial_\theta H_\theta)e^{-iH_\theta s}ds,
\end{align}
then
\begin{align}
F_Q(\rho_\theta(t))\le 4\|\mathcal G_\theta(t)\|^2.
\end{align}
\end{fact}

\subsection{Proof of Lemma~\ref{lem:sql-lower}}
We first prove the lower bound for one coefficient, which is Lemma~\ref{lem:sql-lower}, and we restate it as follows. 

\begin{lemma}[Lower bound for single-parameter estimation]\label{lem:sql-lower-restate}
There is a two-sparse one-qubit family
\begin{align}
H_\theta=\Lambda(\cos\theta Z+\sin\theta X),
\qquad
|\theta|\le\frac{\pi}{6},
\end{align}
such that the following holds. 
For every $0<\epsilon\le\Lambda/4$, any control-free protocol with deterministic total evolution time $T_{\rm tot}$ that estimates the $X$-coefficient to accuracy $\epsilon$ with success probability at least $2/3$ for all $|\theta|\le\pi/6$ must satisfy
\begin{align}
T_{\rm tot}=\Omega\left(\frac{\Lambda}{\epsilon^2}\right).
\end{align}
\end{lemma}

\begin{proof}
At $\theta=0$, we have $H_0=\Lambda Z$ and $\partial_\theta H_\theta\big|_{\theta=0}=\Lambda X$.
For one shot of duration $t$,
\begin{align}
\mathcal G_0(t)=\int_0^t e^{i\Lambda Zs}(\Lambda X)e^{-i\Lambda Zs}ds=\frac{\sin(2\Lambda t)}{2}X-\frac{1-\cos(2\Lambda t)}{2}Y.
\end{align}
Hence, we have $\|\mathcal G_0(t)\|^2=\sin^2(\Lambda t)\le\Lambda t$.
The same bound holds for every $\theta$, since the family
$H_\theta$ is just a rotation of $H_0$.

Now consider the full transcript of an adaptive protocol. Fisher information adds over shots, after conditioning on the previous transcript.
Using Fact~\ref{fact:bc-bound} on each shot, $I(\theta)\le 4\Lambda T_{\rm tot}$ for every $\theta$.
Choose $\theta_\epsilon\in(0,\pi/6]$ by $\Lambda\sin\theta_\epsilon=2\epsilon$.
At $+\theta_\epsilon$, the $X$-coefficient is $+2\epsilon$, while at $-\theta_\epsilon$, it is $-2\epsilon$. 
Therefore, any $\epsilon$-accurate estimator can distinguish these two cases by the sign of its output. 
Since we require the success probability to be at least $2/3$, we need
\begin{align}
\TV(P_{+\theta_\epsilon},P_{-\theta_\epsilon})\ge\frac13 .
\end{align}
On the other hand, Fact~\ref{fact:tv-fi} gives
\begin{align}
\TV(P_{+\theta_\epsilon},P_{-\theta_\epsilon})\le\frac12\int_{-\theta_\epsilon}^{\theta_\epsilon}\sqrt{4\Lambda T_{\rm tot}}d\xi=2\theta_\epsilon\sqrt{\Lambda T_{\rm tot}}.
\end{align}
Since $\sin x\ge x/2$ on $[0,\pi/6]$, we have $\theta_\epsilon\le2\sin\theta_\epsilon=\frac{4\epsilon}{\Lambda}$.
Combining the last three displays gives $\frac13\le\frac{8\epsilon}{\Lambda}\sqrt{\Lambda T_{\rm tot}}$.
This rearranges to
\begin{align}
T_{\rm tot}=\Omega\left(\frac{\Lambda}{\epsilon^2}\right).
\end{align}
\end{proof}

\subsection{Hard family}
Next, we will prove that the log factor is also essential. 
The idea is to hide one label in each of $M$ blocks. 
To recover the support, the learner must find all $M$
hidden labels.

The next lemmas make this intuition precise for adaptive protocols. 
After averaging over the hidden signs, one shot can be represented as follows: it either gives no information about $x$, or it points to one uniformly random block. 
Therefore, even an adaptive learner still has to see all $M$ blocks, which leads to the usual coupon-collector cost of order $M\log M$.

We now define the hard family. Let $M,R\ge2$, and set $r\coloneqq \lceil\log_2 R\rceil$.
Use $M$ marker qubits $m_1,\ldots,m_M$ and $r$ shared label qubits.
Choose $R$ distinct commuting Hermitian Pauli strings $P_1,\ldots,P_R$ on the label register. 
For example, one may take distinct $Z$-type strings.
For $a\in[M]$ and $b\in[R]$, define
\begin{align}
Q_{a,b}\coloneqq \left(\prod_{c<a}Z_{m_c}\right)X_{m_a}P_b.
\end{align}
Let $q_{a,b}\in V$ be the Pauli label of $Q_{a,b}$.
If $a<a'$, then $Q_{a,b}$ and $Q_{a',b'}$ anticommute on marker qubit $m_a$. 
All other tensor factors commute. 
Therefore, any two strings from different blocks anticommute.
And $M\leq 2n+1$ by Proposition 9 of Ref.~\cite{hrubevs2016families}.
For $x=(x_1,\ldots,x_M)\in[R]^M$ and $\sigma=(\sigma_1,\ldots,\sigma_M)\in\{0,1\}^M$, we define
\begin{align}\label{eq:hard-hamiltonian}
H_{x,\sigma}\coloneqq \frac{\Lambda}{\sqrt M}\sum_{a=1}^M(-1)^{\sigma_a}Q_{a,x_a}.
\end{align}
The support label set is $S_x\coloneqq \{q_{a,x_a}:a\in[M]\}$.
The selected Pauli strings pairwise anticommute, so $H_{x,\sigma}^2=\Lambda^2I$ and $\|H_{x,\sigma}\|=\Lambda$.
The signs $\sigma_a$ are included to hide the block signs. 
They do not change the support $S_x$. 
For $g\in\{0,1\}^M$, define
\begin{align}
R_g\coloneqq \prod_{a=1}^M Z_{m_a}^{g_a}.
\end{align}
Then
\begin{align}
R_gQ_{a,b}R_g^\dagger=(-1)^{g_a}Q_{a,b},\qquad
R_gH_{x,\sigma}R_g^\dagger=H_{x,\sigma+g},
\end{align}
where addition is over $\mathbb F_2$.

The signs $\sigma_a$ do not change the support $S_x$. 
They only choose whether the active Pauli string in block $a$ has a plus sign or a minus sign. 
The identity above says that changing $\sigma$ to $\sigma+g$ is the same as conjugating the whole experiment by $R_g$.

We use this to remove only the irrelevant sign information from the proof. 
By observed data, we mean the list of all settings chosen so far and all outcomes observed so far. 
If applying the same $R_g$ to every shot turns one observed data set into another, we treat these two data sets as the same for the proof. 
We call this the sign-erased observed data. 
It forgets only the hidden plus/minus convention. 
It does not forget the hidden labels $x_a$ that determine the support.

\subsection{Proof of Theorem~\ref{thm:unified-log}}

\begin{lemma}[Averaging over hidden signs]\label{lem:gauge-sym}
Under the prior $x\sim{\rm Unif}([R]^M)$ and $\sigma\sim{\rm Unif}(\{0,1\}^M)$, every control-free protocol can be replaced by another control-free protocol with the same average success probability for recovering $x$, such that its next shot and final answer depend only on the observed data modulo the common conjugation by $R_g$.
\end{lemma}

\begin{proof}
Let $\Pi$ be any control-free protocol. 
For $g\in\{0,1\}^M$, let $\Pi^g$ be the protocol obtained from $\Pi$ by conjugating every prepared state and every measured observable by $R_g$.
Let $\mathrm{Succ}(\Pi)$ denote the average success probability of $\Pi$ under the product prior on $(x,\sigma)$. 
Since $R_gH_{x,\sigma}R_g^\dagger=H_{x,\sigma+g}$, we have
\begin{align}
\begin{split}
\mathrm{Succ}(\Pi^g)
&=\frac{1}{R^M2^M}\sum_{x\in[R]^M}\sum_{\sigma\in\{0,1\}^M}\Pr_{\Pi^g,x,\sigma}[\widehat x=x]  \\
&=\frac{1}{R^M2^M}\sum_{x\in[R]^M}\sum_{\sigma\in\{0,1\}^M}\Pr_{\Pi,x,\sigma+g}[\widehat x=x]  \\
&=\frac{1}{R^M2^M}\sum_{x\in[R]^M}\sum_{\sigma'\in\{0,1\}^M}\Pr_{\Pi,x,\sigma'}[\widehat x=x]  \\
&=\mathrm{Succ}(\Pi).
\end{split}
\end{align}
Here we used that $\sigma\mapsto\sigma+g$ is a bijection of $\{0,1\}^M$.
Now define the averaged protocol
\begin{align}
\overline\Pi:=2^{-M}\sum_{g\in\{0,1\}^M}\Pi^g .
\end{align}
Equivalently, $\overline\Pi$ first samples $g$ uniformly and then runs $\Pi^g$. 
By the calculation above, we have $\mathrm{Succ}(\overline\Pi)=\mathrm{Succ}(\Pi)$.

It remains to record the symmetry of $\overline\Pi$. 
Let $D$ denote the observed data after some number of shots, and let $gD$ denote the data obtained by applying the same conjugation $R_g$ to all shots. 
By the construction of the average,
\begin{align}
\Pr_{\overline\Pi}[\text{next shot}=A\mid D]=\Pr_{\overline\Pi}[\text{next shot}=gA\mid gD].
\label{eq:gauge-invariant-step}
\end{align}
The same identity holds for the final answer rule.
Therefore, the protocol can be described using only the equivalence class $[D]:=\{gD:g\in\{0,1\}^M\}$.
Indeed, we can choose one representative $D^\circ$ from each class $[D]$. 
If the actual data is $D=gD^\circ$, we sample the next shot for $D^\circ$ and then conjugate that shot by $R_g$. 
Eq.~\eqref{eq:gauge-invariant-step} shows that this gives the same law as $\overline\Pi$. 
The same construction applies to the final answer. 
Thus, the new protocol has the same average success probability and uses only the observed data modulo the common conjugation by $R_g$.
\end{proof}

We now analyze one shot after the common sign flip has been erased. 
The important observation is that the hidden signs are still uniform after conditioning on the sign-erased observed data. 
So the next shot has a simple form: it either gives no information about $x$, or it points to one uniformly random block.

The learner does not observe the index $J$ below. 
It is only a way for us to write the distribution of one shot. 
Later, we can reveal all $J_t$'s to the learner for free, which can only make the learner stronger.

\begin{lemma}[One shot sees at most one block]\label{lem:orbit-mixture}
Consider the protocol from Lemma~\ref{lem:gauge-sym}. 
Condition on the past sign-erased observed data and on the canonical representative of the next chosen shot, and fix $x$.
Suppose this shot prepares $\rho$, evolves for time $\tau$, and measures $\{M_y\}_y$. 
Define
\begin{align}
c_\tau\coloneqq \cos(\Lambda\tau),\qquad
s_\tau\coloneqq \sin(\Lambda\tau).
\end{align}
Then the conditional distribution of the next outcome can be written as
\begin{align}
p_x(y)=c_\tau^2 R(y)+\frac{s_\tau^2}{M}\sum_{a=1}^M Q_{a,x_a}(y),
\end{align}
where $R(y)\coloneqq \Tr(M_y\rho)$ and $Q_{a,b}(y)\coloneqq \Tr(M_yQ_{a,b}\rho Q_{a,b})$.
Equivalently, we may first draw a hidden index $J$ with $\Pr[J=0]=c_\tau^2$ and $\Pr[J=a]=\frac{s_\tau^2}{M}$ for $a\in[M]$.
If $J=0$, the outcome is drawn from $R$, which does not depend on $x$. 
If $J=a\neq0$, the outcome distribution depends on $x$ only through $x_a$.
\end{lemma}

\begin{proof}
First, we check why the signs are still uniform. 
Fix two sign strings $s$ and $s'$. 
Let $g=s+s'$. 
Changing $\sigma=s$ to $\sigma=s'$ is the same as applying the sign flip $R_g$ to all observed data. 
But this common sign flip is erased in the sign-erased observed data. 
Therefore, $s$ and $s'$ give the same probability to the same sign-erased observed data. 
Since the prior on $\sigma$ is uniform, the conditional distribution of $\sigma$ is still uniform. 
The canonical representative of the next shot is chosen from the same sign-erased observed data, so conditioning on it does not change this.

Since $H_{x,\sigma}^2=\Lambda^2I$, we have
\begin{align}
e^{-iH_{x,\sigma}\tau}=c_\tau I-i\frac{s_\tau}{\sqrt M}\sum_{a=1}^M (-1)^{\sigma_a}Q_{a,x_a}.
\end{align}
Now average the next-outcome probability over the uniform signs $\sigma$. 
Every term with one sign factor $(-1)^{\sigma_a}$ averages to zero. 
Every cross term with two different sign factors $(-1)^{\sigma_a+\sigma_{a'}}$, where $a\neq a'$, also averages to zero. 
The only terms left are
\begin{align}
c_\tau^2\Tr(M_y\rho)+\frac{s_\tau^2}{M}\sum_{a=1}^M\Tr(M_yQ_{a,x_a}\rho Q_{a,x_a}).
\end{align}
This is the displayed formula. 
Since $R$ and each $Q_{a,b}$ are probability distributions over $y$, the same formula can be generated by first drawing $J$ and then drawing the outcome from the corresponding distribution.
\end{proof}
For the adaptive argument, we track the history after forgetting the hidden signs. 
At round $t$, the protocol's next shot is determined by this sign-erased history and its private randomness. 
We write this shot using a fixed canonical representative of its orbit under the common conjugation $R_g$. 
Let $\mathcal F_t$ collect this sign-erased history up to time $t$, the private randomness, and the auxiliary labels $J_1,\ldots,J_t$.

Conditioned on $\mathcal F_{t-1}$, on the chosen canonical shot 
$(\rho_t,\tau_t,\{M_{t,y}\}_y)$, and on $x$, the hidden signs remain uniform. 
Hence the next outcome can be coupled with a latent variable $J_t$ such that
\begin{align}
\Pr[J_t=0\mid\mathcal F_{t-1}]
=\cos^2(\Lambda\tau_t),
\qquad
\Pr[J_t=a\mid\mathcal F_{t-1}]
=\frac{\sin^2(\Lambda\tau_t)}{M},
\quad a\in[M].
\end{align}
If $J_t=0$, the conditional outcome distribution is independent of $x$. 
If $J_t=a$, the conditional outcome distribution depends on $x$ only through $x_a$.

Now reveal the hidden indices $J_t$ to the learner after all shots. 
This gives the learner extra information. 
So a lower bound for this stronger learner is also a lower bound for the original learner.

Let $U$ be the number of blocks that are never selected by an informative hidden index
\begin{align}
U\coloneqq\left|\{a\in[M]:J_t\neq a\text{ for every shot }t\}\right|.
\end{align}
Thus, $U$ is the number of unseen blocks.

\begin{lemma}[Unseen blocks must be guessed]\label{lem:unseen-block}
Consider the stronger experiment where the learner also receives all hidden indices $J_t$. 
Let $P_{\rm avg}$ be the average probability of exactly recovering $x$ under the uniform prior. Then
\begin{align}
P_{\rm avg}\le\E[R^{-U}]\le\Pr[U=0]+R^{-1}.
\end{align}
\end{lemma}

\begin{proof}
Inducting over the shots, once the labels $J_1,\ldots,J_t$ are fixed, the sign-erased transcript up to time $t$ can depend on $x$ only through the coordinates that have been hit, namely those $x_a$ with $J_s=a$ for some $s\le t$. 
Indeed, the next shot is chosen from the previous sign-erased history. 
If $J_t=0$, the new outcome carries no information about $x$; if $J_t=a$, it can depend only on $x_a$.
Since $x_a$ was uniform on $[R]$, it remains uniform after conditioning on the revealed sign-erased data and the labels $J_t$.  
Thus, if $U$ blocks are unseen, the probability of guessing all their labels correctly is at most $R^{-U}$. 
Averaging gives
\begin{align}
P_{\rm avg}\le \E[R^{-U}].
\end{align}
Also,
\begin{align}
R^{-U}\le\mathbbm 1[U=0]+R^{-1}\mathbbm 1[U>0],
\end{align}
which gives the second inequality.
\end{proof}

The next lemma is the coupon-collector estimate, which indicates that fewer than a constant times $M\log M$ informative selections are not enough to see all $M$ blocks with good probability.

\begin{lemma}[Coupon collector]\label{lem:coupon}
Let $K\coloneqq|\{t:J_t\neq0\}|$ be the number of informative selections, and set $k_0\coloneqq \lfloor\frac14M\log M\rfloor$.
Then
\begin{align}
\Pr[U=0,\ K<k_0]\le M^{-1/2}.
\end{align}
\end{lemma}

\begin{proof}
At any shot, after conditioning on all past revealed data and on the chosen shot, Lemma~\ref{lem:orbit-mixture} gives
\begin{align}
\Pr[J_t=a\mid J_t\neq0,\text{ past data and chosen shot}]=\frac1M,\qquad a\in[M].
\end{align}
Equivalently, draw an infinite list $A_1,A_2,\ldots$ of independent uniform labels in $[M]$ in advance. 
Each time an informative shot occurs, use the next unused label in this list. 
This is valid because, even after conditioning on the whole past and on the fact that the shot is informative, Lemma~\ref{lem:orbit-mixture} says that the revealed block label is still uniform on $[M]$. 
Thus, adaptivity may decide when informative shots occur, but not which block labels they reveal.
If the real process has fewer than $k_0$ informative selections, we may append extra independent uniform draws. 
Appending extra draws can only help to cover more blocks.

If $U=0$ and $K<k_0$, then all $M$ blocks have already appeared among the first fewer than $k_0$ informative draws. 
Therefore, the first $k_0$ independent uniform draws must also contain all $M$ blocks. 
Hence,
\begin{align}
\Pr[U=0,\ K<k_0]\le\Pr[\{A_1,\ldots,A_{k_0}\}=[M]].
\end{align}

Let $U_{k_0}$ be the number of blocks missed by the first $k_0$ uniform draws. 
Then
\begin{align}
\E[U_{k_0}]=M\left(1-\frac1M\right)^{k_0}.
\end{align}
For $k_0\le M\log M/4$ and $M\ge16$, this expectation is at least $\sqrt M$. 
A standard second-moment bound for coupon collection gives $\Var(U_{k_0})\le \E[U_{k_0}]$.
Therefore, by Chebyshev's inequality,
\begin{align}
\Pr[U_{k_0}=0]\le\Pr\left[|U_{k_0}-\E [U_{k_0}]|\ge \E [U_{k_0}]\right]\le\frac{\Var(U_{k_0})}{(\E [U_{k_0}])^2}\le M^{-1/2}.
\end{align}
This proves the claim.
\end{proof}

We now put the above pieces together. 
We choose $M$ so that each active coefficient has size about $\epsilon$. 
Then support recovery is the same as recovering the hidden vector $x$.
Now, we prove Theorem~\ref{thm:unified-log}, which is restated as follows.

\begin{theorem}[Lower bound for $\epsilon$-detectable support]\label{thm:unified-log-support}
There is an absolute constant $c>0$ such that the following holds. 
Let $0<\epsilon\le\Lambda/16$, and assume the number of qubits is large enough to embed the hard family in Eq.~\eqref{eq:hard-hamiltonian}, for example $n\ge C_0\frac{\Lambda^2}{\epsilon^2}$ for a sufficiently large absolute constant $C_0$. 
Then any control-free protocol that solves Problem~\ref{prob:sparse_ham_recover} with success probability at least $2/3$ must have total evolution time
\begin{align}
T_{\rm tot} \ge c\frac{\Lambda}{\epsilon^2}
\log\frac{\Lambda}{\epsilon}
\end{align}
where one may take $c=2^{-11}$.
\end{theorem}

\begin{proof}
We set $M\coloneqq\lfloor(\tfrac{\Lambda}{4\epsilon})^2\rfloor$ and $R\coloneqq M$.
Since $0<\epsilon\le\Lambda/16$, we have $M\ge16$ and $\frac{\Lambda^2}{32\epsilon^2}\le M \le\frac{\Lambda^2}{16\epsilon^2}$.
The family uses $M+\lceil\log_2M\rceil\le2M$ qubits. 

For every Hamiltonian in the hard family, each active coefficient has a magnitude $\frac{\Lambda}{\sqrt M}\ge4\epsilon$.
All inactive coefficients are zero. 
Hence, any valid output must be exactly $\widehat S_\epsilon=S_x$.
Since $x\mapsto S_x$ is one-to-one, support recovery is the same as recovering $x$.

Put the uniform prior on $x$ and on the signs $\sigma$. 
A worst-case success probability of at least $2/3$ implies an average success probability of at least $2/3$. 
By Lemma~\ref{lem:gauge-sym}, we may use the sign-erased protocol with the same average success probability. 

Now reveal the hidden indices $J_t$ to the learner. 
This can only make the learner stronger, so its average success probability is still at least $2/3$. 
Lemma~\ref{lem:unseen-block} gives
\begin{align}
\frac23\le P_{\rm avg}\le \Pr[U=0]+R^{-1}.
\end{align}
Let $K\coloneqq|\{t:J_t\neq0\}|$ and $k_0\coloneqq\lfloor\frac14M\log M\rfloor$.
Using Lemma~\ref{lem:coupon} and $R=M$,
\begin{align}
\begin{split}
P_{\rm avg}&\le\Pr[K\ge k_0]+\Pr[U=0,\ K<k_0]+M^{-1}\\
&\le\Pr[K\ge k_0]+M^{-1/2}+M^{-1}.
\end{split}
\end{align}
Since $M\ge16$, we have $\Pr[K\ge k_0]\ge \frac16$.
Therefore, we have $\E[K]\ge \frac{k_0}{6}$.

Now we upper bound $\E[K]$ by the total evolution time. 
If the $t$-th shot uses time $\tau_t$, Lemma~\ref{lem:orbit-mixture} gives
\begin{align}
\Pr[J_t\neq0\mid \text{past data and chosen shot}]=\sin^2(\Lambda\tau_t)\le\Lambda\tau_t.
\end{align}
Taking expectation shot by shot,
\begin{align}
\E[K]=\E\left[\sum_t \mathbbm 1[J_t\neq0]\right]\le\E\left[\sum_t \Lambda\tau_t\right].
\end{align}
Since the total evolution time is always at most $T_{\rm tot}$, we have $\E[K]\le \Lambda T_{\rm tot}$.
Combining this with $\E[K]\ge k_0/6$ gives $T_{\rm tot}\ge\frac{k_0}{6\Lambda}$.
For $M\ge16$, we have $k_0\ge \frac18M\log M$.
Hence, we deduce that $T_{\rm tot}\ge\frac{M\log M}{48\Lambda}$.
Finally, note that $M\ge \frac{\Lambda^2}{32\epsilon^2}$ and $\log M\ge\frac34\log\frac{\Lambda}{\epsilon}$, where the second inequality uses $\epsilon\le\Lambda/16$. 
Therefore
\begin{align}
T_{\rm tot}\ge\Omega\left(\frac{\Lambda}{\epsilon^2}\log\frac{\Lambda}{\epsilon}\right).
\end{align}
\end{proof}

\end{document}